\documentclass[11pt]{article}
\usepackage[margin=1in]{geometry}
\usepackage{graphicx}
\usepackage{amsmath}
\usepackage{amsthm}
\usepackage{physics}
\usepackage{amssymb}
\usepackage{xcolor}
\usepackage{mathtools}
\usepackage{comment}
\usepackage{subcaption}
\usepackage{enumitem}
\usepackage[normalem]{ulem}
\usepackage{tabularx}
\usepackage{booktabs}
\usepackage{rotating}
\usepackage{hyperref}
\usepackage{longtable}
\usepackage[user]{zref}
\usepackage{array}
\usepackage{ragged2e}
\usepackage{authblk}

\hypersetup{
    colorlinks,
    linkcolor={red!50!black},
    citecolor={blue!50!black},
    urlcolor={blue!80!black}
}

\theoremstyle{plain}
\newtheorem*{theorem*}{Theorem}
\newtheorem{theorem}{Theorem}
\newtheorem{lemma}[theorem]{Lemma}
\newtheorem{proposition}[theorem]{Proposition}

\newtheorem*{problem}{Open problem}

\theoremstyle{definition}

\theoremstyle{remark}

\newtheorem*{claim}{Claim}
\newtheorem{remark}[theorem]{Remark}

\usepackage{dsfont}
\newcommand{\indicator}[1]{\mathds{1}_{#1}}
\DeclarePairedDelimiter{\set}{\{}{\}}
\DeclarePairedDelimiter{\mset}{\langle}{\rangle}
\DeclarePairedDelimiter{\cbraces}{\{}{\}}
\DeclarePairedDelimiter{\brackets}{[}{]}
\DeclarePairedDelimiter{\bars}{|}{|}
\DeclarePairedDelimiter{\ceil}{\lceil}{\rceil}
\DeclarePairedDelimiter{\floor}{\lfloor}{\rfloor}

\newcommand{\R}{\mathbb{R}}
\newcommand{\Z}{\mathbb{Z}}
\newcommand{\bA}{\mathbb{A}}
\newcommand{\bL}{\mathbb{L}}
\newcommand{\bq}{\mathbf{q}}
\newcommand{\bz}{\mathbf{z}}
\newcommand{\cB}{\mathcal{B}}
\newcommand{\cC}{\mathcal{C}}
\newcommand{\cD}{\mathcal{D}}

\newcommand{\cM}{\mathcal{M}}
\newcommand{\cT}{\mathcal{T}}
\newcommand{\cU}{\mathcal{U}}
\newcommand{\cX}{\mathcal{X}}
\newcommand{\cY}{\mathcal{Y}}

\newcommand{\fd}{\mathfrak{d}}
\newcommand{\fk}{\mathfrak{k}}
\newcommand{\fl}{\mathfrak{l}}
\newcommand{\fm}{\mathfrak{m}}
\newcommand{\fn}{\mathfrak{n}}
\newcommand{\fR}{\mathfrak{R}}
\newcommand{\fr}{\mathfrak{r}}

\newcommand{\fT}{\mathfrak{T}}
\newcommand{\fw}{\mathfrak{w}}
\newcommand{\rmd}{\mathrm{d}}
\newcommand{\eps}{\epsilon}
\newcommand{\vareps}{\varepsilon}
\newcommand{\ext}{\mathrm{ext}}
\newcommand{\intr}{\mathrm{int}}
\newcommand{\supp}[1]{\mathrm{supp}(#1)}
\newcommand{\bdin}[2][]{\partial^{\mathrm{in}}_{#1}#2}
\newcommand{\bdex}[2][]{\partial^{\mathrm{ex}}_{#1}#2}

\makeatletter

\newcounter{prefactorparam}
\newcounter{peierlsparam}
\newcounter{spatialparam}
\newcounter{lossparam}
\newcounter{freeparam}

\newcommand{\declareprefactor}[1]{
  \refstepcounter{prefactorparam}
  \label{param:prefactor:#1}
}
\newcommand{\declarePeierls}[1]{
  \refstepcounter{peierlsparam}
  \label{param:Peierls:#1}
}
\newcommand{\declarespatial}[1]{
  \refstepcounter{spatialparam}
  \label{param:spatial:#1}
}
\newcommand{\declareloss}[1]{
  \refstepcounter{lossparam}
  \label{param:loss:#1}
}
\newcommand{\declarefree}[1]{
  \refstepcounter{freeparam}
  \label{param:free:#1}
}

\newcommand{\prefactor}[1]{\phi_{\ref{param:prefactor:#1}}}
\newcommand{\Peierls}[1]{\kappa_{\ref{param:Peierls:#1}}}
\newcommand{\spatial}[1]{\lambda_{\ref{param:spatial:#1}}}
\newcommand{\loss}[1]{\eta_{\ref{param:loss:#1}}}
\newcommand{\free}[1]{\theta_{\ref{param:free:#1}}}

\makeatother

\newcounter{condrow}
\renewcommand{\thecondrow}{C\arabic{condrow}}
\newcommand{\condrowlabel}[1]{
  \leavevmode
  \refstepcounter{condrow}%
  \label{cond:#1}%
  \text{\thecondrow}
}

\let\emptyset\varnothing

\allowdisplaybreaks

\title{Rigorous bound on the aspect ratio for the formation of a nematic phase in hard rod and hard rectangle systems on $\Z^2$}
\author{Qidong He\thanks{Email address: \url{qh97@scarletmail.rutgers.edu}}}
\affil{Department of Mathematics, Rutgers University}
\date{}

\begin{document}

\maketitle

\begin{abstract}
    We prove the existence of a nematic phase in a model of $l\times w$ hard rectangles on the square lattice with two allowed orientations and a large aspect ratio $k:=l/w$.
    The proof is based on a two-scale cluster expansion method developed previously by Disertori--Giuliani for hard rods in 2D and Disertori--Giuliani--Jauslin for hard plates in 3D. 
    Our main contributions lie in explicitly tracking the constants and parameters and completing the arguments left implicit in these works. 
    Hence, the proof produces a sufficient set of quantitative conditions from which estimates for the required aspect ratio can be extracted.
    A non-optimized evaluation of these conditions yields the bound $k\ge 10^{72}$.
    Although it vastly overshoots the numerical prediction, $k_{\min}=7$, our result appears to be the first rigorous estimate of the aspect ratio required for the formation of a nematic phase in hard rectangle systems.
\end{abstract}


\section{Introduction}

Entropy-driven ordering in systems of anisotropic hard particles is one of the oldest and most fascinating subjects in equilibrium statistical mechanics. 
In 1949, Onsager~\cite{onsager1949effects} proposed a foundational theory aimed to explain how elongated particles in the continuum can spontaneously align with purely repulsive interactions: for sufficiently long rods at intermediate densities, the loss of orientational entropy due to alignment may be compensated by the gain in translational entropy from minimizing the \textit{excluded volume} in the system.
This excluded-volume effect provides a paradigmatic explanation of the isotropic-nematic transition in liquid crystals and has profoundly influenced the study of rod-like systems ever since~\cite{bricmont1984structure,disertori2013nematic,ghosh2007orientational}.

From a theoretical point of view, lattice models of hard rods and rectangles provide a convenient testing ground for the ordering mechanisms at play.
In these approximations of true continuum systems, the particles are typically confined to a discrete set of positions and orientations, converting the competition between order and entropy into the combinatorial problem of counting packings~\cite{dhar2021entropy,hadas2025columnar}.
At the same time, these models remain sufficiently complex to exhibit phenomena of physical interest, including the nematic phase in which the orientational symmetry is broken, and the columnar phase where, in addition, translational symmetry becomes partially broken~\cite{kundu2014phase,kundu2015phase}.

\subsection{Background}

By now, the physical picture for hard rods and rectangles in two dimensions is quite well-charted.
For rods of length $k$ (and unit width) on the square lattice, Ghosh--Dhar~\cite{ghosh2007orientational} argued that, for $k\ge 7$, increasing density should first lead to an isotropic-nematic transition and then a second, re-entrant transition from the nematic phase to a high-density disordered phase; in an independent numerical study, Matoz-Fernandez--Linares--Ramirez-Pastor~\cite{matoz2008critical} reported the same length threshold for the formation of a nematic phase in hard rod systems.
Subsequent work by Kundu--Rajesh--Dhar--Stilck~\cite{kundu2013nematic} introduced efficient Monte Carlo dynamics that better addressed the jamming problems of standard local moves, thereby giving strong numerical evidence for the three-phase scenario.
In particular, they studied the nematic-disordered transition in detail on the square and triangular lattices and found a continuous transition for $k=7$ on both lattices and a much larger characteristic length scale in the high-density disordered phase on the square lattice than on the triangular lattice.
More recently, Shah--Dhar--Rajesh~\cite{shah2022phase} argued that the second transition should in fact be first-order for large $k$ and presented numerical evidence for $k\ge 9$.

For hard rectangles, the phase structure is even more elaborate.
Kundu--Rajesh~\cite{kundu2014phase,kundu2015phase} showed numerically that hard rectangles of size $m\times mk$ on the square lattice, for both integer and non-integer \textit{aspect ratios} $k\ge 1$, can exhibit up to four phases---as the density increases, an isotropic phase, a nematic phase, a columnar phase, and a high-density \textit{sublattice} or disordered phase depending on the arithmetic of the side lengths. In particular, for $m=2,3$, the nematic phase appears only when $k\ge 7$, a threshold that thus constitutes a natural physical benchmark for the onset of nematic order on the lattice. 
Their subsequent asymptotic study~\cite{kundu2015asymptotic} showed that the critical packing fraction for the isotropic-nematic transition scales like $k^{-1}$ for large $k$, where the scaling coefficient is independent of $m$, while that for the nematic-columnar transition approaches a nontrivial limit below the maximum packing fraction.
Complementary surface-tension and high-activity perturbative analyses by Nath--Dhar--Rajesh~\cite{nath2016stability} and Nath--Kundu--Rajesh~\cite{nath2015high} explained the nature of the nontrivial limit and clarified the role of the sliding instability in the columnar phase.

In contrast, the mathematical literature on rod-like systems, and more broadly on models with sliding instabilities, remains remarkably sparse~\cite{mazel2025high}.
In a breakthrough in 2013, Disertori--Giuliani~\cite{disertori2013nematic} rigorously combined a two-scale coarse-graining scheme and the cluster expansion method to prove the existence of a nematic phase in systems of \textit{very} long hard rods on the square lattice with two allowed orientations.
Later, Disertori--Giuliani--Jauslin~\cite{disertori2020plate} extended the method to prove the existence of a uniaxial nematic phase in systems of anisotropic hard plates in the three-dimensional continuum with six allowed orientations.
With an added attractive interaction between the ends of rods to promote alignment, Jauslin--Lieb~\cite{jauslin2018nematic} and the author~\cite{he2025extended} proved the existence of a nematic phase in a system of interacting dimers (hard rods with $k=2$), introduced originally by Heilmann--Lieb~\cite{heilmann1972theory,heilmann1979lattice}, over different parameter regimes.
We also highlight a recent result by Hadas--Peled~\cite{hadas2025columnar} on the existence of a columnar phase for a system of $2\times 2$ hard squares on the square lattice, whose proof is of a markedly different flavor from~\cite{disertori2013nematic,disertori2020plate} due to its reliance on reflection-positivity methods which, in turn, limits its generalization to systems of larger particles.

Apart from their sparsity, one limitation of the existing rigorous results on rod-like systems is that they remain largely existential.
Specifically, the Disertori--Giuliani result~\cite{disertori2013nematic} proves nematic order for $zk\ll 1\ll zk^2$, where $z>0$ is the fugacity, but the scaling coefficients are not quantified.
Similarly, for the plates of dimensions $1\times k^\alpha\times k$, $\alpha\in(\frac{3}{4},1]$, the Disertori--Giuliani--Jauslin result~\cite{disertori2020plate} proves uniaxial nematic order for $zk^{3-\alpha}\ll 1\ll zk^{3\alpha}(\log k)^{-1}$, also without estimating the scaling coefficients.
Thus, there is a quantitative gap between the rigorous theory and the physical predictions, where one would like to say not only that nematic order exists for \textit{sufficiently} elongated (or anisotropic) particles, but also \textit{how} elongated they need to be.

\subsection{Method and main result}

The purpose of the present paper is to bridge this gap for two-dimensional hard rods and rectangles.
In particular, to our knowledge, hard rectangle systems have not been studied in this generality in the mathematical literature.

Our method is best understood as a quantitative hybrid of those of Disertori--Giuliani~\cite{disertori2013nematic} and Disertori--Giuliani--Jauslin~\cite{disertori2020plate}.
As in these works, for $l\times w$ rectangles, we partition the space into boxes of side length approximately $l/2$ and implement a coarse-graining scheme in which a ``spin'' is assigned to each box according to the orientation of the rectangles centered in it.
By construction, rectangles centered in the same box necessarily have the same orientation (neglecting a slight caveat related to empty boxes), so the procedure leads to an effective spin model with two ground states, corresponding to having all horizontal or all vertical rectangles, to which contour methods reminiscent of those used in the implementation of Pirogov--Sinai theory~\cite{zahradnik1984alternate} can be applied.
By working in the intermediate-density regime, the density is \textit{low} enough to allow the use of the cluster expansion of the uniformly oriented system to control the error from the coarse-graining, but \textit{high} enough to ensure that boxes are likely to have at least one rectangle centered in them so that the ``spin'' is unambiguously defined with high probability.

The necessity of combining the methods of both papers is due to the nature of the hard-core interaction in the uniformly oriented system, which is only effectively one-dimensional for rectangles of unit width (i.e., rods); see Rem.~\ref{rmk: non-one-dimensionality of the interaction between rectangles}.
Apart from this, our contribution lies in the explicit tracking of the constants and parameters and the completion of the arguments left implicit in~\cite{disertori2013nematic,disertori2020plate}, so the proof of the main theorem yields a sufficient set of quantitative conditions which can be checked and converted into a rigorous bound on the aspect ratio required for the formation of a nematic phase.
To our knowledge, the present paper is the first in the literature to rigorously produce such an estimate.
In simple terms, our main result is as follows.

\begin{theorem*}[Nematic order]
    \label{thm: simplified main}
    Hard rectangle models on $\Z^2$ with two allowed orientations and aspect ratios $k\ge 10^{72}$ exhibit a nematic phase at intermediate fugacities.
\end{theorem*}

\subsection{Open problems}

Many questions about hard rectangle systems still await rigorous answers.

From the point of view of Thm.~\ref{thm: simplified main}, the most immediate one is quantitative: the lower bound of $10^{72}$, while explicit, is not physically meaningful and leaves an enormous gap from the numerical prediction of $7$~\cite{ghosh2007orientational,matoz2008critical}.
It should not be too difficult to improve the bound in Thm.~\ref{thm: simplified main} by optimizing the bookkeeping and choices of parameters; once this is done, the remaining gap (which is most likely present) would give a concrete measure of the losses in the two-scale cluster expansion argument.
Closing the gap further would require a sharper way of isolating the entropic mechanism responsible for orientational order; in particular, Onsager's excluded volume effect may start to lose its efficacy at lower aspect ratios~\cite{frenkel1987onsager,xiao2013generalized}.

Several qualitative problems also remain completely open from a mathematical point of view.
The first one is to ascertain the existence of the remaining phases predicted in the physical literature, including the columnar phase at higher densities and the re-entrant disordered phase at the highest densities.
Setting aside the high-density disordered phase which is completely uncharted territory (see~\cite[Sec.~9]{disertori2013nematic} for a potentially simpler combinatorial problem about the residual entropy at close-packing), a proof of the nematic-columnar transition would already be a major advance.
It will require new methods for managing the delicate sliding instabilities that operate far outside the low-density regime in which the standard cluster expansion used here is effective.
Inspired by the discussions in~\cite[Sec.~2]{nath2015high}, we propose here a potentially more tractable problem:
\begin{problem}
    For a system of $l\times w$ hard rectangles on the square lattice in which all the rectangles are \emph{constrained} to be horizontally oriented, prove the existence of columnar order at all sufficiently high fugacities when $k=l/w\gg 1$.
\end{problem}

There are two main reasons we expect this problem to be more tractable.
First, by removing the orientational degree of freedom, there is no need to establish orientational order before addressing translational symmetry breaking.
Second, when $k\gg 1$, the close-packed configurations exhibit a strong entropic bias in favor of breaking translational symmetry in the direction normal to the long axes of the rectangles.
Indeed, every close-packed configuration of the model may be classified as either \textit{row-aligned}, meaning that all the rectangles have the same $y$-coordinate modulo $w$ and are free to slide within their respective rows, or \textit{column-aligned} with the analogous meaning.
The former family vastly outnumbers the latter when $k\gg 1$ on any torus with side divisible by both $l$ and $w$.
Moreover, it is the former that is associated with the expected columnar phase, in which one of the $w$ row (residue) classes is selected, and the remaining sliding degrees of freedom are consistent with the required translation invariance along the horizontal direction.
The closest rigorous precedent is the proof of columnar order for $2\times2$ hard squares in~\cite{hadas2025columnar}, though reflection positivity methods do not apply in this setting.
Pirogov--Sinai theory provides a more promising approach~\cite{alberici2016cluster,jauslin2018nematic}, though the combinatorics of contours can be treacherous to navigate in sliding models with only hard-core interactions.
Finally, the advanced Peierls-type arguments of~\cite{peled2017condition,peled2020long,peled2023rigidity}, developed for systems with positive residual entropy, may offer valuable guidance for controlling the entropy of defects in this system (which has zero residual entropy per site at close-packing).

Returning to the qualitative problems that remain open mathematically, the second one concerns the asymptotic range of fugacities in which nematic order manifests.
As in~\cite{disertori2013nematic}, the present proof works within the regime $zwl\ll1\ll zl^2$.
While the lower limit $z\gg l^{-2}$ is consistent with the asymptotic prediction in~\cite{kundu2015asymptotic} (using that $z\approx\rho$ at low densities), the upper limit is far from what is expected~\cite{kundu2015asymptotic,nath2016stability}.
Improving the latter will require new arguments that do not rely on the cluster expansion of the uniformly oriented system, which the current method uses extensively.

Finally, there is the long-standing open problem of proving orientational symmetry breaking in a model with continuous rotational symmetry. 
The problem is old enough to border on clich\'e, but it remains one of the central reasons for studying the more rigid lattice models, which provide a setting in which the entropic mechanism can be explicitly tested before one confronts the additional difficulties associated with continuous symmetries.

\section{Preliminaries}

\subsection{Model and main result}
\label{sec: model and main result}

In this subsection, we set up the formal definitions and notation for the hard rectangle model and state our main result in more precise terms than earlier.

\paragraph{Elementary objects} 
Let $\fw\ge 1$ and $\fl\ge\max\set{3,2\fw-1}$ be integers.
In this paper, a \textbf{rectangle} refers to an $\fl\times\fw$, axis-parallel rectangle in $\R^2$ with all four corners in $\Z^2$, oriented either horizontally or vertically.
The ratio $\fk:=\fl/\fw\ge 1$ is called the \textbf{aspect ratio} of the rectangle.
The \textbf{anchor} of a rectangle is the point $(x,y)\in\Z^2$ with minimum Euclidean distance to its center $(x_c,y_c)\in\R^2$ subject to the constraints $x\le x_c$ and $y\le y_c$. 
Here, we are interested in the regime $\fk\gg1$.

Equip $\Z^2$ with the $L_\infty$-distance, denoted by $\rmd$.
Introduce the \textit{mesoscopic} length scale $\ell:=\ceil{\fl/2}$ and the scaled square lattice $\bL:=\ell\Z^2$, equipped with the connectivity relation in which two vertices are \textbf{adjacent} if and only if their $\rmd$-distance equals the lattice spacing $\ell$.
Let $\rmd_\bL:=\rmd/\ell$ be the rescaled $\rmd$-distance on $\bL$.
Anchored at each vertex $v\in\bL$ is a \textbf{tile} $T_v:=v+[0,\ell)^2$.
We declare two tiles to be adjacent if their anchors are adjacent in $\bL$.
If $\Lambda$ is a union of tiles, then its \textbf{internal} (resp. \textbf{external}) $\bL$-\textbf{boundary}, denoted by $\bdin[\bL]{\Lambda}$ (resp. $\bdex[\bL]{\Lambda}$), is defined as the union of those tiles in $\Lambda$ which are adjacent to a tile in $\R^2\setminus\Lambda$ (resp. those in $\R^2\setminus\Lambda$ which are adjacent to a tile in $\Lambda$).
The preceding definitions extend naturally to scaled versions of $\bL$.
In particular, the scaled tiles of the lattice $6\bL$ are called \textbf{smoothing squares} and denoted by $S_u:=u+[0,6\ell)^2$, $u\in 6\bL$.
Lastly, let $\cT(\Lambda)$ be the set of \textit{finite} unions of tiles in $\Lambda\cup\bdex[\bL]{\Lambda}$, $\cT_\ast(\Lambda)\subset\cT(\Lambda)$ consist of the connected ones, and $\cB_\ast(\Lambda)$ be the power set of $\cT_\ast(\Lambda)$.

The following simple but crucial observation will be proved in Sec.~\ref{sec: one orientation per tile}.

\begin{lemma}[Excluded volume]
    \label{lem: one orientation per tile}
    Let $u,v\in\Z^2$.
    If $\rmd(u,v)\le\ell-1$, then two rectangles with different orientations anchored at $u$ and $v$ necessarily have overlapping interiors.
    In particular, if two rectangles anchored in the same tile have different orientations, then their interiors overlap.
\end{lemma}

\paragraph{Configuration spaces}

We think of a \textbf{configuration} of hard rectangles intuitively as a set of rectangles with pairwise disjoint interiors.
Equivalently, a configuration is represented by a function $\omega:\Z^2\to\set*{-1,0,1}$ such that the vertices $v$ with $\omega(v)\ne 0$ are the anchors of the rectangles, and the value of $\omega$ at such a vertex indicates the orientation of the rectangle, with $1$ (resp. $-1$) corresponding to a horizontal (resp. vertical) rectangle.
Given $q\in\set{-1,1}$ and $\Lambda\subset\R^2$, denote the set of $q$-oriented rectangles anchored in $\Lambda$ by $R^q(\Lambda)$, and the set of all rectangles anchored in $\Lambda$ by $R(\Lambda)$.

A \textbf{profile} is a function $\psi:\bL\to\set*{-1,0,1}$.
Denote by $\Psi$ the set of all profiles.

Let $\Lambda$ be a union of smoothing squares.
Denote by $\Omega(\Lambda)$ the set of configurations where every rectangle is anchored in $\Lambda$.
For $q\in\set*{-1,1}$, the set of profiles on $\Lambda$ with $q$-boundary conditions, denoted by $\Psi_\Lambda^q$, consists of those $\psi\in\Psi$ such that $\psi(v)=q$ whenever the enclosing smoothing square $S_u$ of the tile $T_v$ satisfies $\rmd_{6\bL}(u,\Lambda^c)\le 2$.
Given $\psi\in\Psi_\Lambda^q$, the set of configurations in $\Lambda$ which are \textbf{compatible} with $\psi$ is the subset $\Omega(\Lambda,\psi)\subset\Omega(\Lambda)$ consisting of those configurations such that, for every $v\in\Lambda\cap\bL$,
\begin{itemize}
    \item if $\psi(v)=0$, then no rectangle is anchored in $T_v$;
    \item if $\psi(v)=1$, then every rectangle anchored in $T_v$ is horizontal, even if vacuously;
    \item if $\psi(v)=-1$, then every rectangle anchored in $T_v$ is vertical, even if vacuously.
\end{itemize}
Furthermore, given $R_{\Lambda^c}\in\Omega(\Lambda^c)$, the set of configurations in $\Lambda$ compatible with both $\psi$ and $R_{\Lambda^c}$, denoted by $\Omega(\Lambda,\psi\mid R_{\Lambda^c})$, is the subset of configurations $R_{\Lambda}\in\Omega(\Lambda,\psi)$ such that $R_{\Lambda}\cup R_{\Lambda^c}$ is a (valid) configuration.

\paragraph{Partition functions}

Let $\bz:\Z^2\times\set{-1,1}\to[0,\infty)$ be a function and $\Lambda$ be a union of smoothing squares.
The \textbf{weight} of a configuration $R\in\Omega(\Lambda)$ at \textbf{fugacity} $\bz$ is defined as $\bz(R):=\prod_{(v,q)\in R}\bz(v,q)$; the notation $\bz(\cdot)$ will later be applied to \textit{multisets} of rectangles as well.
If $\fr_i$ is a fixed rectangle, we write $\partial_i:=\pdv*{\bz(\fr_i)}$ and $\cD_i:=\bz(\fr_i)\pdv*{\bz(\fr_i)}$.
For $q\in\set*{-1,1}$, the \textbf{partition function} in $\Lambda$ with $q$-boundary conditions is given by
\begin{equation}
    \label{eqn: q-magnetized partition function}
    \Xi_\bz(\Lambda\mid q):=\sum_{\substack{R\in\Omega(\Lambda)\\\exists\psi\in\Psi_\Lambda^q:R\sim\psi}}\bz(R).
\end{equation}
Let $\bq$ be the profile taking the constant value $q$ on $\bL$, and denote $\Omega^q(\Lambda):=\Omega(\Lambda,\bq)$.
We also introduce the \textbf{$q$-oriented partition function} in $\Lambda$ as
\begin{equation}
    \label{eqn: q-oriented partition function}
    \Xi_\bz^q(\Lambda):=\sum_{R\in\Omega^q(\Lambda)}\bz(R).
\end{equation}
Furthermore, given $R_{\Lambda^c}\in\Omega(\Lambda^c)$, the constrained partition functions $\Xi_\bz(\Lambda\mid q,R_{\Lambda^c})$ and $\Xi_\bz^q(\Lambda\mid R_{\Lambda^c})$ are defined analogously to~\eqref{eqn: q-magnetized partition function} and~\eqref{eqn: q-oriented partition function} by restricting the respective sums to the spaces $\Omega(\Lambda,\psi\mid R_{\Lambda^c})$ and $\Omega^q(\Lambda\mid R_{\Lambda^c})$.

\paragraph{Main result}

Given rectangles $\fr_1,\fr_2$, denote by $\fT_i$ the tile containing the anchor of $\fr_i$, and define $\fd_{12}$ as in Appx.~\ref{appx: derived quantities}.
The \textbf{one-} and \textbf{two-point correlation functions} of the hard rectangle model in $\Lambda$ with $q$-boundary conditions at fugacity $z$ are defined as
\begin{equation}
    \rho^q_{z,\Lambda;1}(\fr_1)
    :=\eval{\cD_1\log \Xi_\bz(\Lambda\mid q)}_z,
\end{equation}
\begin{equation}
    \rho^q_{z,\Lambda;2}(\fr_1,\fr_2)-\rho^q_{z,\Lambda;1}(\fr_1)\rho^q_{z,\Lambda;1}(\fr_2)
    :=\eval{\cD_1\cD_2\log \Xi_\bz(\Lambda\mid q)}_z.
\end{equation}
Our main result is as follows.

\begin{theorem}[Nematic order]
    \label{thm: main}
    Let $q\in\set{-1,1}$, $\fr_1,\fr_2$ be (compatible) rectangles such that $\fd_{12}\ge 6$, and $\Lambda$ be a finite union of smoothing squares.
    Suppose that $\fk\ge 10^{72}$.
    Then, with $W$ denoting the principal branch of the Lambert $W$-function, the open interval
    \begin{equation}
        \bigg(\max\set[\bigg]{\frac{500}{\ell^2}W\bigg(\frac{\ell^2}{500e(2\fw-1)(2\fl-1)}\bigg),\frac{150000}{\ell^2}}
        ,\frac{1000}{\ell^2}W\bigg(\frac{\ell^2}{1000e(2\fw-1)(2\fl-1)}\bigg)\bigg)
    \end{equation}
    is nonempty.
    Moreover, for every fugacity $z$ in this interval, the estimates
    \begin{equation}
        \abs{\rho^q_{z,\Lambda;1}(\fr_1)-\indicator{R^q(\Lambda)}(\fr_1)z}
        \le \indicator{R(\Lambda)}(\fr_1)\bigg(\frac{1}{4}+e\eps\bigg)z,
    \end{equation}
    \begin{equation}
        \abs{\rho^q_{z,\Lambda;2}(\fr_1,\fr_2)-\rho^q_{z,\Lambda;1}(\fr_1)\rho^q_{z,\Lambda;1}(\fr_2)}
        \le \indicator{R(\Lambda)}(\fr_1,\fr_2)\brackets[\Big]{e^3\eps^{\fd_{12}-2}+z^2 e^{-(z\ell^2/3000000)(\fd_{12}-2)}}
    \end{equation}
    hold uniformly in $q$, $\fr_1$, $\fr_2$, and $\Lambda$, where $\eps<e^{-150}$ is a small positive parameter.
    Consequently,
    \begin{equation}
        \liminf_{\Lambda\uparrow\Z^2}\rho^q_{z,\Lambda;1}(\fr_1)
        \ge \bigg(\frac{3}{4}-e\eps\bigg)z,
        \quad\text{if }\fr_1\in R^{q}(\Z^2),
    \end{equation}
    \begin{equation}
        \limsup_{\Lambda\uparrow\Z^2}\rho^q_{z,\Lambda;1}(\fr_1)
        \le \bigg(\frac{1}{4}+e\eps\bigg)z,
        \quad\text{if }\fr_1\in R^{-q}(\Z^2),
    \end{equation}
    which indicates the breaking of orientational symmetry, and
    \begin{equation}
        \limsup_{\Lambda\uparrow\Z^2}\abs{\rho^q_{z,\Lambda;2}(\fr_1,\fr_2)-\rho^q_{z,\Lambda;1}(\fr_1)\rho^q_{z,\Lambda;1}(\fr_2)}
        \le e^3\eps^{\fd_{12}-2}+z^2 e^{-(z\ell^2/3000000)(\fd_{12}-2)},
    \end{equation}
    which indicates the absence of translational order.
\end{theorem}

\subsection{Cluster expansion tools}
\label{sec: cluster expansion tools}

In this subsection, we review Ueltschi's formulation~\cite{ueltschi2004cluster} of the cluster expansion method in the context of abstract polymer models with hard-core interactions.
Here and in the remainder of the paper, we will use $\set{\cdot}$, $\mset{\cdot}$, and $(\cdot)$ to denote sets, multisets, and ordered tuples, respectively.

Let $\bA$ be a finite or countable set whose elements are called \textbf{polymers}.
The weight of a polymer is specified by a function $w:\bA\to\R$.
Let $\sim$ be an \textit{irreflexive} and \textit{symmetric} relation on $\bA$ called \textbf{compatibility}.
The (formal) \textbf{polymer partition function} is defined as
\begin{equation}
    \Xi:=\sum_{L\subset\bA}\brackets[\Bigg]{\prod_{A\in L}w(A)}\brackets[\Bigg]{\prod_{\set*{A,A'}\subset L}\indicator{A\sim A'}}.
\end{equation}

A finite multi-subset $M$ of $\bA$ is a \textbf{cluster} if it cannot be partitioned into two nonempty sub-multisets $M_1,M_2$ such that every polymer in $M_1$ is compatible with every polymer in $M_2$.
Denote by $\cM(\bA)$ the collection of all clusters.
For $n\ge 1$, let $G_n$ denote the complete graph on $n$ vertices, and write $G\subset G_n$ if $G$ is a spanning subgraph of $G_n$.
Define the function $\varphi:\cM(\bA)\to\R$ by
\begin{equation}
    \varphi(M):=\prod_{A\in\bA}\frac{1}{n_A(M)!}\cdot\sum_{\substack{G\subset G_n\\\text{connected}}}\prod_{\set*{i,j}\in E(G)}(-\indicator{A_i\not\sim A_j})
\end{equation}
for $M=\mset*{A_1,\dots,A_n}$, where $n_A(M)$ denotes the multiplicity of $A$ in $M$.

\begin{theorem}[{Cluster expansion,~\cite{ueltschi2004cluster}}]
    \label{thm: Ueltschi criterion}
    Suppose that there exist functions $a,b:\bA\to[0,\infty)$ such that
    \begin{equation}
        \label{eqn: Ueltschi criterion}
        \sum_{A\not\sim A_\ast}e^{a(A)+b(A)}\abs{w(A)}\le a(A_\ast)\quad\text{for every }A_\ast\in\bA.
    \end{equation}
    Then, for all $A_1\in\bA$,
    \begin{equation}
        \label{eqn: cluster bound with one anchor}
        \sum_{M\ni A_1}\abs{\varphi(M)}n_{A_1}(M)\prod_{A\in M\setminus\mset*{A_1}}\left[e^{b(A)}\abs{w(A)}\right]\le e^{a(A_1)},
    \end{equation}
    \begin{equation}
        \label{eqn: cluster bound with two identical anchors}
        \sum_{M\supset\mset*{A_1,A_1}}\abs{\varphi(M)}n_{A_1}(M)(n_{A_1}(M)-1)\prod_{A\in M\setminus\mset*{A_1,A_1}}\left[e^{b(A)}\abs{w(A)}\right]
        \le 2e^{3a(A_1)}.
    \end{equation}
    Moreover, for all distinct $A_1,A_2\in\bA$,
    \begin{equation}
        \label{eqn: cluster bound with two distinct anchors}
        \sum_{M\supset\mset*{A_1,A_2}}\abs{\varphi(M)}n_{A_1}(M)n_{A_2}(M)\prod_{A\in M\setminus\mset*{A_1,A_2}}\left[e^{b(A)}\abs{w(A)}\right]
        \le e^{\frac{3}{2}[a(A_1)+a(A_2)]}(1+\indicator{A_1\not\sim A_2}).
    \end{equation}
    If, in addition to~\eqref{eqn: Ueltschi criterion}, it holds that
    \begin{equation}
        \sum_{A\in\bA}e^{a(A)}\abs{w(A)}<\infty,
    \end{equation}
    then the polymer partition function $\Xi$ is well-defined, $\Xi>0$, and its logarithm admits the absolutely convergent series expansion
    \begin{equation}
        \log\Xi=\sum_{M\in\cM(\bA)}\varphi(M)\prod_{A\in M}w(A).
    \end{equation}
\end{theorem}

\begin{remark}
    In~\cite{ueltschi2004cluster}, a cluster is an ordered tuple of polymers rather than a multiset thereof.
    However, the formulation of the theorem given here is easily shown to be equivalent to the original. 
\end{remark}

\subsubsection{Application to the $q$-oriented partition function}
\label{sec: cluster expansion of the q-oriented partition function}

In this sub-subsection, we apply the cluster expansion method to the $q$-oriented partition function.
Here, we assume that there exist $z,\nu>0$ and an integer $n\ge 1$ such that the fugacity $\bz$ takes values in $[0,z]$ for all but at most $n$ inputs, for which it takes values in $(z,e^{\nu/n}z)$.

Let $\Lambda$ be a \textit{finite} union of smoothing squares, $q\in\set*{-1,1}$, and $R_{\Lambda^c}\in\Omega(\Lambda^c)$.
The constrained $q$-oriented partition function $\Xi^q_\bz(\Lambda\mid R_{\Lambda^c})$ can be viewed as a polymer partition function with $\bA$ the set of $q$-oriented rectangles anchored in $\Lambda$, the relation $\sim$ prescribing disjoint interiors, and the weight given by $w(r):=\indicator{r\sim R_{\Lambda^c}}\bz(r)$.
Provided that
\begin{equation}
    \label{eqn: convergence condition for expansion of q-oriented partition function}
    e^{1+\nu}(2\fw-1)(2\fl-1)z\le 1,
\end{equation}
we find that the condition~\eqref{eqn: Ueltschi criterion} is satisfied with $a(A)\equiv 1$ and $b(A)\equiv b$, where $b\ge 0$ is such that $e^{1+\nu+b}(2\fw-1)(2\fl-1)z=1$.
Hence, we have the expansion
\begin{equation}
    \log \Xi_\bz^q(\Lambda\mid R_{\Lambda^c})
    =\sum_{\substack{M\in\cM^q(\Lambda)\\M\sim R_{\Lambda^c}}}\varphi(M)\bz(M),
\end{equation}
where, as in the rest of the paper, $\cM^q(\Lambda)$ denotes the collection of clusters of $q$-oriented rectangles anchored in $\Lambda$.
Moreover, introducing the decay factor $\epsilon:=e^{-b}=e^{1+\nu}(2\fw-1)(2\fl-1)z$, we have the tail estimate
\begin{equation}
    \label{eqn: rectangle cluster tail estimate}
    \sum_{\substack{M\in\cM^q(\Z^2)\\M\ni \fr_1\\\abs{M}\ge n}}\abs{\varphi(M)}n_{\fr_1}(M)\bz(M\setminus\mset{\fr_1})
    \le e\eps^{n-1},\quad\text{for every }n\ge 1,
\end{equation}
as a consequence of~\eqref{eqn: cluster bound with one anchor}.
From here on, we will assume that the condition~\eqref{eqn: convergence condition for expansion of q-oriented partition function} is satisfied.

\begin{lemma}
    \label{lem: lower bound on the size of clusters with anchors in distant tiles}
    If $T_1,T_2$ are distinct tiles and $M\in\cM^q(\Z^2)$ contains rectangles anchored in both $T_1$ and $T_2$, then $\abs{M}\ge 1+\ceil{\rmd_{\bL}(T_1,T_2)/2}$.
\end{lemma}

\begin{proof}
    We prove the lemma for $q=1$ (horizontal rectangles); the case $q=-1$ is analogous.
    It is easy to check that the interiors of two horizontal rectangles anchored respectively at $(x_1,y_1)$ and $(x_2,y_2)$ intersect if and only if
    \begin{equation}
        \abs{x_1-x_2}\le \fl-1,\quad\abs{y_1-y_2}\le \fw-1.
    \end{equation}
    Since $\fl-1\le 2\ell$, the maximum $\rmd$-distance between tiles containing the anchors of two overlapping horizontal rectangles is $2$.
    Hence, $\abs{M}\ge 1+\rmd_{\bL}(T_1,T_2)/2$. 
    The proof is complete after rounding up the RHS to the nearest integer.
\end{proof}

Below, we record some additional implications of Thm.~\ref{thm: Ueltschi criterion} for the $q$-oriented partition function which will be used in the subsequent sections.

\begin{lemma}
    Let $\Lambda\subset\Z^2$ be bounded.
    The following estimates hold.
    \begin{enumerate}[label=(\roman*),ref=\thelemma(\roman*)]
        \item \label{itm: expansion of the log q-oriented partition function}
        $\bars[\big]{\log\Xi_\bz^q(\Lambda)-\sum_{r\in R^q(\Lambda)}\bz(r)}
        \le e^{1+\nu}z\eps\abs{\Lambda}$.
        \item \label{itm: expansion of the first-order derivative of the log q-oriented partition function} If $\fr_1$ is a rectangle, then
        $\bars*{\partial_1\log\Xi_\bz^q(\Lambda)-\indicator{R^q(\Lambda)}(\fr_1)}\le \indicator{R^q(\Lambda)}(\fr_1)e\eps$.
        \item \label{itm: expansion of the second-order derivative of the log q-oriented partition function} If $\fr_1,\fr_2$ are compatible rectangles, then
        $\bars*{\partial_1\partial_2\log\Xi_\bz^q(\Lambda)}
        \le \indicator{R^q(\Lambda)}(\fr_1,\fr_2)e^3\eps^{\fd_{12}-2}$.
    \end{enumerate}
\end{lemma}

\begin{proof}[Proof of~\ref{itm: expansion of the log q-oriented partition function}]
    Isolating the one-rectangle clusters in the cluster expansion yields
    \begin{equation}
        \log \Xi_\bz^q(\Lambda)=\sum_{r\in R^q(\Lambda)}\bz(r)+\sum_{\substack{M\in\cM^q(\Lambda)\\\abs{M}\ge 2}}\varphi(M)\bz(M),
    \end{equation}
    where
    \begin{equation}
        \sum_{\substack{M\in\cM^q(\Lambda)\\\abs{M}\ge 2}}\abs{\varphi(M)}\bz(M)
        \le\sum_{r\in R^q(\Lambda)}\bz(r)\sum_{\substack{M\in\cM^q(\Z^2)\\M\ni r\\\abs{M}\ge 2}}\abs{\varphi(M)}\bz(M\setminus\mset{r})
        \le e^{1+\nu}z\eps\abs{\Lambda}
    \end{equation}
    by the tail estimate~\eqref{eqn: rectangle cluster tail estimate}.
\end{proof}

\begin{proof}[Proof of~\ref{itm: expansion of the first-order derivative of the log q-oriented partition function}]
    Differentiating the cluster expansion term-by-term yields
    \begin{equation}
        \partial_1\log\Xi^q_\bz(\Lambda)
        =\indicator{R^q(\Lambda)}(\fr_1)\brackets[\Bigg]{1+\sum_{\substack{M\in\cM^{q}(\Lambda)\\M\ni \fr_1\\\abs{M}\ge 2}}\varphi(M)n_{\fr_1}(M)\bz(M\setminus\mset{\fr_1})}.
    \end{equation}
    The remaining series is bounded using~\eqref{eqn: rectangle cluster tail estimate}.
\end{proof}

\begin{proof}[Proof of~\ref{itm: expansion of the second-order derivative of the log q-oriented partition function}]
    Differentiating the cluster expansion term-by-term yields
    \begin{equation}
        \partial_1\partial_2\log\Xi^q_\bz(\Lambda)
        =\indicator{R^q(\Lambda)}(\fr_1,\fr_2)\sum_{\substack{M\in\cM^{q}(\Lambda)\\M\supset\mset{\fr_1,\fr_2}\\\abs{M}\ge\fd_{12}}}\varphi(M)n_{\fr_1}(M)n_{\fr_2}(M)\bz(M\setminus\mset{\fr_1,\fr_2}),
    \end{equation}
    where the constraint $\abs{M}\ge\fd_{12}$ follows from Lem.~\ref{lem: lower bound on the size of clusters with anchors in distant tiles}.
    The remaining series is bounded using~\eqref{eqn: cluster bound with two distinct anchors}.
\end{proof}

\subsection{Proof of Lemma~\ref{lem: one orientation per tile}}
\label{sec: one orientation per tile}

Let $r_h$ be a horizontal rectangle anchored at $(x_h,y_h)\in\Z^2$.
It is easy to check that the interior of a vertical rectangle $r_v$ intersects that of $r_h$ if and only if its anchor $(x_v,y_v)\in\Z^2$ satisfies
\begin{align}
    -\floor[\bigg]{\frac{\fl}{2}}-\ceil[\bigg]{\frac{\fw}{2}}+1&{}\le x_v-x_h\le \ceil[\bigg]{\frac{\fl}{2}}+\floor[\bigg]{\frac{\fw}{2}}-1, \\
    -\ceil[\bigg]{\frac{\fl}{2}}-\floor[\bigg]{\frac{\fw}{2}}+1&{}\le y_v-y_h\le \floor[\bigg]{\frac{\fl}{2}}+\ceil[\bigg]{\frac{\fw}{2}}-1.
\end{align}
Since
\begin{equation}
    \ell-1
    =\ceil[\bigg]{\frac{\fl}{2}}-1
    \le\min\set[\bigg]{\ceil[\bigg]{\frac{\fl}{2}}+\floor[\bigg]{\frac{\fw}{2}}-1,\floor[\bigg]{\frac{\fl}{2}}+\ceil[\bigg]{\frac{\fw}{2}}-1},
\end{equation}
if $\max\set{\abs{x_h-x_v},\abs{y_h-y_v}}\le\ell-1$, then $r_h$ and $r_v$ have intersecting interiors; in particular, this holds if $(x_h,y_h)$ and $(x_v,y_v)$ belong to the same tile.

\section{Contour and polymer representations}
\label{sec: contour and polymer representations}

In this section, we write the partition function of the hard rectangle model as a polymer partition function which will be studied via the cluster expansion method in Sec.~\ref{sec: peierls estimates for the polymer model}.
We achieve this via a two-scale coarse-graining scheme similar to that in~\cite{disertori2013nematic,disertori2020plate} by examining the homogeneity of the \textit{orientational} profiles with which the configurations are compatible, keeping in mind that typical configurations should be mostly compatible with one of the two constant profiles $\bq$, $q\in\set*{-1,1}$, in the nematic regime.

\subsection{Contours}
\label{sec: contours}

In this subsection, we describe the construction of contours for a given profile and a compatible configuration.

For $\psi\in\Psi$ and $q\in\set*{-1,1}$, a smoothing square $S$ is \textbf{$q$-correct} (with respect to $\psi$) if $\psi(v)=q$ for every tile $T_v$ contained in $S$ or an adjacent smoothing square; the smoothing square $S$ is \textbf{incorrect} if it is not $q$-correct for either $q$.

Let $\Lambda$ be a finite union of smoothing squares, $q\in\set*{-1,1}$, $\psi\in\Psi_\Lambda^q$, and $R\in\Omega(\Lambda,\psi)$.
A \textbf{contour} associated to the pair $(\psi,R)$ is a triple $\gamma=(\supp{\gamma},\psi_\gamma,R_\gamma)$, where $\supp{\gamma}$, called its \textbf{support}, is a (\textit{maximal}) connected component of incorrect smoothing squares, $\psi_\gamma$ is the restriction of $\psi$ to the set $\set*{v\in\Lambda\cap\bL\mid T_v\subset\supp{\gamma}}$, and $R_\gamma\subset R$ is the sub-configuration of rectangles anchored in $\supp{\gamma}$.
The connected components of $\R^2\setminus\supp{\gamma}$ are denoted by $\ext(\gamma),\intr_1(\gamma),\dots,\intr_{h_\gamma}(\gamma)$, where $\ext(\gamma)$ is the unique unbounded component, and $h_\gamma\ge 0$.
For each connected component $C$ of the internal $6\bL$-boundary of $\supp{\gamma}$, the restricted spin configuration $\psi_\gamma$ is constant on the set of anchors of tiles contained in $C$~\cite{timar2013boundary}.
In particular, if $C$ is adjacent to $\ext(\gamma)$, then this constant value is called the \textbf{type} of $\gamma$; and if $C$ is adjacent to $\intr_i(\gamma)$, this value is called the \textbf{type} of $\intr_i(\gamma)$ and denoted by $m_i(\gamma)$.
We write $c(\gamma):=\supp{\gamma}\cup_{j=1}^{h_\gamma}\intr_j(\gamma)$.

Let $C(\Lambda,q)$ denote the set of contours of type $q$ associated with any spin configuration $\psi\in\Psi_{\Lambda}^q$ and a compatible configuration $R\in\Omega(\Lambda,\psi)$.
Two contours $\gamma,\gamma'\in C(\Lambda,q)$ are \textbf{compatible} if their supports are not adjacent.
Moreover, $\gamma$ is said to be \textbf{external} to $\gamma'$ if $\supp{\gamma'}\subset\ext(\gamma)$.
Let $\cC(\Lambda,q)$ denote the collection of all sets of pairwise compatible contours from $C(\Lambda,q)$, and $\cC_\ext(\Lambda,q)\subset\cC(\Lambda,q)$ denote the subcollection of sets of mutually external contours.
Observe that the lattice spacing of $6\bL$ ensures that the configurations associated with compatible contours are compatible.

\subsection{Contour representation}
\label{sec: contour representation}

In this subsection, we express the partition function of the hard rectangle model as the partition function of a contour model with multi-body interactions.

We introduce the following notation.
Given $\Gamma\in\cC(\Lambda,q)$, we write $\supp{\Gamma}:=\cup_{\gamma\in\Gamma}\supp{\gamma}$ and $R_\Gamma^\ext:=\cup_{\gamma\in\Gamma}(R_\gamma\cap\bdin[6\bL]{c(\gamma)})$.
If $\Gamma\in\cC_\ext(\Lambda,q)$, we also write $\ext(\Gamma):=\cap_{\gamma\in\Gamma}\ext(\gamma)$.

We start by writing~\eqref{eqn: q-magnetized partition function} as a nested sum over profiles in $\Psi_\Lambda^q$ and their compatible configurations.
A complication that arises is that a tile with no rectangle anchored in it is compatible with every possible profile value.
To avoid overcounting, we associate a multiplicative factor of $-1$ to each instance of the profile value $0$~\cite[(3.5)]{disertori2013nematic}, which yields the identity
\begin{equation}
    \label{eqn: q-magnetized partition function rewritten}
    \Xi_\bz(\Lambda\mid q)=\sum_{\psi\in\Psi_{\Lambda}^q}\sum_{R\in\Omega(\Lambda,\psi)}(-1)^{\abs{\set*{\psi=0}\cap\Lambda}}\bz(R).
\end{equation}

Next, we isolate the external contours by writing
\begin{equation}
    \Xi_\bz(\Lambda\mid q)
    =\sum_{\Gamma\in\cC_\ext(\Lambda,q)}\Xi_\bz^q(\ext(\Gamma)\cap\Lambda\mid R_\Gamma^\ext)\prod_{\gamma\in\Gamma}\brackets[\Bigg]{(-1)^{\abs{\set*{\psi_\gamma=0}}}\bz(R_\gamma)\prod_{j=1}^{h_\gamma}\Xi_\bz(\intr_j(\gamma)\mid m_j(\gamma),R_\gamma)}.
\end{equation}
After inserting trivial identities, we obtain
\begin{equation}
    \label{eqn: q-magnetized partition function external contour representation}
    \frac{\Xi_\bz(\Lambda\mid q)}{\Xi_\bz^q(\Lambda)}
    =\sum_{\Gamma\in\cC_\ext(\Lambda,q)}\prod_{\gamma\in\Gamma}\brackets[\Bigg]{\zeta_q^0(\gamma;\bz)\prod_{j=1}^{h_\gamma}\frac{\Xi_\bz(\intr_j(\gamma)\mid q)}{\Xi_\bz^q(\intr_j(\gamma))}}e^{-W_0^\ext(\Gamma;\bz)},
\end{equation}
where 
\begin{equation}
    \label{eqn: bare weight of a contour}
    \zeta_q^0(\gamma;\bz)
    :=\frac{(-1)^{\abs{\set*{\psi_\gamma=0}}}\bz(R_\gamma)}{\Xi_\bz^q(\supp{\gamma})}\prod_{j=1}^{h_\gamma}\frac{\Xi_\bz(\intr_j(\gamma)\mid m_j(\gamma),R_\gamma)}{\Xi_\bz(\intr_j(\gamma)\mid q)}
\end{equation}
is the bare weight of the contour $\gamma$ at fugacity $\bz$, and
\begin{equation}
    e^{-W_0^\ext(\Gamma;\bz)}
    :=\frac{\Xi_\bz^q(\ext(\Gamma)\cap\Lambda\mid R_\Gamma^\ext)\prod_{\gamma\in\Gamma}\brackets[\big]{\Xi_\bz^q(\supp{\gamma})\prod_{j=1}^{h_\gamma}\Xi_\bz^q(\intr_j(\gamma))}}{\Xi_\bz^q(\Lambda)}.
\end{equation}

By iterating~\eqref{eqn: q-magnetized partition function external contour representation}, we obtain the advertised contour representation:
\begin{equation}
    \label{eqn: contour representation}
    \frac{\Xi_\bz(\Lambda\mid q)}{\Xi_\bz^q(\Lambda)}
    =\sum_{\Gamma\in\cC(\Lambda,q)}\brackets[\Bigg]{\prod_{\gamma\in\Gamma}\zeta_q^0(\gamma;\bz)}e^{-W_0(\Gamma;\bz)},
\end{equation}
where
\begin{equation}
    \label{eqn: multi-body interaction of contours}
    e^{-W_0(\Gamma;\bz)}:=\frac{\Xi_\bz^q(\Lambda\setminus\supp{\Gamma}\mid R_\Gamma^\ext)\prod_{\gamma\in\Gamma}\Xi_\bz^q(\supp{\gamma})}{\Xi_\bz^q(\Lambda)}
\end{equation}
is a multi-body interaction.

\subsection{Polymer representation}
\label{sec: polymer representation}

In this subsection, we rewrite the contour representation derived in Sec.~\ref{sec: contour representation} as a polymer partition function with hard-core interactions by evaluating the multi-body interaction~\eqref{eqn: multi-body interaction of contours} using the cluster expansion method.

Given $\Gamma\in\cC(\Lambda,q)$, let $P_\Lambda(\Gamma)$ denote the partition of $\Lambda$ into the supports $\supp{\gamma}$ for $\gamma\in\Gamma$ and the connected components of $\Lambda\setminus\supp{\Gamma}$.
Write
\begin{equation}
    e^{-W_0(\Gamma;\bz)}
    =\frac{\Xi_\bz^q(\Lambda\setminus\supp{\Gamma})\prod_{\gamma\in\Gamma}\Xi_\bz^q(\supp{\gamma})}{\Xi_\bz^q(\Lambda)}
    \cdot\frac{\Xi_\bz^q(\Lambda\setminus\supp{\Gamma}\mid R_\Gamma^\ext)}{\Xi_\bz^q(\Lambda\setminus\supp{\Gamma})}.
\end{equation}
Using Thm.~\ref{thm: Ueltschi criterion}, we expand
\begin{equation}
    \label{eqn: patching q-oriented regions}
    \begin{multlined}
        \log\frac{\Xi_\bz^q(\Lambda\setminus\supp{\Gamma})\prod_{\gamma\in\Gamma}\Xi_\bz^q(\supp{\gamma})}{\Xi_\bz^q(\Lambda)}
        \\
        =\sum_{M\in\cM^q(\Lambda\setminus\supp{\Gamma})}\varphi(M)\bz(M)
        +\sum_{\gamma\in\Gamma}\sum_{M\in\cM^q(\supp{\gamma})}\varphi(M)\bz(M)
        -\sum_{M\in\cM^q(\Lambda)}\varphi(M)\bz(M),
    \end{multlined}
\end{equation}
and
\begin{equation}
    \label{eqn: compatibility with rectangles in the outer layers}
    \log\frac{\Xi_\bz^q(\Lambda\setminus\supp{\Gamma}\mid R_\Gamma^\ext)}{\Xi_\bz^q(\Lambda\setminus\supp{\Gamma})}
    =\sum_{\substack{M\in\cM^q(\Lambda\setminus\supp{\Gamma})\\M\sim R_\Gamma^\ext}}\varphi(M)\bz(M)
    -\sum_{M\in\cM^q(\Lambda\setminus\supp{\Gamma})}\varphi(M)\bz(M).
\end{equation}
After cancellation, the only clusters that contribute to~\eqref{eqn: patching q-oriented regions} and~\eqref{eqn: compatibility with rectangles in the outer layers} are, respectively, those in $\cM^q(\Lambda)$ containing at least two rectangles anchored in different components of $P_\Lambda(\Gamma)$, and those confined within a single component of $\Lambda\setminus\supp{\Gamma}$ and incompatible with $R_\Gamma^\ext$.
Together, the contributing clusters are those $M\in\cM^q(\Lambda)$ for which there exists $\gamma\in\Gamma$ such that at least one of the following conditions holds:
\begin{enumerate}
    \item There exist rectangles $r_1,r_2\in M$ such that $r_1$ is anchored in $\supp{\gamma}$ and $r_2$ is anchored in $\Z^2\setminus\supp{\gamma}$.
    \item There exist rectangles $r\in M$ anchored in $\ext(\gamma)$ and $r'\in R_\gamma$ such that $r\not\sim r'$.
\end{enumerate}
Let $\chi_\gamma(M)$ be the indicator function that the pair $(M,\gamma)$ satisfies at least one of the above conditions.
Then,
\begin{equation}
    e^{-W_0(\Gamma;\bz)}=\exp\cbraces[\Bigg]{-\sum_{M\in\cM^q(\Lambda)}\varphi(M)\brackets[\Bigg]{1-\prod_{\gamma\in\Gamma}(1-\chi_\gamma(M))}\bz(M)}.
\end{equation}
After expanding the product and reorganizing, we rewrite the contour representation~\eqref{eqn: contour representation} as
\begin{equation}
    \frac{\Xi_\bz(\Lambda\mid q)}{\Xi_\bz^q(\Lambda)}
    =\sum_{\Gamma\in\cC(\Lambda,q)}\brackets[\Bigg]{\prod_{\gamma\in\Gamma}\zeta_q^{(\Lambda)}(\gamma;\bz)}e^{-W(\Gamma;\bz)},
\end{equation}
where
\begin{equation}
    \label{eqn: finite-volume contour weight}
    \zeta_q^{(\Lambda)}(\gamma;\bz)=\zeta_q^0(\gamma;\bz)\exp\cbraces[\Bigg]{-\sum_{M\in\cM^q(\Lambda)}\varphi(M)\chi_\gamma(M)\bz(M)},
\end{equation}
and
\begin{equation}
    \label{eqn: interaction-polymer factor between contours before partitioning clusters by minimal tile unions}
    e^{-W(\Gamma;\bz)}=\exp\cbraces[\Bigg]{\sum_{M\in\cM^q(\Lambda)}\varphi(M)\brackets[\Bigg]{\sum_{\substack{\Gamma'\subset\Gamma\\\abs{\Gamma'}\ge 2}}(-1)^{\abs{\Gamma'}}\prod_{\gamma\in\Gamma'} \chi_{\gamma}(M)}\bz(M)}.
\end{equation}

Given $Y\in\cT_\ast(\Lambda)$ and $M\in\cM^q(\Lambda)$, let $\pi_{Y}(M)$ be the indicator function that $Y$ is the smallest set in $\cT_\ast(\Lambda)$ (with respect to set containment) containing the \textit{interiors} of all the rectangles in $M$.
Inserting the identity $1=\sum_{Y\in\cT_\ast(\Lambda)}\pi_{Y}(M)$ into~\eqref{eqn: interaction-polymer factor between contours before partitioning clusters by minimal tile unions} and defining
\begin{equation}
    \label{eqn: finite-volume cluster interaction potential}
    V_\Gamma^{(\Lambda)}(Y;\bz):=\sum_{M\in\cM^q(\Lambda)}\varphi(M)\brackets[\Bigg]{\sum_{\substack{\Gamma'\subset\Gamma\\\abs{\Gamma'}\ge 2}}(-1)^{\abs{\Gamma'}}\prod_{\gamma\in\Gamma'} \chi_{\gamma}(M)}\pi_{Y}(M)\bz(M),
\end{equation}
we get that 
\begin{equation}
    e^{-W(\Gamma;\bz)}
    =\prod_{Y\in\cT_\ast(\Lambda)}\brackets*{\left(e^{V_\Gamma^{(\Lambda)}(Y;\bz)}-1\right)+1}
    =\sum_{\cY\in\cB_\ast(\Lambda)}\prod_{Y\in\cY}\left(e^{V_\Gamma^{(\Lambda)}(Y;\bz)}-1\right),
\end{equation}
and so
\begin{equation}
    \frac{\Xi_\bz(\Lambda\mid q)}{\Xi_\bz^q(\Lambda)}
    =\sum_{\Gamma\in\cC(\Lambda,q)}
    \brackets[\Bigg]{\prod_{\gamma\in\Gamma}\zeta_q^{(\Lambda)}(\gamma;\bz)}
    \brackets[\Bigg]{\sum_{\cY\in\cB_\ast(\Lambda)}\prod_{Y\in\cY}\left(e^{V_\Gamma^{(\Lambda)}(Y;\bz)}-1\right)}.
\end{equation}

To arrive at a polymer representation, we group together those $\Gamma\in\cC(\Lambda,q)$ and $\cY\in\cB_\ast(\Lambda)$ which share the same, \textit{connected} union of supports, $\supp{\Gamma}\cup\bigcup_{Y\in\cY}Y=X$.
Thus, a \textbf{polymer} is an element $X\in\cT_\ast(\Lambda)$ with weight given by
\begin{equation}
    \label{eqn: finite-volume polymer weight}
    K_q^{(\Lambda)}(X;\bz)=\sum_{\substack{X_1,X_2\in\cT(\Lambda)\\X_1\cup X_2=X}}
    \sum_{\substack{\Gamma\in\cC(\Lambda,q)\\\supp{\Gamma}=X_1}}
    \brackets[\Bigg]{\prod_{\substack{\gamma\in\Gamma}}\zeta_q^{(\Lambda)}(\gamma;\bz)}
    \brackets[\Bigg]{\sum_{\substack{\cY\in\cB_\ast(\Lambda)\\\cup_{Y\in\cY}Y=X_2}}\prod_{Y\in\cY}\left(e^{V_\Gamma^{(\Lambda)}(Y;\bz)}-1\right)},
\end{equation}
and two polymers are \textbf{compatible} if they are disconnected.
Thus, we obtain the advertised polymer representation:
\begin{equation}
    \label{eqn: polymer representation}
    \frac{\Xi_\bz(\Lambda\mid q)}{\Xi_\bz^q(\Lambda)}
    =\sum_{\cX\subset\cT_\ast(\Lambda)}\brackets[\Bigg]{\prod_{X\in\cX}K_q^{(\Lambda)}(X;\bz)}\brackets[\Bigg]{\prod_{\set*{X,X'}\subset\cX}\indicator{X\sim X'}}.
\end{equation}

\begin{remark}
    \label{rmk: non-one-dimensionality of the interaction between rectangles}
    Except when $\fw=1$, the interaction between rectangles within the $q$-oriented partition function is \textit{not} one-dimensional in the sense that it does \textit{not} factorize into a product of $q$-oriented partition functions over independent one-dimensional segments.
    To unify the treatment of all integer values of $\fw$, we have opted to adapt the general strategy of~\cite{disertori2020plate}, which is not tailored to the one-dimensional case like the one in~\cite{disertori2013nematic}.
\end{remark}

\section{Peierls estimates for the polymer model}
\label{sec: peierls estimates for the polymer model}

In this section, we apply the cluster expansion method to the polymer model derived in Sec.~\ref{sec: contour and polymer representations}.
For the expansion to converge, the polymer weights must satisfy a \textit{Peierls condition}, namely, they must be exponentially small in the volume of the polymer.
To this end, the recursive structure of the polymer representation, which traces back to the definition of the bare contour weight~\eqref{eqn: bare weight of a contour}, makes it natural to pursue an inductive study of the ratios of partition functions appearing in~\eqref{eqn: bare weight of a contour}, as is standard in applications of Pirogov--Sinai theory.
The main result of this section is as follows.

\begin{proposition}
    \label{prop: main}
    Suppose that there exist $z,\nu>0$ and an integer $n\ge 1$ such that $\bz$ takes values in $(e^{-\nu/n}z,e^{\nu/n}z)$, and $\bz\ne z$ for no more than $n$ inputs.
    For suitable $(\fw,\fl,z)$ (see Appx.~\ref{appx: convergence and closure conditions} for the list of conditions), there exist parameters $\tau,\prefactor{eqn: inductive bound on the first-order derivative of the log polymer partition function},\prefactor{eqn: inductive bound on the second-order derivative of the log polymer partition function},\spatial{eqn: inductive bound on the second-order derivative of the log polymer partition function}>0$ such that the following holds.
    Let $q\in\set{-1,1}$, $\fr_1,\fr_2$ be (compatible) rectangles such that $\fd_{12}\ge 6$, $\Lambda$ be a finite union of smoothing squares, and $\Delta$ be such that $\Lambda\setminus\bdin[2\bL]{\Lambda}\subset\Delta\subset\Lambda$. 
    Write $\bz_\Delta:=\indicator{R(\Delta)}\bz$.
    Then, the following estimates hold:
    \begin{equation}
        \label{eqn: inductive bound on the log flipping term}
        \abs{\frac{\Xi_{\bz_\Delta}(\Lambda\mid -q)}{\Xi_{\bz_\Delta}(\Lambda\mid q)}}
        \le e^{\tau\bars{\bdex[\bL]{\Lambda}\cap\bL}/2},
    \end{equation}
    \declareprefactor{eqn: inductive bound on the first-order derivative of the log polymer partition function}
    \begin{equation}
        \label{eqn: inductive bound on the first-order derivative of the log polymer partition function}
        \abs{\cD_1\log\frac{\Xi_{\bz_\Delta}(\Lambda\mid q)}{\Xi_{\bz_\Delta}^q(\Lambda)}}
        \le \indicator{R(\Delta)}(\fr_1)\prefactor{eqn: inductive bound on the first-order derivative of the log polymer partition function},
    \end{equation}
    and
    \declareprefactor{eqn: inductive bound on the second-order derivative of the log polymer partition function}
    \declarespatial{eqn: inductive bound on the second-order derivative of the log polymer partition function}
    \begin{equation}
        \label{eqn: inductive bound on the second-order derivative of the log polymer partition function}
        \abs{\cD_1\cD_2\log\frac{\Xi_{\bz_\Delta}(\Lambda\mid q)}{\Xi_{\bz_\Delta}^q(\Lambda)}}
        \le \indicator{R(\Delta)}(\fr_1,\fr_2)\prefactor{eqn: inductive bound on the second-order derivative of the log polymer partition function}e^{-\spatial{eqn: inductive bound on the second-order derivative of the log polymer partition function}(\fd_{12}-2)}.
    \end{equation}
\end{proposition}

\subsection{Setup for the inductive proof}

\subsubsection{Truncated weights}
\label{sec: truncated weights}

The main tool for studying the ratio of partition functions in~\eqref{eqn: inductive bound on the log flipping term} will be the cluster expansion of the polymer model in $\Lambda$.
This creates a \textit{catch-22 situation} where, to verify the Peierls condition for the polymer weights, we need control of the ratio of partition functions, but to control the ratio of partition functions, we need the Peierls condition for the polymer weights to know that the cluster expansion method applies~\cite[Sec.~10.4.2]{presutti2009scaling}.
To break this circularity, we follow the approach of~\cite[Sec.~10.5.1]{presutti2009scaling} and introduce \textit{truncated} weights, as follows.

If $\Lambda$ is a finite union of smoothing squares, define the \textbf{order} of $\Lambda$ to be the largest integer $n\ge 0$ for which there exist $\set{\gamma_1,\dots,\gamma_n}\subset C(\Lambda,q)$ such that $\supp{\gamma_{i+1}}\subset c(\gamma_i)$ for $1\le i\le n-1$.
Define the \textbf{order} of a contour $\gamma\in C(\Z^2,q)$ to be the integer $n\ge 1$ such that $n-1$ is the maximal order of $\intr_j(\gamma)$ for $1\le j\le h_\gamma$; in particular, the order of $\gamma$ is set to $1$ if $h_\gamma=0$.

Recalling~\eqref{eqn: bare weight of a contour} and~\eqref{eqn: polymer representation}, we define the truncated bare contour weight by
\begin{equation}
    \label{eqn: truncated bare weight of a contour}
    \begin{multlined}
        \hat{\zeta}_q^0(\gamma;\bz):=\frac{(-1)^{\abs{\set*{\psi_\gamma=0}}}\bz(R_\gamma)}{\Xi_\bz^q(\supp{\gamma})}
        \\
        \times \prod_{j=1}^{h_\gamma}\min\set[\Bigg]{\frac{\Xi_\bz(\intr_j(\gamma)\mid m_j(\gamma),R_\gamma)}{\Xi_\bz(\intr_j(\gamma)\mid m_j(\gamma))}\frac{\hat{\Xi}_\bz(\intr_j(\gamma)\mid m_j(\gamma))}{\hat{\Xi}_\bz(\intr_j(\gamma)\mid q)},e^{\tau\bars{\bdex[\bL]{\intr_j(\gamma)}\cap\bL}}},
    \end{multlined}
\end{equation}
and the partition function in $\Lambda$ with $q$-boundary conditions and truncated polymer weights by
\begin{equation}
    \label{eqn: truncated partition function}
    \hat{\Xi}_\bz(\Lambda\mid q):=\Xi_\bz^q(\Lambda)\sum_{\cX\subset\cT_\ast(\Lambda)}\brackets[\Bigg]{\prod_{X\in\cX}\hat{K}_q^{(\Lambda)}(X;\bz)}\brackets[\Bigg]{\prod_{\set*{X,X'}\subset\cX}\indicator{X\sim X'}},
\end{equation}
where $\tau>0$ is a parameter to be chosen, and the truncated polymer weight $\hat{K}_q^{(\Lambda)}$ is constructed by inserting $\hat{\zeta}_q^0$ into~\eqref{eqn: finite-volume contour weight} to obtain the truncated contour weight $\hat{\zeta}_q^{(\Lambda)}$ and then inserting the latter into~\eqref{eqn: finite-volume polymer weight}.
The definitions~\eqref{eqn: truncated bare weight of a contour} and~\eqref{eqn: truncated partition function} should be understood \textit{inductively} in the order of the contours and regions~\cite[Thm.~10.5.1.2]{presutti2009scaling}: regions of order $0$ do not contain any contours, so $\hat{\Xi}_\bz=\Xi_\bz$, while regions of order $n\ge 1$ contain contours whose interiors have order at most $n-1$, which have been defined by induction.

We note the following simple criterion for the recovery of the true weights from the truncated weights, which can also be proved by exploiting the inductive structure of the above definitions; see~\cite[Thm.~10.5.2.1]{presutti2009scaling} for details.

\begin{lemma}[Recovery of true weights]
    \label{lem: recovery of true weights}
    Let $n\ge 0$.
    Suppose that, for each $q\in\set{-1,1}$ and all regions $\Delta$ of order at most $n$, it holds that
    \begin{equation}
        \label{eqn: recovery condition}
        \frac{\hat{\Xi}_\bz(\Delta\mid -q)}{\hat{\Xi}_\bz(\Delta\mid q)}\le e^{\tau\bars{\bdex[\bL]{\Delta}\cap\bL}/2}.
    \end{equation}
    Then, for each $q\in\set{-1,1}$ and all regions $\Lambda$ and contours $\gamma\in C(\Z^2,q)$ of order $n+1$, the truncation is inactive in the sense that $\hat{\Xi}_\bz(\Lambda\mid q)=\Xi_\bz(\Lambda\mid q)$ and $\hat{\zeta}_q^0(\gamma;\bz)=\zeta_q^0(\gamma;\bz)$.
\end{lemma}

Hence, our goal has shifted from proving~\eqref{eqn: inductive bound on the log flipping term} to verifying the recovery condition in Lem.~\ref{lem: recovery of true weights}.

\subsubsection{Infinite-volume weights and potentials}

In the study of the ratio of partition functions in~\eqref{eqn: recovery condition}, it will be convenient to work with polymer weights that do not depend on the specific choice of the region $\Lambda$.
Hence, in analogy to~\eqref{eqn: finite-volume contour weight}, \eqref{eqn: finite-volume cluster interaction potential}, and~\eqref{eqn: finite-volume polymer weight}, we define the infinite-volume contour weight, contour-interaction potential, and polymer weight as
\begin{equation}
    \label{eqn: infinite-volume contour weight}
    \zeta_q(\gamma;\bz):=\zeta_q^0(\gamma;\bz)\exp\cbraces[\Bigg]{-\sum_{M\in\cM^q(\Z^2)}\varphi(M)\chi_\gamma(M)\bz(M)},
\end{equation}
\begin{equation}
    V_\Gamma(Y;\bz):=\sum_{M\in\cM^q(\Z^2)}\varphi(M)\brackets[\Bigg]{\sum_{\substack{\Gamma'\subset\Gamma\\\abs{\Gamma'}\ge 2}}(-1)^{\abs{\Gamma'}}\prod_{\gamma\in\Gamma'} \chi_{\gamma}(M)}\pi_{Y}(M)\bz(M),
\end{equation}
and
\begin{equation}
    \label{eqn: infinite-volume polymer weight}
    K_q(X;\bz):=\sum_{\substack{X_1,X_2\in\cT(\Z^2)\\X_1\cup X_2=X}}
    \sum_{\substack{\Gamma\in\cC(\Z^2,q)\\\supp{\Gamma}=X_1}}
    \brackets[\Bigg]{\prod_{\substack{\gamma\in\Gamma}}\zeta_q(\gamma;\bz)}
    \brackets[\Bigg]{\sum_{\substack{\cY\in\cB_\ast(\Z^2)\\\cup_{Y\in\cY}Y=X_2}}\prod_{Y\in\cY}\left(e^{V_\Gamma(Y;\bz)}-1\right)},
\end{equation}
by replacing each instance of $\Lambda$ by $\Z^2$ in the respective finite-volume definitions.
Truncated infinite-volume contour weights $\hat{\zeta}_q$ and polymer weights $\hat{K}_q$ are defined as in Sec.~\ref{sec: truncated weights} by inserting~\eqref{eqn: truncated bare weight of a contour} into~\eqref{eqn: infinite-volume contour weight} and~\eqref{eqn: infinite-volume polymer weight}, respectively.

\subsubsection{Auxiliary estimates}
\label{sec: auxiliary estimates}

Here, we collect some elementary estimates that will assist in the inductive proof.
Specifically, Lem.~\ref{lem: lower bounds on the m's} will enter in Sec.~\ref{sec: bounds on the interaction-polymer factors}, Lem.~\ref{lem: covering sums} in Sec.~\ref{sec: bounds on the polymer weights}, and Lem.~\ref{lem: polymer sums} will be applied at several places in both Secs.~\ref{sec: bounds on weights and potentials} and~\ref{sec: proof of the main proposition}.

\begin{lemma}
    \label{lem: lower bounds on the m's}
    Let $Y\in\cT_\ast(\Z^2)$.
    \begin{enumerate}[label=(\roman*),ref=\thelemma(\roman*)]
        \item \label{itm: lower bound on m(delta)-1} 
        \declarefree{itm: lower bound on m(delta)-1}
        For $\free{itm: lower bound on m(delta)-1}\in(0,\frac{1}{16})$, $m(Y)-1\ge\free{itm: lower bound on m(delta)-1}\abs{Y\cap\bL}+1+\{(\frac{1}{5}-\free{itm: lower bound on m(delta)-1})\abs{Y\cap\bL}-\frac{11}{5}\}_+$.
        \item \label{itm: lower bound on m_12(delta)-2} 
        \declarefree{itm: lower bound on m_12(delta)-2} 
        Suppose that $\fd_{12}\ge 6$. 
        For $\free{itm: lower bound on m_12(delta)-2}\in(0,\frac{2}{31})$,
        \begin{equation}
            \fm_{12}(Y)-2\ge \free{itm: lower bound on m_12(delta)-2}\abs{Y\cap\bL}+2+\bigg\{\bigg(\frac{1}{5}-\free{itm: lower bound on m_12(delta)-2}\bigg)\abs{Y\cap\bL}-\frac{21}{5}\bigg\}_+.
        \end{equation}
        \item \label{itm: lower bound on m_Gamma^(Lambda)(delta)-1} 
        \declarefree{itm: lower bound on m_Gamma^(Lambda)(delta)-1}
        Let $\Lambda$ be a finite union of smoothing squares and $\Gamma\in\cC(\Lambda,q)$.
        Suppose that $Y\in\cT_\ast(\Lambda)$.
        For $\free{itm: lower bound on m_Gamma^(Lambda)(delta)-1}\in(0,\frac{1}{16})$,
        \begin{equation}
            \begin{multlined}
                m_\Gamma^{(\Lambda)}(Y)-1
                \ge \free{itm: lower bound on m_Gamma^(Lambda)(delta)-1}\abs{Y\cap\bL}
                +\frac{1}{6}(1-16\free{itm: lower bound on m_Gamma^(Lambda)(delta)-1})\rmd_\bL(\supp{\Gamma},\bdin[2\bL]{\Lambda})
                +1
                \\
                +\frac{1}{15}(2+\free{itm: lower bound on m_Gamma^(Lambda)(delta)-1})\{\abs{Y\cap\bL}-16\}_+.
            \end{multlined}
        \end{equation}
    \end{enumerate}
\end{lemma}

\begin{proof}[Proof of~\ref{itm: lower bound on m(delta)-1}]
    If the positive part vanishes, then the RHS is bounded above by $2$ by the assumption that $0<\free{itm: lower bound on m(delta)-1}<\frac{1}{16}$.
    Otherwise, the RHS is bounded by $\frac{1}{5}\abs{Y\cap\bL}-\frac{6}{5}$.
\end{proof}

\begin{proof}[Proof of~\ref{itm: lower bound on m_12(delta)-2}]
    If the positive part vanishes, then the RHS is bounded above by $\frac{21\free{itm: lower bound on m_12(delta)-2}/5}{1/5-\free{itm: lower bound on m_12(delta)-2}}+2<4\le\fd_{12}-2$
    by the assumptions that $\fd_{12}\ge 6$ and $0<\free{itm: lower bound on m_12(delta)-2}<\frac{2}{31}$.
    Otherwise, the RHS is bounded by $\frac{1}{5}\abs{Y\cap\bL}-\frac{11}{5}\le m(Y)-2$.
\end{proof}

\begin{proof}[Proof of~\ref{itm: lower bound on m_Gamma^(Lambda)(delta)-1}]
    Introduce the shorthand notation $y:=\abs{Y\cap\bL}$ and $d:=\rmd_\bL(\supp{\Gamma},\bdin[2\bL]{\Lambda})$, and note that $d\ge 5$.
    If the positive part vanishes, then $y\le 16$, and the RHS is bounded above by
    \begin{equation}
        16\free{itm: lower bound on m_Gamma^(Lambda)(delta)-1}+(1-16\free{itm: lower bound on m_Gamma^(Lambda)(delta)-1})\frac{d}{6}+1
        \le \max\set[\bigg]{1,\frac{d}{6}}+1
        \le \ceil[\bigg]{\frac{d}{2}}-1.
    \end{equation}
    Otherwise, we rewrite the RHS as
    \begin{equation}
        \frac{1-16\free{itm: lower bound on m_Gamma^(Lambda)(delta)-1}}{3}\cdot \frac{d-2}{2}+\frac{2+16\free{itm: lower bound on m_Gamma^(Lambda)(delta)-1}}{3}\cdot\frac{y-6}{5}
        \le\max\set[\bigg]{\frac{d-2}{2},\frac{y-6}{5}}
        \le\max\set[\bigg]{\ceil[\bigg]{\frac{d}{2}},\ceil[\bigg]{\frac{y-1}{5}}}-1,
    \end{equation}
    as required.
\end{proof}

\begin{lemma}
    \label{lem: covering sums}
    Let $X\in\cT(\Z^2)$ and $\vareps_1,\vareps_2>0$.
    Let $X_1,X_2$ be summation variables over $\cT(\Z^2)$.
    \begin{enumerate}[label=(\roman*),ref=\thelemma(\roman*)]
        \item \label{itm: covering sum zeroth moment}
        $\sum_{X_1\cup X_2=X}\vareps_1^{\abs{X_1\cap\bL}}\vareps_2^{\abs{X_2\cap\bL}}=(\vareps_1+\vareps_2+\vareps_1\vareps_2)^{\abs{X\cap\bL}}$.
        \item \label{itm: covering sum first moment} $\sum_{X_1\cup X_2=X}\abs{X_1\cap\bL}\vareps_1^{\abs{X_1\cap\bL}}\vareps_2^{\abs{X_2\cap\bL}}=\abs{X\cap\bL}\vareps_1(1+\vareps_2)(\vareps_1+\vareps_2+\vareps_1\vareps_2)^{\abs{X\cap\bL}-1}$.
        \item \label{itm: covering sum second moment} Suppose that $\vareps_1+\vareps_2+\vareps_1\vareps_2\le 1$.
        Then,
        \begin{equation}
            \sum_{X_1\cup X_2=X}\abs{X_1\cap\bL}^2\vareps_1^{\abs{X_1\cap\bL}}\vareps_2^{\abs{X_2\cap\bL}}\le\abs{X\cap\bL}^2\vareps_1(1+\vareps_2)(\vareps_1+\vareps_2+\vareps_1\vareps_2)^{\abs{X\cap\bL}-2}.
        \end{equation}
    \end{enumerate}
\end{lemma}

\begin{proof}[Proof of~\ref{itm: covering sum zeroth moment}]
    Each tile in $X$ may be covered by only $X_1$, only $X_2$, or both, which correspond to the contributions $\vareps_1$, $\vareps_2$, and $\vareps_1\vareps_2$.
    Multiplying over all tiles in $X$ completes the proof.
\end{proof}

\begin{proof}[Proof of~\ref{itm: covering sum first moment}]
    Inserting the identity $\abs{X_1\cap\bL}=\sum_{v\in X\cap\bL}\indicator{X_1}(v)$ into the LHS, we get
    \begin{equation}
        \sum_{v\in X\cap\bL}\sum_{\substack{X_1\cup X_2=X\\X_1\ni v}}\vareps_1^{\abs{X_1\cap\bL}}\vareps_2^{\abs{X_2\cap\bL}}.
    \end{equation}
    To evaluate the inner sum, observe that $T_v$ may be covered by only $X_1$ or by both $X_1$ and $X_2$, which correspond to the contributions $\vareps_1$ and $\vareps_1\vareps_2$, respectively.
    The remaining tiles in $X$ are treated as in the proof of Lem.~\ref{itm: covering sum zeroth moment}.
\end{proof}

\begin{proof}[Proof of~\ref{itm: covering sum second moment}]
    Inserting the identity 
    \begin{equation}
        \abs{X_1\cap\bL}^2
        =\sum_{v\in X\cap\bL}\indicator{X_1}(v)+2\sum_{\set{v,u}\subset X\cap\bL}\indicator{X_1}(v,u)
    \end{equation}
    into the LHS and following the same reasoning as in the proof of Lem.~\ref{itm: covering sum first moment}, we find that the LHS is equal to
    \begin{equation}
        \abs{X\cap\bL}\vareps_1(1+\vareps_2)(\vareps_1+\vareps_2+\vareps_1\vareps_2)^{\abs{X\cap\bL}-2}\brackets[\Big]{(\vareps_1+\vareps_2+\vareps_1\vareps_2)+(\abs{X\cap\bL}-1)\vareps_1(1+\vareps_2)}.
    \end{equation}
    Using the assumption that $\vareps_1+\vareps_2+\vareps_1\vareps_2\le 1$, we bound the last factor by $\abs{X\cap\bL}$, which completes the proof.
\end{proof}

\begin{lemma}
    \label{lem: polymer sums}
    Let $v\in\bL$, $Y\in\cT(\Z^2)$, $\delta,\xi\in[0,1)$, $a,b,C>0$, and $n,k\ge 0$ be integers.
    Let $X,X_1,X_2$ be summation variables over $\cT_\ast(\Z^2)$.
    Except in~\ref{itm: incompatible polymer sum with piecewise linear decay}, we assume that $\vareps\in[0,\frac{1}{64})$.
    \begin{enumerate}[label=(\roman*),ref=\thelemma(\roman*)]
        \item \label{itm: anchored polymer sum} $\sum_{X\supset T_v}\abs{X\cap\bL}^n \vareps^{\abs{X\cap\bL}}\le c_n(\vareps)$.
        \item \label{itm: incompatible polymer sum} $\sum_{X\not\sim Y}\abs{X\cap\bL}^n \vareps^{\abs{X\cap\bL}}\le 9c_n(\vareps)\abs{Y\cap\bL}$.
        \item \label{itm: polymer sum with distance decay from one point} $\sum_{X}\abs{X\cap\bL}^n \vareps^{\abs{X\cap\bL}}\delta^{\{\ceil{\rmd_\bL(X,T_u)/2}-k\}_+}\le c_n(\vareps)s_k(\delta)$.
        \item \label{itm: polymer sum with distance decay from region} $\sum_{X}\abs{X\cap\bL}^n \vareps^{\abs{X\cap\bL}}\delta^{\rmd_\bL(X,Y)}\le c_n(\vareps)s(\delta)\abs{Y\cap\bL}$.
        \item \label{itm: polymer sum with combined distance decay from two points} 
        \declarefree{itm: polymer sum with combined distance decay from two points}
        For $\free{itm: polymer sum with combined distance decay from two points}\in(0,1)$,
        \begin{equation}
            \sum_{X}\abs{X\cap\bL}^n \vareps^{\abs{X\cap\bL}}\delta^{\fn_{12}(X)-2}\le \delta^{(1-\free{itm: polymer sum with combined distance decay from two points})(\fd_{12}-2)}c_n(\vareps)s_2(\delta^{\free{itm: polymer sum with combined distance decay from two points}}).
        \end{equation}
        \item \label{itm: polymer sum with separate distance decay from two points} 
        \declarefree{itm: polymer sum with separate distance decay from two points} 
        For $\free{itm: polymer sum with separate distance decay from two points}\in(0,1)$ such that $\vareps^{\free{itm: polymer sum with separate distance decay from two points}}<\frac{1}{64}$,
        \begin{equation}
            \sum_{X}\abs{X\cap\bL}^n \vareps^{\abs{X\cap\bL}}\delta^{\fn_1(X)+\fn_2(X)-2}\le \max\set{\vareps,\delta}^{(1-\free{itm: polymer sum with separate distance decay from two points})(\fd_{12}-2)}c_n(\vareps^{\free{itm: polymer sum with separate distance decay from two points}})s_1(\delta^{\free{itm: polymer sum with separate distance decay from two points}}).
        \end{equation}
        \item \label{itm: double polymer sum with distance decay from two points} 
        \declarefree{itm: double polymer sum with distance decay from two points} 
        For $\free{itm: double polymer sum with distance decay from two points}\in(0,1)$ such that $\vareps^{\free{itm: double polymer sum with distance decay from two points}}<\frac{1}{64}$, 
        \begin{equation}
            \begin{multlined}
                \sum_{X_1,X_2}\xi^{\rmd_\bL(X_1,X_2)}\prod_{i=1,2}\abs{X_i\cap\bL}^n \vareps^{\abs{X_i\cap\bL}}\delta^{\fn_i(X_i)-1}
                \\
                \le \max\set{\vareps,\delta,\xi}^{(1-\free{itm: double polymer sum with distance decay from two points})(\fd_{12}-2)}c_n(\vareps^{\free{itm: double polymer sum with distance decay from two points}})^2 s_1(\delta^{\free{itm: double polymer sum with distance decay from two points}})^2.
            \end{multlined}
        \end{equation}
        \item \label{itm: incompatible polymer sum with piecewise linear decay} Suppose that $C\vareps^a\in[0,\frac{1}{64})$.
        Then, $\sum_{X\not\sim Y}C^{\abs{X\cap\bL}}\vareps^{\{a\abs{X\cap\bL}-b\}_+}\le c(\vareps;a,b,C)\abs{Y\cap\bL}$.
    \end{enumerate}
\end{lemma}

\begin{proof}[Proof of Lems.~\ref{itm: anchored polymer sum},~\ref{itm: incompatible polymer sum} and~\ref{itm: incompatible polymer sum with piecewise linear decay}]
    The proofs of these estimates all follow the same strategy: view a polymer as a connected subgraph of $\bL$, and bound the number of polymers $X$ with $\abs{X\cap\bL}=n$ and $X\supset T_v$, for a suitably chosen $v\in\bL$, by the number of walks on $\bL$ of length $2n$ starting from $v$~\cite[Lem.~3.38]{friedli2017statistical}.
    Lem.~\ref{itm: anchored polymer sum} is immediate by this argument. 
    Lems.~\ref{itm: incompatible polymer sum} and~\ref{itm: incompatible polymer sum with piecewise linear decay} follow similarly by replacing the constraint $X\not\sim Y$ with the one that $X$ intersects $Y\cup\partial^{\text{ex}}_\bL Y$ and applying a union bound.
\end{proof}

\begin{proof}[Proof of~\ref{itm: polymer sum with distance decay from one point}]
    Let $T_X^{(u)}$ denote the lexicographically minimum tile in $X$ with minimum $\rmd_\bL$-distance to $T_u$.
    We rewrite and bound the LHS by
    \begin{equation}
        \label{eqn: proof of polymer sum with distance decay from one point}
        \begin{multlined}
            \sum_{v\in\bL}\delta^{\{\ceil{\rmd_\bL(T_v,T_u)/2}-k\}_+}\sum_{\substack{X\in\cT_\ast(\Z^2)\\T_X^{(u)}=T_v}}\abs{X\cap\bL}^n \vareps^{\abs{X\cap\bL}}
            \\
            \le \sum_{v\in\bL}\delta^{\{\ceil{\rmd_\bL(T_v,T_u)/2}-k\}_+}\sum_{\substack{X\in\cT_\ast(\Z^2)\\X\supset T_v}}\abs{X\cap\bL}^n \vareps^{\abs{X\cap\bL}}
        \le s_k(\delta)c_n(\vareps),
        \end{multlined}
    \end{equation}
    having used Lem.~\ref{itm: anchored polymer sum} in the last step.
\end{proof}
    
\begin{proof}[Proof of~\ref{itm: polymer sum with distance decay from region}]
    Let $T_X^{(Y)}$ denote the lexicographically minimum tile in $X$ with minimum $\rmd_\bL$-distance to $Y$.
    We rewrite and bound the LHS by
    \begin{equation}
        \begin{multlined}
            \sum_{v\in\bL}\delta^{\rmd_\bL(T_v,Y)}\sum_{\substack{X\in\cT_\ast(\Z^2)\\T_X^{(Y)}=T_v}}\abs{X\cap\bL}^n\vareps^{\abs{X\cap\bL}}
            \le c_n(\vareps)\sum_{v\in\bL}\delta^{\rmd_\bL(T_v,Y)}
            \\
            \le c_n(\vareps)\sum_{v\in\bL}\sum_{u\in Y\cap\bL}\delta^{\rmd_\bL(T_v,T_u)}
            =c_n(\vareps)s(\delta)\abs{Y\cap\bL},
        \end{multlined}
    \end{equation}
    having replaced the constraint $T_X^{(Y)}=T_v$ by $X\supset T_v$ as in~\eqref{eqn: proof of polymer sum with distance decay from one point} in the first inequality.
\end{proof}

\begin{proof}[Proof of~\ref{itm: polymer sum with combined distance decay from two points}]
    We use
    \begin{equation}
        \fn_{12}(X)-2
        \ge(1-\free{itm: polymer sum with combined distance decay from two points})(\fd_{12}-2)+\free{itm: polymer sum with combined distance decay from two points}\{\fn_1(X)-2\}_+
        =(1-\free{itm: polymer sum with combined distance decay from two points})(\fd_{12}-2)+\free{itm: polymer sum with combined distance decay from two points}\{\ceil{\rmd_\bL(X,\fT_1)/2}-2\}_+
    \end{equation}
    and apply Lem.~\ref{itm: polymer sum with distance decay from one point}.
\end{proof}

\begin{proof}[Proof of~\ref{itm: polymer sum with separate distance decay from two points}]
    Combining the bound $\rmd_\bL(\fT_i,X)\le 2\fn_i(X)$ and the geometric inequality
    \begin{equation}
        \rmd_\bL(\fT_1,\fT_2)\le\rmd_\bL(\fT_1,X)+(\abs{X\cap\bL}-1)+\rmd_\bL(X,\fT_2),
    \end{equation}
    we get, using the subadditivity of the ceiling function,
    \begin{equation}
        \label{eqn: proof of polymer sum with separate distance decay from two points}
        \begin{multlined}
            \fd_{12}-2
            =\ceil{\rmd_\bL(\fT_1,\fT_2)/2}-1\le[\fn_1(X)-1]+\ceil{(\abs{X\cap\bL}-1)/2}+1+[\fn_2(X)-1]
            \\
            \le[\fn_1(X)-1]+\abs{X\cap\bL}+[\fn_2(X)-1].
        \end{multlined}
    \end{equation}
    Subsequently, bounding
    \begin{equation}
        \vareps^{\abs{X\cap\bL}}\delta^{\fn_1(X)+\fn_2(X)-2}
        \le \max\set{\vareps,\delta}^{(1-\free{itm: polymer sum with separate distance decay from two points})(\fd_{12}-2)} \vareps^{\free{itm: polymer sum with separate distance decay from two points}\abs{X\cap\bL}}\delta^{\free{itm: polymer sum with separate distance decay from two points}[\fn_1(X)-1]}
    \end{equation}
    after discarding a factor of $\delta^{\free{itm: polymer sum with separate distance decay from two points}[\fn_2(X)-1]}$, we can apply Lem.~\ref{itm: polymer sum with distance decay from one point} to complete the proof.
\end{proof}

\begin{proof}[Proof of~\ref{itm: double polymer sum with distance decay from two points}]
    Using the geometric inequality
    \begin{equation}
        \rmd_\bL(\fT_1,\fT_2)\le\rmd_\bL(\fT_1,X_1)+(\abs{X_1\cap\bL}-1)+\rmd_\bL(X_1,X_2)+(\abs{X_2\cap\bL}-1)+\rmd_\bL(X_2,\fT_2)
    \end{equation}
    and following the same argument in~\eqref{eqn: proof of polymer sum with separate distance decay from two points}, we get that
    \begin{equation}
        \fd_{12}-2\le [\fn_1(X_1)-1]+\abs{X_1\cap\bL}+\rmd_\bL(X_1,X_2)+\abs{X_2\cap\bL}+[\fn_2(X_2)-1].
    \end{equation}
    Subsequently, bounding
    \begin{equation}
        \xi^{\rmd_\bL(X_1,X_2)}\prod_{i=1,2}\vareps^{\abs{X_i\cap\bL}}\delta^{\fn_i(X_i)-1}
        \le \max\set{\vareps,\delta,\xi}^{(1-\free{itm: double polymer sum with distance decay from two points})(\fd_{12}-2)}\prod_{i=1,2}\vareps^{\free{itm: double polymer sum with distance decay from two points}\abs{X_i\cap\bL}}\delta^{\free{itm: double polymer sum with distance decay from two points}[\fn_i(X_i)-1]}
    \end{equation}
    and applying Lem.~\ref{itm: polymer sum with distance decay from one point} completes the proof.
\end{proof}

\begin{remark}
    \label{rem: reuse of free parameters}
    Each application of an item of Lems.~\ref{lem: lower bounds on the m's} and~\ref{lem: polymer sums} that involves a free parameter ($\theta$) creates an independent instance of that parameter. 
    However, for simplicity, we will not track the separate instances but rather reuse the same value of the parameter across all applications.
    Such a simplification is unlikely to be optimal but is sufficient for our purposes.
\end{remark}

\subsection{Estimates for the weights and potentials}
\label{sec: bounds on weights and potentials}

In this subsection, we derive estimates for the individual components of the polymer weight.
The analysis proceeds component by component, each treated in a separate sub-subsection.
For each component, we first bound its finite- and infinite-volume versions, then the finite-volume correction, and finally the first- and second-order derivatives.
Throughout, we work under the setup of Prop.~\ref{prop: main}.
Moreover, we assume that, in each application of the auxiliary estimates in Sec.~\ref{sec: auxiliary estimates}, the corresponding assumption is verified by the choices of the relevant parameters.

\subsubsection{Estimates for the contour self-potential}

\begin{proposition}
    The following estimates hold for the contour self-potential.
    Except in~\ref{itm: bounds on the finite- and infinite-volume contour self-potentials}, we assume that $\gamma\in C(\Lambda,q)$.
    \begin{enumerate}[label=(\alph*),ref=\theproposition(\alph*)]
        \item \label{itm: bounds on the finite- and infinite-volume contour self-potentials} For $\gamma\in C(\Lambda,q)$, the finite-volume contour self-potential satisfies
        \begin{equation}
            \label{eqn: bound on the finite-volume contour self-potential}
            -\sum_{M\in\cM^q(\Lambda)}\varphi(M)\chi_\gamma(M)\bz_\Delta(M)
            \le 3e^{1+\nu}(z\ell^2)\eps\abs{\supp{\gamma}\cap\bL}.
        \end{equation}
        For $\gamma\in C(\Z^2,q)$, the infinite-volume contour self-potential satisfies
        \begin{equation}
            \label{eqn: bound on the infinite-volume contour self-potential}
            -\sum_{M\in\cM^q(\Z^2)}\varphi(M)\chi_\gamma(M)\bz(M)
            \le 3e^{1+\nu}(z\ell^2)\eps\abs{\supp{\gamma}\cap\bL}.
        \end{equation}
        \item \label{itm: bound on the finite-volume correction to the contour self-potential} The finite-volume correction to the contour self-potential satisfies
        \begin{equation}
            \label{eqn: bound on the finite-volume correction to the contour self-potential}
            \sum_{M\in\cM^q(\Z^2)\setminus\cM^q(\Lambda\setminus\bdin[2\bL]{\Lambda})}\abs{\varphi(M)}\chi_\gamma(M)\bz(M)
            \le 3e^{1+\nu}(z\ell^2)\eps^{\ceil{\rmd_\bL(\supp{\gamma},\bdin[2\bL]{\Lambda})/2}-1}\abs{\supp{\gamma}\cap\bL}.
        \end{equation}
        \item \label{itm: bound on the first-order derivative of the finite-volume contour self-potential} The first-order derivative of the finite-volume contour self-potential satisfies
        \begin{equation}
            \label{eqn: bound on the first-order derivative of the finite-volume contour self-potential}
            \bars[\Bigg]{\cD_1\sum_{M\in\cM^q(\Lambda)}\varphi(M)\chi_\gamma(M)\bz_\Delta(M)}
            \le \indicator{R^q(\Delta)}(\fr_1)e^{1+\nu/n} z\eps^{\fn_1(\gamma)-1}.
        \end{equation}
        \item \label{itm: bound on the second-order derivative of the finite-volume contour self-potential} The second-order derivative of the finite-volume contour self-potential satisfies
        \begin{equation}
            \label{eqn: bound on the second-order derivative of the finite-volume contour self-potential}
            \bars[\Bigg]{\cD_1\cD_2\sum_{M\in\cM^q(\Lambda)}\varphi(M)\chi_\gamma(M)\bz_\Delta(M)}
            \le \indicator{R^q(\Delta)}(\fr_1,\fr_2)e^{3+2\nu/n}z^2\eps^{\fn_{12}(\gamma)-2}.
        \end{equation}
    \end{enumerate}
\end{proposition}

\begin{proof}[Proof of~\ref{itm: bounds on the finite- and infinite-volume contour self-potentials}]
    Starting with the finite-volume contour self-potential, we first discard the nonpositive contribution from single-rectangle clusters ($\abs{M}=1$) to bound
    \begin{equation}
        -\sum_{M\in\cM^q(\Lambda)}\varphi(M)\chi_\gamma(M)\bz_\Delta(M)
        \le\sum_{\substack{M\in\cM^q(\Lambda)\\\abs{M}\ge2}}\abs{\varphi(M)}\chi_\gamma(M)\bz(M).
    \end{equation}
    For the remaining clusters, the condition $\chi_\gamma(M)=1$ forces $M$ to contain either a rectangle anchored in $\supp{\gamma}$ or one anchored sufficiently close to $\supp{\gamma}$ to be incompatible with $R_\gamma$.
    In either case, this rectangle is anchored in $\supp{\gamma}\cup\bdex[2\bL]{\supp{\gamma}}$.
    After substituting $\indicator{M\cap R^q(\supp{\gamma}\cup\bdex[2\bL]{\supp{\gamma}})\ne\emptyset}$ for $\chi_\gamma(M)$, an application of the tail estimate~\eqref{eqn: rectangle cluster tail estimate} yields
    \begin{equation}
        \sum_{\substack{M\in\cM^q(\Lambda)\\\abs{M}\ge2}}\abs{\varphi(M)}\chi_\gamma(M)\bz(M)
        \le \sum_{\substack{M\in\cM^q(\Lambda)\\\abs{M}\ge2\\M\cap R^q(\supp{\gamma}\cup\bdex[2\bL]{\supp{\gamma}})\ne\emptyset}}\abs{\varphi(M)}\bz(M)
        \le 3e^{1+\nu}z\eps\abs{\supp{\gamma}},
    \end{equation}
    where we used the geometric bound
    \begin{equation}
        \abs{\supp{\gamma}\cup\bdex[2\bL]{\supp{\gamma}}}\le3\abs{\supp{\gamma}}
    \end{equation}
    in the last inequality, noting that the external $2\bL$-boundary of each smoothing square consists of $10^2-6^2<2\cdot 6^2$ tiles.
    The same argument applies to the infinite-volume contour self-potential.
\end{proof}

\begin{proof}[Proof of~\ref{itm: bound on the finite-volume correction to the contour self-potential}]
    The bound~\eqref{eqn: bound on the finite-volume correction to the contour self-potential} follows analogously. 
    The condition $\chi_\gamma(M)=1$ restricts the sum to clusters intersecting $R^q(\supp{\gamma}\cup\bdex[2\bL]{\supp{\gamma}})$.
    Combined with the condition $M\in\cM^q(\Z^2)\setminus\cM^q(\Lambda\setminus\bdin[2\bL]{\Lambda})$, an application of Lem.~\ref{lem: lower bound on the size of clusters with anchors in distant tiles} yields the constraint
    \begin{equation}
       \abs{M}
       \ge\ceil{\rmd_\bL(\bdex[2\bL]{\supp{\gamma}},\bdin[2\bL]{\Lambda})/2}+1
       =\ceil{\rmd_\bL(\supp{\gamma},\bdin[2\bL]{\Lambda})/2}.
    \end{equation}
    Applying~\eqref{eqn: rectangle cluster tail estimate} with this constraint and bounding $\abs{\supp{\gamma}\cup\bdex[2\bL]{\supp{\gamma}}}$ as before yields~\eqref{eqn: bound on the finite-volume correction to the contour self-potential}.
\end{proof}

\begin{proof}[Proof of~\ref{itm: bound on the first-order derivative of the finite-volume contour self-potential}]
    After differentiating term-by-term and applying the triangle inequality, it suffices to bound
    \begin{equation}
        \indicator{R^q(\Delta)}(\fr_1)\sum_{\substack{M\in\cM^q(\Delta)\\M\ni \fr_1}}\abs{\varphi(M)}\chi_\gamma(M)n_{\fr_1}(M)\bz(M). 
    \end{equation}
    We extract a cardinality constraint on the contributing clusters $M$ from the conditions $\chi_\gamma(M)=1$ and $M\ni\fr_1$, as follows.
    If $\fr_1\in R^q(\supp{\gamma}\cup\bdex[2\bL]{\supp{\gamma}})$, then we simply retain that $\abs{M}\ge 1$.
    Otherwise, by Lem.~\ref{lem: lower bound on the size of clusters with anchors in distant tiles}, we have that 
    \begin{equation}
        \abs{M}
        \ge\ceil{\rmd_\bL(\fT_1,\bdex[2\bL]{\supp{\gamma}})/2}+1
        =\ceil{\rmd_\bL(\fT_1,\supp{\gamma})/2}.
    \end{equation}
    Combining the two cases, we have that $\abs{M}\ge\fn_1(\gamma)$.
    Applying the tail estimate~\eqref{eqn: rectangle cluster tail estimate} with this constraint yields~\eqref{eqn: bound on the first-order derivative of the finite-volume contour self-potential}.
\end{proof}

\begin{proof}[Proof of~\ref{itm: bound on the second-order derivative of the finite-volume contour self-potential}]
    After differentiating term-by-term and applying the triangle inequality, it suffices to bound
    \begin{equation}
        \indicator{R^q(\Delta)}(\fr_1,\fr_2)\sum_{\substack{M\in\cM^q(\Delta)\\M\supset\mset{\fr_1,\fr_2}}}\abs{\varphi(M)}\chi_\gamma(M)n_{\fr_1}(M)n_{\fr_2}(M)\bz(M). 
    \end{equation}
    As before, the conditions that $M$ contains $\fr_1,\fr_2$ and $\chi_\gamma(M)=1$ imply the constraint $\abs{M}\ge \fn_{12}(\gamma)$.
    Applying~\eqref{eqn: cluster bound with two distinct anchors} with this constraint yields~\eqref{eqn: bound on the second-order derivative of the finite-volume contour self-potential}.
\end{proof}

\subsubsection{Estimates for the contour-interaction potential}

\begin{lemma}
    \label{lem: contour-interaction potential auxiliary lemma}
    Let $\Gamma\in\cC(\Z^2,q)$ and $Y\in\cT_\ast(\Z^2)$.
    Let $T$ be an arbitrary tile in $Y$.
    For all $M\in\cM^q(\Z^2)$,
    \begin{equation}
        \label{eqn: contour-interaction potential auxiliary bound}
        \pi_{Y}(M)\sum_{\substack{\Gamma'\subset\Gamma\\\abs{\Gamma'}\ge 2}}\prod_{\gamma\in\Gamma'} \chi_{\gamma}(M)
        \le \indicator{\rmd_\bL(Y,\supp{\Gamma})\le1}\indicator{M\cap R^q(T\cup\bdex[\bL]{T})\ne\emptyset}\indicator{\abs{M}\ge m(Y)}2^{\abs{Y\cap\bL}}.
    \end{equation}
\end{lemma}

\begin{proof}
    Suppose that $\pi_{Y}(M)=1$; otherwise there is nothing to prove.
    Observe that if $M$ satisfies $\pi_{Y}(M)=1$, then $M$ necessarily intersects $R^q(T\cup\bdex[\bL]{T})$; moreover, under this condition, $\chi_\gamma(M)=1$ only if $\rmd_\bL(Y,\supp{\gamma})\le1$.
    We use this observation to derive a cardinality constraint on $M$, as follows.
    The conditions $\pi_{Y}(M)=1$ and $\prod_{\gamma\in\Gamma'}\indicator{\rmd_\bL(Y,\supp{\gamma})\le1}=1$ for some subset $\Gamma'\subset\Gamma$ with $\abs{\Gamma'}\ge 2$ imply two geometric constraints on $M$.
    First, given distinct contours $\gamma_1,\gamma_2\in\Gamma'$, it holds that $\rmd_\bL(\supp{\gamma_1},\supp{\gamma_2})\ge 6$.
    Since $\rmd_\bL(Y,\supp{\gamma_i})\le1$, $M$ must contain rectangles $r_i$, $i=1,2$, anchored respectively in $\supp{\gamma_i}\cup\bdex[2\bL]{\supp{\gamma_i}}$.
    Using Lem.~\ref{lem: lower bound on the size of clusters with anchors in distant tiles}, we get that $\abs{M}\ge 3$.
    Second, observe that $\fl\le2\ell$ and, by the assumption that $\fl\ge 2\fw-1$ (see Sec.~\ref{sec: model and main result}), $\fw=\ceil{(2\fw-1)/2}\le\ceil{\fl/2}=\ell$.
    Hence, upon labeling the rectangles in $M$ such that $r_i\not\sim\mset{r_1,\dots,r_{i-1}}$ for all $i$, we have that $r_1$ overlaps with at most six tiles, and each subsequent rectangle overlaps with at most five \textit{new} tiles (which do not overlap with any of its predecessors), so $\abs{Y\cap\bL}\le 6+5(\abs{M}-1)$.
    It follows that $\abs{M}\ge m(Y)$.
    
    Now that all the indicator functions on the RHS of~\eqref{eqn: contour-interaction potential auxiliary bound} are accounted for, we bound
    \begin{equation}
        \sum_{\substack{\Gamma'\subset\Gamma\\\abs{\Gamma'}\ge 2}}\prod_{\gamma\in\Gamma'}\indicator{\rmd_\bL(Y,\supp{\gamma})\le1}
        \le 2^{\abs{\set*{\gamma\in\Gamma\mid\rmd_\bL(Y,\supp{\gamma})\le1}}}
        \le 2^{\abs{Y\cap\bL}},
    \end{equation}
    using that each tile is adjacent to at most one contour in $\Gamma$.
\end{proof}

\begin{proposition}
    The following estimates hold for the contour-interaction potential.
    Except in~\ref{itm: bound on the finite- and infinite-volume contour-interaction potentials}, we assume that $\Gamma\in\cC(\Lambda,q)$ and $Y\in\cT_\ast(\Lambda)$.
    \begin{enumerate}[label=(\alph*),ref=\theproposition(\alph*)]
        \item \label{itm: bound on the finite- and infinite-volume contour-interaction potentials} For $\Gamma\in\cC(\Lambda,q)$ and $Y\in\cT_\ast(\Lambda)$, the finite-volume contour-interaction potential satisfies
        \begin{equation}
            \label{eqn: bound on the finite-volume contour-interaction potential}
            \abs{V_\Gamma^{(\Lambda)}(Y;\bz_\Delta)}
            \le \indicator{\rmd_\bL(Y,\supp{\Gamma})\le1}9e^{1+\nu}(z\ell^2) 2^{\abs{Y\cap\bL}}\eps^{m(Y)-1}.
        \end{equation}
        For $\Gamma\in\cC(\Z^2,q)$ and $Y\in\cT_\ast(\Z^2)$, the infinite-volume contour-interaction potential satisfies
        \begin{equation}
            \label{eqn: bound on the infinite-volume contour-interaction potential}
            \abs{V_\Gamma(Y;\bz)}
            \le \indicator{\rmd_\bL(Y,\supp{\Gamma})\le1}9e^{1+\nu}(z\ell^2) 2^{\abs{Y\cap\bL}}\eps^{m(Y)-1}.
        \end{equation}
        \item \label{itm: bound on the finite-volume correction to the contour-interaction potential} The finite-volume correction to the contour-interaction potential satisfies
        \begin{equation}
            \label{eqn: bound on the finite-volume correction to the contour-interaction potential}
            \abs{V_\Gamma^{(\Lambda)}(Y;\bz_\Delta)-V_\Gamma(Y;\bz)}
            \le \indicator{\rmd_\bL(Y,\supp{\Gamma})\le1}9e^{1+\nu}(z\ell^2) 2^{\abs{Y\cap\bL}}\eps^{m_\Gamma^{(\Lambda)}(Y)-1}.
        \end{equation}
        \item \label{itm: bound on the first-order derivative of the finite-volume contour-interaction potential} The first-order derivative of the finite-volume contour-interaction potential satisfies
        \begin{equation}
            \label{eqn: bound on the first-order derivative of the finite-volume contour-interaction potential}
            \abs{\cD_1V_\Gamma^{(\Lambda)}(Y;\bz_\Delta)}
            \le \indicator{R^q(\Delta)\cap R(Y)}(\fr_1)\indicator{\rmd_\bL(Y,\supp{\Gamma})\le1}2^{\abs{Y\cap\bL}}e^{1+\nu}z\eps^{m(Y)-1}.
        \end{equation}
        \item \label{itm: bound on the second-order derivative of the finite-volume contour-interaction potential} The second-order derivative of the finite-volume contour-interaction potential satisfies
        \begin{equation}
            \label{eqn: bound on the second-order derivative of the finite-volume contour-interaction potential}
            \abs{\cD_1\cD_2V_\Gamma^{(\Lambda)}(Y;\bz_\Delta)}
            \le \indicator{R^q(\Delta)\cap R(Y)}(\fr_1,\fr_2)\indicator{\rmd_\bL(Y,\supp{\Gamma})\le1}2^{\abs{Y\cap\bL}}e^{3+\nu}z^2\eps^{\fm_{12}(Y)-2}.
        \end{equation}
    \end{enumerate}
\end{proposition}

\begin{proof}[Proof of~\ref{itm: bound on the finite- and infinite-volume contour-interaction potentials}]
    We first prove the bound~\eqref{eqn: bound on the finite-volume contour-interaction potential} for the finite-volume contour-interaction potential.
    Let $T$ be an arbitrary tile contained in $Y$.
    Recalling the definition of $V_\Gamma^{(\Lambda)}(Y;\bz_\Delta)$ from~\eqref{eqn: finite-volume cluster interaction potential}, we use Lem.~\ref{lem: contour-interaction potential auxiliary lemma} to bound
    \begin{equation}
        \abs{V_\Gamma^{(\Lambda)}(Y;\bz_\Delta)}
        \le\indicator{\rmd_\bL(Y,\supp{\Gamma})\le1}2^{\abs{Y\cap\bL}}\sum_{\substack{M\in\cM^q(\Lambda)\\M\cap R^q(T\cup\bdex[\bL]{T})\ne\emptyset\\\abs{M}\ge m(Y)}}\abs{\varphi(M)}\bz(M).
    \end{equation}
    Applying the tail estimate~\eqref{eqn: rectangle cluster tail estimate} yields~\eqref{eqn: bound on the finite-volume contour-interaction potential} for the finite-volume contour-interaction potential.
    The same argument applies to the infinite-volume contour-interaction potential.
\end{proof}

\begin{proof}[Proof of~\ref{itm: bound on the finite-volume correction to the contour-interaction potential}]
    Under the assumption that $\rmd_\bL(Y,\supp{\Gamma})\le1$, let $T$ be a tile contained in $Y$ realizing this distance.
    Using Lem.~\ref{lem: contour-interaction potential auxiliary lemma}, we estimate the difference
    \begin{equation}
        \abs{V_\Gamma^{(\Lambda)}(Y;\bz_\Delta)-V_\Gamma(Y;\bz)}
        \le\indicator{\rmd_\bL(Y,\supp{\Gamma})\le1}2^{\abs{Y\cap\bL}}\sum_{\substack{M\in\cM^q(\Z^2)\setminus\cM^q(\Lambda\setminus\bdin[2\bL]{\Lambda})\\M\cap R^q(T\cup\bdex[\bL]{T})\ne\emptyset\\\abs{M}\ge m(Y)}}\abs{\varphi(M)}\bz(M).
    \end{equation}
    We now sharpen the constraint on $\abs{M}$.
    Note that the condition $M\in\cM^q(\Z^2)\setminus\cM^q(\Lambda\setminus\bdin[2\bL]{\Lambda})$ dictates that $M$ contain at least one rectangle anchored outside $\Lambda\setminus\bdin[2\bL]{\Lambda}$.
    By Lem.~\ref{lem: lower bound on the size of clusters with anchors in distant tiles}, we have that
    \begin{equation}
        \abs{M}
        \ge\ceil{\rmd_\bL(T\cup\bdex[\bL]{T},\bdin[2\bL]{\Lambda})/2}+1
        \ge\ceil{\rmd_\bL(\supp{\Gamma},\bdin[2\bL]{\Lambda})/2},
    \end{equation}
    having recalled the choice of $T$ in the last inequality.
    Applying the tail estimate~\eqref{eqn: rectangle cluster tail estimate} with the new constraint $\abs{M}\ge m_\Gamma^{(\Lambda)}(Y)$ yields~\eqref{eqn: bound on the finite-volume correction to the contour-interaction potential}.
\end{proof}

\begin{proof}[Proof of~\ref{itm: bound on the first-order derivative of the finite-volume contour-interaction potential}]
    Differentiating term-by-term and applying Lem.~\ref{lem: contour-interaction potential auxiliary lemma}, we bound
    \begin{equation}
        \bars*{\cD_1V_\Gamma^{(\Lambda)}(Y;\bz_\Delta)}
        \le\indicator{R^q(\Delta)\cap R(Y)}(\fr_1)\indicator{\rmd_\bL(Y,\supp{\Gamma})\le1}2^{\abs{Y\cap\bL}}\sum_{\substack{M\in\cM^q(\Delta)\\M\ni \fr_1\\\abs{M}\ge m(Y)}}\abs{\varphi(M)}n_{\fr_1}(M)\bz(M).
    \end{equation}
    Applying~\eqref{eqn: rectangle cluster tail estimate} to the series yields~\eqref{eqn: bound on the first-order derivative of the finite-volume contour-interaction potential}.
\end{proof}

\begin{proof}[Proof of~\ref{itm: bound on the second-order derivative of the finite-volume contour-interaction potential}]
    Differentiating term-by-term and applying Lem.~\ref{lem: contour-interaction potential auxiliary lemma}, we bound
    \begin{equation}
        \begin{multlined}
            \bars{\cD_1\cD_2V_\Gamma^{(\Lambda)}(Y;\bz_\Delta)}
            \\
            \le \indicator{R^q(\Delta)\cap R(Y)}(\fr_1,\fr_2)\indicator{\rmd_\bL(Y,\supp{\Gamma})\le1}2^{\abs{Y\cap\bL}}\sum_{\substack{M\in\cM^q(\Delta)\\M\supset\mset{\fr_1,\fr_2}\\\abs{M}\ge m(Y)}}\abs{\varphi(M)}n_{\fr_1}(M)n_{\fr_2}(M)\bz(M).
        \end{multlined}
    \end{equation}
    Since $M\ni\fr_1,\fr_2$, we can use Lem.~\ref{lem: lower bound on the size of clusters with anchors in distant tiles} to sharpen the cardinality constraint to $\abs{M}\ge \fm_{12}(Y)$.
    Applying~\eqref{eqn: cluster bound with two distinct anchors} to the remaining series yields~\eqref{eqn: bound on the second-order derivative of the finite-volume contour-interaction potential}.
\end{proof}

\subsubsection{Estimates for the interaction-polymer factor}
\label{sec: bounds on the interaction-polymer factors}

\begin{lemma}
    \label{lem: interaction-polymer factor auxiliary computation}
    For $X_2\in\cT(\Z_2)$, define $\cD(X_2):=\set{\cY\in\cB_\ast(\Z^2)\mid\cup_{Y\in\cY}Y=X_2}$.
    Suppose $0<\xi<1$.
    Let $a,h:\cT_\ast(\Z^2)\to[0,\infty)$ be summable functions.
    Then, for any nonnegative functions $b_1,\dots,b_r$ on $\cT_\ast(\Z^2)$ with $r\in\set{0,1,2}$, it holds that
    \begin{equation}
        \begin{multlined}
            \sum_{\cY\in\cD(X_2)}\brackets[\Bigg]{\prod_{Y\in\cY}e^{h(Y)}}\sum_{\substack{Y_1,\dots,Y_r\in\cY\\\text{distinct}}}\brackets[\Bigg]{\prod_{i=1}^r\xi^{\abs{Y_i\cap\bL}}b_i(Y_i)}\brackets[\Bigg]{\prod_{Y\in\cY\setminus\set{Y_1,\dots,Y_r}}\xi^{\abs{Y\cap\bL}}a(Y)}
            \\
            \le \xi^{\abs{X_2\cap\bL}}\exp\cbraces[\Bigg]{\sum_{Y\in\cT_\ast(\Z^2)}a(Y)+\sum_{Y\in\cT_\ast(\Z^2)}h(Y)}\prod_{i=1}^r\sum_{Y\in\cT_\ast(\Z^2)}b_i(Y).
        \end{multlined}
    \end{equation}
\end{lemma}

\begin{proof}
    Using the covering constraint $\cup_{Y\in\cY}Y=X_2$, we may bound the powers of $\xi$ by an independent factor of $\xi^{\abs{X_2\cap\bL}}$.
    Extending the product of $e^{h(Y)}$ to all of $\cT_\ast(\Z^2)$ yields an independent factor of $\exp\cbraces{\sum_{Y\in\cT_\ast(\Z^2)}h(Y)}$.
    After extending the remaining product over $Y$ to all of $\cY$, we can bound the sum of products of the $b_i$ by extending the summation to $Y_1,\dots,Y_r\in\cT_\ast(\Z^2)$ and dropping the distinctness constraint, which yields the independent factor $\prod_{i=1}^r\sum_{Y\in\cT_\ast(\Z^2)}b_i(Y)$.
    It remains to bound
    \begin{equation}
        \sum_{\cY\in\cD(X_2)}\prod_{Y\in\cY}a(Y)
        \le \prod_{Y\in\cT_\ast(\Z^2)}[1+a(Y)]
        \le \prod_{Y\in\cT_\ast(\Z^2)}e^{a(Y)}
        =\exp\cbraces[\Bigg]{\sum_{Y\in\cT_\ast(\Z^2)}a(Y)},
    \end{equation}
    as required.
\end{proof}

\begin{proposition}
    \label{prop: bounds on the interaction-polymer factors}
    The following estimates hold for the interaction-polymer factor.
    Except in~\ref{itm: bound on the finite- and infinite-volume interaction-polymer factors}, we assume that $\Gamma\in\cC(\Lambda,q)$ and $X_2\in\cT(\Lambda)$.
    Throughout, we write $X_1:=\supp{\Gamma}$.
    \begin{enumerate}[label=(\alph*),ref=\theproposition(\alph*)]
        \item \label{itm: bound on the finite- and infinite-volume interaction-polymer factors} 
        \declareloss{eqn: bound on the finite-volume interaction-polymer factor} 
        \declarePeierls{eqn: bound on the finite-volume interaction-polymer factor} 
        For $\Gamma\in\cC(\Lambda,q)$ and $X_2\in\cT(\Lambda)$, the finite-volume interaction-polymer factor satisfies
        \begin{equation}
            \label{eqn: bound on the finite-volume interaction-polymer factor}
            \sum_{\substack{\cY\in\cB_\ast(\Lambda)\\\cup_{Y\in\cY}Y=X_2}}\prod_{Y\in\cY}\abs{e^{V_\Gamma^{(\Lambda)}(Y;\bz_\Delta)}-1}
            \le e^{\loss{eqn: bound on the finite-volume interaction-polymer factor}\abs{X_1\cap\bL}}e^{-\Peierls{eqn: bound on the finite-volume interaction-polymer factor}\abs{X_2\cap\bL}}.
        \end{equation}
        For $\Gamma\in\cC(\Z^2,q)$ and $X_2\in\cT(\Z^2)$, the infinite-volume interaction-polymer factor satisfies
        \begin{equation}
            \label{eqn: bound on the infinite-volume interaction-polymer factor}
            \sum_{\substack{\cY\in\cB_\ast(\Z^2)\\\cup_{Y\in\cY}Y=X_2}}\prod_{Y\in\cY}\abs{e^{V_\Gamma(Y;\bz)}-1}
            \le e^{\loss{eqn: bound on the finite-volume interaction-polymer factor}\abs{X_1\cap\bL}}e^{-\Peierls{eqn: bound on the finite-volume interaction-polymer factor}\abs{X_2\cap\bL}}.
        \end{equation}
        \item \label{itm: bound on the finite-volume correction to the interaction-polymer factors} 
        \declareprefactor{eqn: bound on the finite-volume correction to the interaction-polymer factors} 
        \declareloss{eqn: bound on the finite-volume correction to the interaction-polymer factors} 
        \declarePeierls{eqn: bound on the finite-volume correction to the interaction-polymer factors} 
        \declarespatial{eqn: bound on the finite-volume correction to the interaction-polymer factors} 
        The finite-volume correction to the interaction-polymer factor satisfies
        \begin{equation}
            \label{eqn: bound on the finite-volume correction to the interaction-polymer factors}
            \begin{multlined}
                \sum_{\substack{\cY\in\cB_\ast(\Lambda)\\\cup_{Y\in\cY}Y=X_2}}\bars[\Bigg]{\prod_{Y\in\cY}\left(e^{V_\Gamma^{(\Lambda)}(Y;\bz_\Delta)}-1\right)-\prod_{Y\in\cY}\left(e^{V_\Gamma(Y;\bz)}-1\right)}
                \\
                \le \prefactor{eqn: bound on the finite-volume correction to the interaction-polymer factors}\abs{X_1\cap\bL}e^{\loss{eqn: bound on the finite-volume correction to the interaction-polymer factors}\abs{X_1\cap\bL}}e^{-\Peierls{eqn: bound on the finite-volume correction to the interaction-polymer factors}\abs{X_2\cap\bL}}e^{-\spatial{eqn: bound on the finite-volume correction to the interaction-polymer factors}\rmd_\bL(X_1,\bdin[2\bL]{\Lambda})}.
            \end{multlined}
        \end{equation}
        \item \label{itm: bound on the first-order derivative of the finite-volume interaction-polymer factor} 
        \declareprefactor{eqn: bound on the first-order derivative of the finite-volume interaction-polymer factor}
        The first-order derivative of the finite-volume interaction-polymer factor satisfies
        \begin{equation}
            \label{eqn: bound on the first-order derivative of the finite-volume interaction-polymer factor}
            \bars[\Bigg]{\cD_1\sum_{\substack{\cY\in\cB_\ast(\Lambda)\\\cup_{Y\in\cY}Y=X_2}}\prod_{Y\in\cY}\left(e^{V_\Gamma^{(\Lambda)}(Y;\bz_\Delta)}-1\right)}
            \le\indicator{R^q(\Delta)\cap R(X_2)}(\fr_1)\prefactor{eqn: bound on the first-order derivative of the finite-volume interaction-polymer factor}\abs{X_1\cap\bL}e^{\loss{eqn: bound on the finite-volume interaction-polymer factor}\abs{X_1\cap\bL}}e^{-\Peierls{eqn: bound on the finite-volume interaction-polymer factor}\abs{X_2\cap\bL}}.
        \end{equation}
        \item \label{itm: bound on the second-order derivative of the finite-volume interaction-polymer factor} 
        \declareprefactor{eqn: bound on the second-order derivative of the finite-volume interaction-polymer factor}
		\declarePeierls{eqn: bound on the second-order derivative of the finite-volume interaction-polymer factor}
        The second-order derivative of the finite-volume interaction-polymer factor satisfies
        \begin{equation}
            \label{eqn: bound on the second-order derivative of the finite-volume interaction-polymer factor}
            \begin{multlined}
                \bars[\Bigg]{\cD_1\cD_2\sum_{\substack{\cY\in\cB_\ast(\Lambda)\\\cup_{Y\in\cY}Y=X_2}}\prod_{Y\in\cY}\left(e^{V_\Gamma^{(\Lambda)}(Y;\bz_\Delta)}-1\right)}
                \\
                \le \indicator{R^q(\Delta)\cap R(X_2)}(\fr_1,\fr_2)\prefactor{eqn: bound on the second-order derivative of the finite-volume interaction-polymer factor}\abs{X_1\cap\bL}^2 e^{\loss{eqn: bound on the finite-volume interaction-polymer factor}\abs{X_1\cap\bL}}e^{-\Peierls{eqn: bound on the second-order derivative of the finite-volume interaction-polymer factor}\abs{X_2\cap\bL}}.
            \end{multlined}
        \end{equation}
    \end{enumerate}    
\end{proposition}

\begin{proof}[Proof of~\ref{itm: bound on the finite- and infinite-volume interaction-polymer factors}]
    We first establish the bound~\eqref{eqn: bound on the finite-volume interaction-polymer factor} for the finite-volume interaction-polymer factors.
    By applying the elementary inequality $\abs{e^x-1}\le e^{\abs{x}}\abs{x}$ to the factor $\bars{e^{V_\Gamma^{(\Lambda)}(Y;\bz_\Delta)}-1}$, inserting the bound on $\bars{V_\Gamma^{(\Lambda)}(Y;\bz_\Delta)}$ from Prop.~\ref{itm: bound on the finite- and infinite-volume contour-interaction potentials}, and applying Lem.~\ref{itm: lower bound on m(delta)-1}, we may use Lem.~\ref{lem: interaction-polymer factor auxiliary computation} (with $r=0$) to bound
    \begin{equation}
        \begin{multlined}
            \sum_{\substack{\cY\in\cB_\ast(\Lambda)\\\cup_{Y\in\cY}Y=X_2}}\prod_{Y\in\cY}\abs{e^{V_\Gamma^{(\Lambda)}(Y;\bz_\Delta)}-1}
            \\
            \le \eps^{\free{itm: lower bound on m(delta)-1}\abs{X_2\cap\bL}}\exp\cbraces[\Bigg]{18e^{1+\nu}(z\ell^2) \sum_{\substack{Y\in\cT_\ast(\Z^2)\\\rmd_\bL(Y,X_1)\le1}}2^{\abs{Y\cap\bL}}\eps^{1+\{(\frac{1}{5}-\free{itm: lower bound on m(delta)-1})\abs{Y\cap\bL}-\frac{11}{5}\}_+}}.
        \end{multlined}
    \end{equation}
    Applying Lem.~\ref{itm: incompatible polymer sum with piecewise linear decay} completes the proof for the finite-volume interaction-polymer factor.
    The same proof applies to the infinite-volume interaction-polymer factor.
\end{proof}

\begin{claim}
    \label{clm: difference of products bound}
    For any finite index set $I$ and real numbers $(x_i)_{i\in I}$ and $(y_i)_{i\in I}$,
    \begin{equation}
        \bars[\Bigg]{\prod_{i\in I}(e^{x_i}-1)-\prod_{i\in I}(e^{y_i}-1)}
        \le\prod_{i\in I}e^{2\max\set{\abs{x_i},\abs{y_i}}}\cdot\sum_{i\in I}\abs{x_i-y_i}\prod_{\substack{j\in I\\j\ne i}}\max\set*{\abs{x_j},\abs{y_j}}.
    \end{equation}
\end{claim}

\begin{proof}[Proof of~\ref{itm: bound on the finite-volume correction to the interaction-polymer factors}] 
    Applying Claim~\ref{clm: difference of products bound} to the difference of products in the LHS of~\eqref{eqn: bound on the finite-volume correction to the interaction-polymer factors}, inserting the bounds from Props.~\ref{itm: bound on the finite- and infinite-volume contour-interaction potentials},~\ref{itm: bound on the finite-volume correction to the contour-interaction potential}, Lems.~\ref{itm: lower bound on m(delta)-1},~\ref{itm: lower bound on m_Gamma^(Lambda)(delta)-1},
    and applying Lem.~\ref{lem: interaction-polymer factor auxiliary computation} (with $r=1$) yields
    \begin{align*}
        {}&\sum_{\substack{\cY\in\cB_\ast(\Lambda)\\\cup_{Y\in\cY}Y=X_2}}\bars[\Bigg]{\prod_{Y\in\cY}\left(e^{V_\Gamma^{(\Lambda)}(Y;\bz_\Delta)}-1\right)-\prod_{Y\in\cY}\left(e^{V_\Gamma(Y;\bz)}-1\right)}
        \\
        \le {}&
        \eps^{\min\set{\free{itm: lower bound on m(delta)-1},\free{itm: lower bound on m_Gamma^(Lambda)(delta)-1}}\abs{X_2\cap\bL}}\eps^{\frac{1}{6}(1-16\free{itm: lower bound on m_Gamma^(Lambda)(delta)-1})\rmd_\bL(X_1,\bdin[2\bL]{\Lambda})}
        \\
        \times{}&
        \exp\cbraces[\Bigg]{27e^{1+\nu}(z\ell^2)\sum_{\substack{Y\in\cT_\ast(\Z^2)\\\rmd_\bL(Y,X_1)\le1}}2^{\abs{Y\cap\bL}}\eps^{1+\{(\frac{1}{5}-\free{itm: lower bound on m(delta)-1})\abs{Y\cap\bL}-\frac{11}{5}\}_+}} 
        \\
        \times{}&\sum_{\substack{Y\in\cT_\ast(X_2)\\\rmd_\bL(Y,X_1)\le1}}9e^{1+\nu}(z\ell^2)2^{\abs{Y\cap\bL}}\eps^{1+\frac{1}{15}(2+\free{itm: lower bound on m_Gamma^(Lambda)(delta)-1})\{\abs{Y\cap\bL}-16\}_+}.
        \tag{\stepcounter{equation}\theequation}
    \end{align*}
    Applying Lem.~\ref{itm: incompatible polymer sum with piecewise linear decay} completes the proof of~\eqref{eqn: bound on the finite-volume correction to the interaction-polymer factors}.
\end{proof}

\begin{proof}[Proof of~\ref{itm: bound on the first-order derivative of the finite-volume interaction-polymer factor}]
    Applying Props.~\ref{itm: bound on the finite- and infinite-volume contour-interaction potentials},~\ref{itm: bound on the first-order derivative of the finite-volume contour-interaction potential} and Lem.~\ref{itm: lower bound on m(delta)-1}, we bound using Lem.~\ref{lem: interaction-polymer factor auxiliary computation} (with $r=1$)
    \begin{align*}
        {}&\sum_{\substack{\cY\in\cB_\ast(\Lambda)\\\cup_{Y\in\cY}Y=X_2}}\prod_{Y\in\cY}e^{\abs{V^{(\Lambda)}_\Gamma(Y_1;\bz_\Delta)}}\sum_{Y_1\in\cY}\abs{\cD_1V^{(\Lambda)}_\Gamma(Y_1;\bz_\Delta)}\prod_{Y\in\cY\setminus\set{Y_1}}\abs{V_\Gamma^{(\Lambda)}(Y;\bz_\Delta)}
        \\
        \le {}& 
        \indicator{R^q(\Delta)}(\fr_1)\eps^{\free{itm: lower bound on m(delta)-1}\abs{X_2\cap\bL}}\exp\cbraces[\Bigg]{18e^{1+\nu}(z\ell^2) \sum_{\substack{Y\in\cT_\ast(\Z^2)\\\rmd_\bL(Y,X_1)\le1}}2^{\abs{Y\cap\bL}}\eps^{1+\{(\frac{1}{5}-\free{itm: lower bound on m(delta)-1})\abs{Y\cap\bL}-\frac{11}{5}\}_+}}
        \\
        \times{}& \sum_{\substack{Y\in\cT_\ast(\Z^2)\\\rmd_\bL(Y,X_1)\le1}}\indicator{R(Y)}(\fr_1)2^{\abs{Y\cap\bL}}e^{1+\nu}z\eps^{1+\{(\frac{1}{5}-\free{itm: lower bound on m(delta)-1})\abs{Y\cap\bL}-\frac{11}{5}\}_+}.
        \tag{\stepcounter{equation}\theequation}
    \end{align*}
    Bounding $\indicator{R(Y)}(\fr_1)\le\indicator{R(X_2)}(\fr_1)$ and applying Lem.~\ref{itm: incompatible polymer sum with piecewise linear decay} completes the proof of~\eqref{eqn: bound on the first-order derivative of the finite-volume interaction-polymer factor}.
\end{proof}

\begin{proof}[Proof of~\ref{itm: bound on the second-order derivative of the finite-volume interaction-polymer factor}]
    The second-order derivative is equal to the sum of the following quantities:
    \begin{equation}
        \begin{multlined}
            \sum_{\substack{\cY\in\cB_\ast(\Lambda)\\\cup_{Y\in\cY}Y=X_2}}\sum_{Y_1\in\cY}\brackets[\Bigg]{\cD_1\cD_2V_\Gamma^{(\Lambda)}(Y_1;\bz_\Delta)+\prod_{i=1,2}\cD_i V_\Gamma^{(\Lambda)}(Y_1;\bz_\Delta)}e^{V_\Gamma^{(\Lambda)}(Y_1;\bz_\Delta)}\prod_{Y\in\cY\setminus\set{Y_1}}\left(e^{V_\Gamma^{(\Lambda)}(Y;\bz_\Delta)}-1\right)
            \\
            +\sum_{\substack{\cY\in\cB_\ast(\Lambda)\\\cup_{Y\in\cY}Y=X_2}}\sum_{Y_1\in\cY}
            \sum_{Y_2\in\cY\setminus\set{Y_1}}
            \brackets[\Bigg]{\prod_{i=1,2}\cD_i V_\Gamma^{(\Lambda)}(Y_i;\bz_\Delta)e^{V_\Gamma^{(\Lambda)}(Y_i;\bz_\Delta)}}
            \brackets[\Bigg]{\prod_{Y\in\cY\setminus\set{Y_1,Y_2}}\left(e^{V_\Gamma^{(\Lambda)}(Y;\bz_\Delta)}-1\right)}.
        \end{multlined}
    \end{equation}
    Applying Props.~\ref{itm: bound on the finite- and infinite-volume contour-interaction potentials},~\ref{itm: bound on the first-order derivative of the finite-volume contour-interaction potential} and~\ref{itm: bound on the second-order derivative of the finite-volume contour-interaction potential} and Lems.~\ref{itm: lower bound on m(delta)-1} and~\ref{itm: lower bound on m_12(delta)-2}, we bound the above in absolute value using Lem.~\ref{lem: interaction-polymer factor auxiliary computation} (with $r=1,2$) by
    \begin{align*}
        {}&\indicator{R^q(\Delta)\cap R(X_2)}(\fr_1,\fr_2)e^{2\nu}z^2
        \eps^{\min\set{\free{itm: lower bound on m(delta)-1},\free{itm: lower bound on m_12(delta)-2}}\abs{X_2\cap\bL}}
        \\
        {}&\times \exp\cbraces[\Bigg]{18e^{1+\nu}(z\ell^2) \sum_{\substack{Y\in\cT_\ast(\Z^2)\\\rmd_\bL(Y,X_1)\le1}}2^{\abs{Y\cap\bL}}\eps^{1+\{(\frac{1}{5}-\free{itm: lower bound on m(delta)-1})\abs{Y\cap\bL}-\frac{11}{5}\}_+}}
        \\
        {}&\times\Bigg\{
        \sum_{\substack{Y_1\in\cT_\ast(\Z^2)\\\rmd_\bL(Y_1,X_1)\le1}}
        \left[e^3 2^{\abs{Y_1\cap\bL}}\eps^{2+\{(\frac{1}{5}-\free{itm: lower bound on m_12(delta)-2})\abs{Y_1\cap\bL}-\frac{21}{5}\}_+}
        +e^2 2^{2\abs{Y_1\cap\bL}}\eps^{2+2\{(\frac{1}{5}-\free{itm: lower bound on m(delta)-1})\abs{Y_1\cap\bL}-\frac{11}{5}\}_+}\right]
        \\
        {}&+\brackets[\Bigg]{\sum_{\substack{Y_1\in\cT_\ast(\Z^2)\\\rmd_\bL(Y_1,X_1)\le1}}2^{\abs{Y_1\cap\bL}}e\eps^{1+\{(\frac{1}{5}-\free{itm: lower bound on m(delta)-1})\abs{Y_1\cap\bL}-\frac{11}{5}\}_+}}^2
        \Bigg\}.
        \tag{\stepcounter{equation}\theequation}
    \end{align*}
    Applying Lem.~\ref{itm: incompatible polymer sum with piecewise linear decay} completes the proof of~\eqref{eqn: bound on the second-order derivative of the finite-volume interaction-polymer factor}.
\end{proof}

\subsubsection{Estimates for the contour weight}

\begin{lemma}
    \label{lem: Peierls bound for contour weight}
    Suppose that $\Peierls{eqn: Peierls bound for contour weight}>0$.
    Let $\gamma\in C(\Z^2,q)$ and $\fR\in\Omega(\Z^2)$ be finite.
    Then,
    \declarePeierls{eqn: Peierls bound for contour weight}
    \begin{equation}
        \label{eqn: Peierls bound for contour weight}
        \frac{1}{\Xi_\bz^q(\supp{\gamma})}\sum_{\substack{R_\gamma\in\Omega(\supp{\gamma},\psi_\gamma)\\R_\gamma\supset\fR}}\bz(R_\gamma)
        \le z^{\abs{\fR}}e^{-\Peierls{eqn: Peierls bound for contour weight}\abs{\supp{\gamma}\cap\bL}}.
    \end{equation}
\end{lemma}

\begin{proof}
    We start with two simplifications.
    First, recalling the assumption on the range of $\bz$ from the statement of Prop.~\ref{prop: main}, we pass to constant fugacity by bounding 
    \begin{equation}
        \frac{1}{\Xi_\bz^q(\supp{\gamma})}\sum_{\substack{R_\gamma\in\Omega(\supp{\gamma},\psi_\gamma)\\R_\gamma\supset\fR}}\bz(R_\gamma)
        \le \frac{e^{2\nu}}{\Xi_z^q(\supp{\gamma})}\sum_{\substack{R_\gamma\in\Omega(\supp{\gamma},\psi_\gamma)\\R_\gamma\supset\fR}} z^{\bars{R_\gamma}}.
    \end{equation}
    Second, we remove the constraint $R_\Gamma\supset\fR$ via the following observation: If $R_\gamma\in\Omega(\supp{\gamma},\psi_\gamma)$ and $R_\gamma\supset\fR$, then $R_\gamma\setminus\fR\in\Omega(\supp{\gamma},\psi_\gamma)$.
    Hence,
    \begin{equation}
        \sum_{\substack{R_\gamma\in\Omega(\supp{\gamma},\psi_\gamma)\\R_\gamma\supset\fR}}z^{\bars{R_\gamma}}
        \le z^{\bars{\fR}} \sum_{\substack{R_\gamma\in\Omega(\supp{\gamma},\psi_\gamma)\\R_\gamma\cap\fR=\emptyset}}z^{\bars{R_\gamma}}
        \le z^{\bars{\fR}} \sum_{R_\gamma\in\Omega(\supp{\gamma},\psi_\gamma)}z^{\bars{R_\gamma}}.
    \end{equation}

    To proceed further, we rely on the following claim.

    \begin{claim}
        There exists a partition $P$ of $\supp{\gamma}$ into tiles and unions of pairs of \textit{contiguous} tiles (i.e., dominoes) with the following properties:
        \begin{enumerate}
            \item Every tile with $\psi_\gamma=0$ is its own tile in $P$. \label{itm: partition property empty tiles}
            \item Every domino in $P$ consists of tiles with opposite, nonzero $\psi_\gamma$. \label{itm: partition property dominoes}
            \item Every domino in $\supp{\gamma}$ consisting of tiles with opposite, nonzero $\psi_\gamma$ intersects a domino in $P$. \label{itm: partition property maximality}
        \end{enumerate}
        Define $n_0(\gamma)$ as the number of tiles in $\supp{\gamma}$ with $\psi_\gamma=0$ and $n_\pm(\gamma)$ as the \textit{maximum} number of dominoes in $P$ among all such partitions $P$.
        Then, $n_0(\gamma)+\frac{4}{3}n_\pm(\gamma)\ge \abs{\supp{\gamma}\cap\bL}/324$.
    \end{claim}

    \begin{proof}[Proof of the claim]
        We construct such a partition $P$ as follows.
        First, we add to $P$ as many \textit{disjoint} dominoes consisting of tiles with opposite, nonzero $\psi_\gamma$ as possible.
        Then, we add to $P$ all the remaining tiles in $\supp{\gamma}$ as singletons.
        By construction, $P$ satisfies all the required properties.

        It remains to verify the bound on $n_0(\gamma)+\frac{4}{3}n_\pm(\gamma)$.
        Let $(\psi,R)$ be a profile-configuration pair from which $\gamma$ was constructed and $S$ be a smoothing square in $\supp{\gamma}$.
        If any tile $T'\subset S\cup\bdex[6\bL]{S}$ has $\psi=0$, then its enclosing smoothing square $S'$ is incorrect with respect to $(\psi,R)$, but $\rmd_{6\bL}(S',S)\le 1$, so $S'\subset\supp{\gamma}$.
        Otherwise, no tile in $S\cup\bdex[6\bL]{S}$ has $\psi=0$.
        In this case, for $S$ to be incorrect with respect to $(\psi,R)$, there must exist a domino contained in $S\cup\bdex[6\bL]{S}$ consisting of tiles with opposite, nonzero $\psi$.
        By Itm.~\ref{itm: partition property maximality}, this domino intersects a domino in $P$.
        The latter domino thus intersects $S\cup\bdex[6\bL]{S}$, and any smoothing square $S'\subset S\cup\bdex[6\bL]{S}$ it intersects is incorrect with respect to $(\psi,R)$ and hence contained in $\supp{\gamma}$.
        Now, since each tile belongs to $S\cup\bdex[6\bL]{S}$ for at most $9$ smoothing squares $S$, each domino intersects $S\cup\bdex[6\bL]{S}$ for at most $12$ smoothing squares $S$, and each smoothing square contains $36$ tiles, we conclude that
        \begin{equation}
            \abs{\supp{\gamma}\cap\bL}/36\le 9n_0(\gamma)+12n_\pm(\gamma),
        \end{equation}
        from which the desired bound follows.
    \end{proof}

    Let $P$ be a partition of $\supp{\gamma}$ guaranteed by the above claim which contains $n_\pm(\gamma)$ dominoes.
    We decouple the sets in $P$ by neglecting the hard-core interaction across their boundaries:
    \begin{equation}
        \sum_{R_\gamma\in\Omega(\supp{\gamma},\psi_\gamma)}z^{\bars{R_\gamma}}
        \le\prod_{Y\in P}\sum_{R_Y\in\Omega(Y,\psi_\gamma)}z^{\bars{R_Y}}.
    \end{equation}
    On the other hand, by Lem.~\ref{itm: expansion of the log q-oriented partition function}, we get that
    \begin{equation}
        \Xi_z^q(\supp{\gamma})\ge e^{-2ez\eps\abs{\supp{\gamma}}}\prod_{Y\in P}\Xi_z^q(Y).
    \end{equation}
    Hence,
    \begin{equation}
        \label{eqn: Peierls bound for contour weight after decoupling}
        \frac{1}{\Xi_z^q(\supp{\gamma})}\sum_{R_\gamma\in\Omega(\supp{\gamma},\psi_\gamma)}z^{\bars{R_\gamma}}
        \le e^{2ez\eps\abs{\supp{\gamma}}}\prod_{Y\in P}\frac{1}{\Xi_z^q(Y)}\sum_{R_Y\in\Omega(Y,\psi_\gamma)}z^{\bars{R_Y}}.
    \end{equation}

    Suppose first that $Y$ is a single tile.
    If $\psi_\gamma(Y)\ne 0$, then
    \begin{equation}
        \label{eqn: Peierls bound for contour weight - ground state tile}
        \sum_{R_Y\in\Omega(Y,\psi_\gamma)}z^{\bars{R_Y}}=\Xi_z^{\psi_\gamma(Y)}(Y)=\Xi_z^q(Y).
    \end{equation}
    Otherwise, $\psi_\gamma(Y)=0$, and, using Lem.~\ref{itm: expansion of the log q-oriented partition function}, we have that
    \begin{equation}
        \label{eqn: Peierls bound for contour weight - empty tile}
        \frac{1}{\Xi_z^q(Y)}\sum_{R_Y\in\Omega(Y,\psi_\gamma)}z^{\bars{R_Y}}
        =\frac{1}{\Xi_z^q(Y)}
        \le e^{-(1-e\eps)(z\ell^2)}.
    \end{equation}

    Suppose now that $Y$ is a domino.
    Without loss of generality, we may assume that the bottomleft corner of $Y$ is at the origin.
    Let $t\in\set{x,y}$ be the coordinate direction normal to the common edge of the two tiles that make up $Y$, which we denote by $Y_1$ and $Y_2$, ordered in the direction of increasing $t$.
    Write $q_i:=\psi_\gamma(Y_i)$ for $i=1,2$.
    Define
    \begin{equation}
        I_1:=Y_1\cap\set[\bigg]{t\ge\ell-\ceil[\bigg]{\frac{\ell}{2}}},\quad
        I_2:=Y_2\cap\set[\bigg]{t\le\ell+\ceil[\bigg]{\frac{\ell}{2}}-1},
    \end{equation}
    and, as will be useful when $\ell$ is odd,
    \begin{equation}
        J_1:=Y_1\cap\set[\bigg]{t\ge\ell-\ceil[\bigg]{\frac{\ell}{2}}+1},\quad
        J_2:=Y_2\cap\set[\bigg]{t\le\ell+\ceil[\bigg]{\frac{\ell}{2}}-2}.
    \end{equation}
    Denote the coordinate direction perpendicular to $t$ by $t'$.

    \begin{claim}[Excluded volume effect]
        If $R_Y\in\Omega(Y,\psi_\gamma)$, then $I_1$ is empty of anchors (i.e., $R_Y\cap R(I_1)=\emptyset$), or $I_2$ is empty of anchors, or, when $\ell$ is odd, $J_1\cup J_2$ is empty of anchors.
    \end{claim}

    \begin{proof}[Proof of the claim]
        When $\ell$ is even, the $t$-coordinates of a vertex in $I_1$ and a vertex in $I_2$ differ by at most $\ell-1$ by construction, while the dimensions of the tiles ensure that their $t'$-coordinates, again, differ by at most $\ell-1$.
        By Lem.~\ref{lem: one orientation per tile}, either $I_1$ or $I_2$ must be empty of anchors.

        When $\ell$ is odd, suppose that $I_1$ and $I_2$ both contain anchors of $R_Y$.
        Noting that the $t$-coordinates of a vertex in $I_1$ (resp. $I_2$) and a vertex in $J_2$ (resp. $J_1$) differ by at most $\ell-1$, we have that $J_1\cup J_2$ is empty of anchors by Lem.~\ref{lem: one orientation per tile}.
    \end{proof}

    Decoupling $Y_1$ and $Y_2$ as before yields
    \begin{equation}
        \Xi_z^q(Y)\ge e^{-2ez\eps\abs{Y}}\Xi_z^{q_1}(Y_1)\Xi_z^{q_2}(Y_2).
    \end{equation}
    On the other hand, by the claim, we have that
    \begin{equation}
        \sum_{R_Y\in\Omega(Y,\psi_\gamma)}z^{\bars{R_Y}}
        \le \Xi_z^{q_1}(Y_1\setminus I_1)\Xi_z^{q_2}(Y_2)
        +\Xi_z^{q_1}(Y_1)\Xi_z^{q_2}(Y_2\setminus I_2)
        +\indicator{\ell\text{ odd}} \Xi_z^{q_1}(Y_1\setminus J_1)\Xi_z^{q_2}(Y_2\setminus J_2).
    \end{equation}
    Note that $\bars{R^{q_i}(Y_i\setminus I_i)}=\ell\floor{\ell/2}$ and, when $\ell$ is odd, $\bars{R^{q_i}(Y_i\setminus J_i)}=\ell(\floor{\ell/2}+1)=\ell(\ell+1)/2$, uniformly in $i=1,2$ and $q_i\in\set{-1,1}$.
    Hence, by Lem.~\ref{itm: expansion of the log q-oriented partition function}, we get that
    \begin{equation}
        \max\set*{
            \frac{\Xi_z^{q_1}(Y_1\setminus I_1)}{\Xi_z^{q_1}(Y_1)},
            \frac{\Xi_z^{q_2}(Y_2\setminus I_2)}{\Xi_z^{q_2}(Y_2)}
        }
        \le \frac{e^{(1+e\eps)z\ell\floor{\ell/2}}}{e^{-(1-e\eps)z\ell^2}}
        \le e^{-(z\ell^2)/2+3e(z\ell^2)\eps/2},
    \end{equation}
    and, when $\ell$ is odd,
    \begin{equation}
        \max\set*{
            \frac{\Xi_z^{q_1}(Y_1\setminus J_1)}{\Xi_z^{q_1}(Y_1)},
            \frac{\Xi_z^{q_2}(Y_2\setminus J_2)}{\Xi_z^{q_2}(Y_2)}
        }
        \le \frac{e^{(1+e\eps)z\ell(\ell+1)/2}}{e^{-(1-e\eps)z\ell^2}}
        = e^{-z\ell(\ell-1)/2+ez\ell(3\ell+1)\eps/2}.
    \end{equation}
    Therefore, using that $\ell(\ell-1)\ge\ell^2/2$ for $\fl\ge 3$ (see Sec.~\ref{sec: model and main result}) and bounding $3\ell+1\le 4\ell$, we get
    \begin{align}
        {}&\frac{1}{\Xi_z^q(Y)}\sum_{R_Y\in\Omega(Y,\psi_\gamma)}z^{\bars{R_Y}}
        \nonumber
        \\
        \le{}& e^{2ez\eps\abs{Y}}\brackets[\bigg]{
            \frac{\Xi_z^{q_1}(Y_1\setminus I_1)}{\Xi_z^{q_1}(Y_1)}
            +\frac{\Xi_z^{q_2}(Y_2\setminus I_2)}{\Xi_z^{q_2}(Y_2)}
            +\indicator{\ell\text{ odd}}\frac{\Xi_z^{q_1}(Y_1\setminus J_1)\Xi_z^{q_2}(Y_2\setminus J_2)}{\Xi_z^{q_1}(Y_1)\Xi_z^{q_2}(Y_2)}
        }
        \nonumber
        \\
        \le{}& e^{4e(z\ell^2)\eps}\brackets[\big]{
            2e^{-(z\ell^2)/2+3e(z\ell^2)\eps/2}
            +e^{-z\ell(\ell-1)+ez\ell(3\ell+1)\eps}
        }
        \nonumber
        \\
        \le{}& 3e^{-(1-16e\eps)(z\ell^2)/2}.
        \label{eqn: Peierls bound for contour weight - domino}
    \end{align}

    Inserting the bounds~\eqref{eqn: Peierls bound for contour weight - ground state tile}, \eqref{eqn: Peierls bound for contour weight - empty tile}, and~\eqref{eqn: Peierls bound for contour weight - domino} into~\eqref{eqn: Peierls bound for contour weight after decoupling}, we get that
    \begin{equation}
        \begin{multlined}
            \frac{1}{\Xi_z^q(\supp{\gamma})}\sum_{R_\gamma\in\Omega(\supp{\gamma},\psi_\gamma)}z^{\bars{R_\gamma}}
            \\
            \le e^{2e(z\ell^2)\eps\abs{\supp{\gamma}\cap\bL}}
            \brackets[\Big]{e^{-(1-e\eps)(z\ell^2)}}^{n_0(\gamma)}
            \brackets[\Big]{3^{3/4}e^{-3(1-16e\eps)(z\ell^2)/8}}^{\frac{4}{3}n_{\pm}(\gamma)}.
        \end{multlined}
    \end{equation}
    The second bracketed quantity dominates the first.
    Recalling the assumption that $\Peierls{eqn: Peierls bound for contour weight}>0$ and using that $n_0(\gamma)+\frac{4}{3}n_\pm(\gamma)\ge \abs{\supp{\gamma}\cap\bL}/324$, we conclude the proof of~\eqref{eqn: Peierls bound for contour weight}.
\end{proof}

Recall the notation $c(\gamma)$ introduced in Sec.~\ref{sec: contours} for the support of a contour $\gamma$. 
In the following, we extend its definition to all (bounded) sets $X_1\in\cT(\Z^2)$ by noting that $\R^2\setminus X_1$ has a unique unbounded component $\ext(X_1)$, and defining $c(X_1):=\R^2\setminus\ext(X_1)$.

\begin{lemma}
    \label{lem: bound on the first-order derivative of the contour weight}
    Suppose that for all $\Omega$ of order less than that of $\Lambda$ and all $\Omega\setminus\bdin[2\bL]{\Omega}\subset\Pi\subset\Omega$, with $\bz_\Pi:=\indicator{R(\Pi)}\bz$, it holds that 
    \begin{equation}
        \label{eqn: inductive condition for the bound on the first-order derivative of the contour weight}
        \bars*{\cD_1\log\frac{\Xi_{\bz_\Pi}(\Omega\mid q)}{\Xi_{\bz_\Pi}^q(\Omega)}}\le \indicator{\Pi}(\fr_1)\prefactor{eqn: inductive bound on the first-order derivative of the log polymer partition function}.
    \end{equation}  
    Then, for all $\Gamma\in\cC(\Lambda,q)$, writing $X_1:=\supp{\Gamma}$, the first-order derivative of the log contour weight satisfies
    \declareprefactor{eqn: bound on the first-order derivative of the contour weight - localization}
    \declareprefactor{eqn: bound on the first-order derivative of the contour weight - decay}
    \begin{equation}
        \label{eqn: bound on the first-order derivative of the contour weight}
        \bars[\Bigg]{\sum_{\gamma\in\Gamma}\cD_1\log\zeta_q^{(\Lambda)}(\gamma;\bz_\Delta)}
        \le \indicator{R_\Gamma}(\fr_1)
        +\brackets[\Big]{\indicator{R(c(X_1))}(\fr_1)\prefactor{eqn: bound on the first-order derivative of the contour weight - localization}
        +\indicator{R^q(\Delta)}(\fr_1)\prefactor{eqn: bound on the first-order derivative of the contour weight - decay}\eps^{\fn_1(X_1)-1}}\abs{X_1\cap\bL}.
    \end{equation}
\end{lemma}

\begin{proof}
    For $\gamma\in C(\Lambda,q)$, we write $\cD_1\log\zeta_q^{(\Lambda)}(\gamma;\bz_\Delta)$ as the sum of the following quantities:
    \begin{align}
        A_q^{1}(\gamma;\bz){}&:=\cD_1\log\frac{\bz(R_\gamma)}{\Xi_{\bz}^q(\supp{\gamma})},
        \\
        A_q^{2}(\gamma;\bz){}&:=\sum_{j=1}^{h_\gamma}\cD_1\log\frac{\Xi_{\bz}(\intr_j(\gamma)\mid m_j(\gamma),R_\gamma)}{\Xi_{\bz}(\intr_j(\gamma)\mid q)},
        \\
        A_q^{3}(\gamma;\bz_\Delta){}&:=-\cD_1\sum_{M\in\cM^q(\Lambda)}\varphi(M)\chi_\gamma(M)\bz_\Delta(M),
    \end{align}
    noting that we can freely interchange $\bz$ with $\bz_\Delta$ in the first two.
    Using Lem.~\ref{itm: expansion of the first-order derivative of the log q-oriented partition function}, we bound
    \begin{equation}
        \bars*{A^1_q(\gamma;\bz)}\le \indicator{R_\gamma}(\fr_1)+\indicator{R^q(\supp{\gamma})}(\fr_1)(1+e\eps)e^{\nu/n} z.
    \end{equation}
    Using the assumption~\eqref{eqn: inductive condition for the bound on the first-order derivative of the contour weight}, and denoting by $V(R_\gamma)$ the excluded volume due to $R_\gamma$, we bound
    \begin{equation}
        \begin{multlined}
            \bars*{A^2_q(\gamma;\bz)}
            \le\sum_{j=1}^{h_\gamma}\Bigg[
            \bars*{\cD_1\log\frac{\Xi_\bz(\intr_j(\gamma)\mid m_j(\gamma),R_\gamma)}{\Xi_\bz^{m_j(\gamma)}(\intr_j(\gamma)\mid R_\gamma)}}
            +\bars*{\cD_1\log\frac{\Xi_\bz(\intr_j(\gamma)\mid q)}{\Xi_\bz^q(\intr_j(\gamma))}}
            \\
            +\bars[\Bigg]{\cD_1\log\frac{\Xi_\bz^{m_j(\gamma)}(\intr_j(\gamma)\mid R_\gamma)}{\Xi_\bz^q(\intr_j(\gamma))}}\Bigg]
            \le \indicator{R(\intr(\gamma))}(\fr_1)2\brackets[\Big]{\prefactor{eqn: inductive bound on the first-order derivative of the log polymer partition function}+(1+e\eps)e^{\nu/n}z}.
        \end{multlined}
    \end{equation}
    We also use Prop.~\ref{itm: bound on the first-order derivative of the finite-volume contour self-potential} to bound $\abs{A^3_q(\gamma;\bz_\Delta)}\le \indicator{R^q(\Delta)}(\fr_1)e^{1+\nu/n}z\eps^{\fn_1(\gamma)-1}$.
    Combining these estimates and bounding $\abs{\Gamma}\le\abs{X_1\cap\bL}$ yields~\eqref{eqn: bound on the first-order derivative of the contour weight}.
\end{proof}

\begin{lemma}
    \label{lem: bound on the second-order derivative of the contour weight}
    Suppose that for all $\Omega$ of order less than that of $\Lambda$ and all $\Omega\setminus\bdin[2\bL]{\Omega}\subset\Pi\subset\Omega$, with $\bz_\Pi:=\indicator{R(\Pi)}\bz$, it holds that 
    \begin{equation}
        \label{eqn: inductive condition for the bound on the second-order derivative of the contour weight}
        \bars*{\cD_1\cD_2\log\frac{\Xi_{\bz_\Pi}(\Omega\mid q)}{\Xi_{\bz_\Pi}^q(\Omega)}}\le \indicator{R(\Pi)}(\fr_1,\fr_2)\prefactor{eqn: inductive bound on the second-order derivative of the log polymer partition function}e^{-\spatial{eqn: inductive bound on the second-order derivative of the log polymer partition function}(\fd_{12}-2)}.
    \end{equation}  
    Then, for all $\Gamma\in\cC(\Lambda,q)$, writing $X_1:=\supp{\Gamma}$, the second-order derivative of the log contour weight satisfies
    \declareprefactor{eqn: bound on the second-order derivative of the contour weight}
    \begin{equation}
        \label{eqn: bound on the second-order derivative of the contour weight}
            \bars[\Bigg]{\sum_{\gamma\in\Gamma}\cD_1\cD_2\log\zeta_q^{(\Lambda)}(\gamma;\bz_\Delta)}
            \le \prefactor{eqn: bound on the second-order derivative of the contour weight}\abs{X_1\cap\bL}\brackets*{\indicator{R(c(X_1))}(\fr_1,\fr_2)
            +\indicator{R^q(\Delta)}(\fr_1,\fr_2)\eps^{\fn_{12}(X_1)-2}}.
    \end{equation}
\end{lemma}

\begin{proof}
    Write $\cD_1\cD_2\log\zeta_q^{(\Lambda)}(\gamma;\bz_\Delta)$ as the sum of the following quantities:
    \begin{align}
        B_q^{1}(\gamma;\bz){}&:=\cD_1\cD_2\log\frac{\bz(R_\gamma)}{\Xi_\bz^q(\supp{\gamma})},
        \\
        B_q^{2}(\gamma;\bz){}&:=\sum_{j=1}^{h_\gamma}\cD_1\cD_2\log\frac{\Xi_\bz(\intr_j(\gamma)\mid m_j(\gamma),R_\gamma)}{\Xi_\bz(\intr_j(\gamma)\mid q)},
        \\
        B_q^{3}(\gamma;\bz_\Delta){}&:=-\cD_1\cD_2\sum_{M\in\cM^q(\Lambda)}\varphi(M)\chi_\gamma(M)\bz_\Delta(M),
    \end{align}
    noting, as in the proof of Lem.~\ref{lem: bound on the first-order derivative of the contour weight}, that we can freely interchange $\bz$ with $\bz_\Delta$ in the first two.
    Using Lem.~\ref{itm: expansion of the second-order derivative of the log q-oriented partition function}, we bound
    \begin{equation}
        \bars*{B^1_q(\gamma;\bz)}\le \indicator{R^q(\supp{\gamma})}(\fr_1,\fr_2)e^{3+2\nu/n}z^2\eps^{\fd_{12}-2}.
    \end{equation}
    Denoting by $V(R_\gamma)$ the excluded volume due to $R_\gamma$, we bound using the assumption~\eqref{eqn: inductive condition for the bound on the second-order derivative of the contour weight} and Lem.~\ref{itm: expansion of the second-order derivative of the log q-oriented partition function}
    \begin{equation}
        \begin{multlined}
            \bars*{B^2_q(\gamma;\bz)}
            \le\sum_{j=1}^{h_\gamma}\brackets[\Bigg]{
            \bars*{\cD_1\cD_2\log\frac{\Xi_\bz(\intr_j(\gamma)\mid m_j(\gamma),R_\gamma)}{\Xi_\bz^{m_j(\gamma)}(\intr_j(\gamma)\mid R_\gamma)}}
            +\bars*{\cD_1\cD_2\log\frac{\Xi_\bz(\intr_j(\gamma)\mid q)}{\Xi_\bz^q(\intr_j(\gamma))}}
            \\
            +\bars[\Bigg]{\cD_1\cD_2\log\frac{\Xi_\bz^{m_j(\gamma)}(\intr_j(\gamma)\mid R_\gamma)}{\Xi_\bz^q(\intr_j(\gamma))}}}
            \le \indicator{R(\intr(\gamma))}(\fr_1,\fr_2)2\brackets[\Big]{\prefactor{eqn: inductive bound on the second-order derivative of the log polymer partition function}e^{-\spatial{eqn: inductive bound on the second-order derivative of the log polymer partition function}(\fd_{12}-2)}+e^{3+2\nu/n}z^2\eps^{\fd_{12}-2}}.
        \end{multlined}
    \end{equation}
    We also use Prop.~\ref{itm: bound on the second-order derivative of the finite-volume contour self-potential} to bound $\abs{B^3_q(\gamma;\bz_\Delta)}\le \indicator{R^q(\Delta)}(\fr_1,\fr_2)e^{3+2\nu/n}z^2\eps^{\fn_{12}(\gamma)-2}$.
    Combining these estimates, bounding $\abs{\Gamma}\le\abs{X_1\cap\bL}$, and bounding all factors involving $\fd_{12}-2$ by $1$ yields~\eqref{eqn: bound on the second-order derivative of the contour weight}.
\end{proof}

\begin{proposition}
    \label{prop: bounds on the contour-polymer factors}
    The following estimates hold for the contour-polymer factor.
    Except in~\ref{itm: bound on the finite- and infinite-volume contour-polymer factors}, we assume that $X_1\in\cT(\Lambda)$.
    \begin{enumerate}[label=(\alph*),ref=\theproposition(\alph*)]
        \item \label{itm: bound on the finite- and infinite-volume contour-polymer factors} Let $\fR\in\Omega(\Z^2)$ be finite.
        For $X_1\in\cT(\Lambda)$, the finite-volume truncated contour-polymer factor satisfies
        \declarePeierls{eqn: bound on the truncated finite-volume contour-polymer factor}
        \begin{equation}
            \label{eqn: bound on the truncated finite-volume contour-polymer factor}
            \sum_{\substack{\Gamma\in\cC(\Lambda,q)\\\supp{\Gamma}=X_1\\R_\Gamma\supset\fR}}\prod_{\substack{\gamma\in\Gamma}}\bars*{\hat{\zeta}_q^{(\Lambda)}(\gamma;\bz_\Delta)}
            \le \indicator{R(X_1)}(\fR)z^{\abs{\fR}}e^{-\Peierls{eqn: bound on the truncated finite-volume contour-polymer factor}\abs{X_1\cap\bL}}.
        \end{equation}
        For $X_1\in\cT(\Z^2)$, the infinite-volume truncated contour-polymer factor satisfies
        \begin{equation}
            \label{eqn: bound on the truncated infinite-volume contour-polymer factor}
            \sum_{\substack{\Gamma\in\cC(\Z^2,q)\\\supp{\Gamma}=X_1}}\prod_{\substack{\gamma\in\Gamma}}\bars*{\hat{\zeta}_q(\gamma;\bz)}
            \le e^{-\Peierls{eqn: bound on the truncated finite-volume contour-polymer factor}\abs{X_1\cap\bL}}.
        \end{equation}
        \item \label{itm: bound on the finite-volume correction to the contour-polymer factor} The finite-volume correction to the truncated contour-polymer factor satisfies
        \declareprefactor{eqn: bound on the finite-volume correction to the contour-polymer factor}
		\declarePeierls{eqn: bound on the finite-volume correction to the contour-polymer factor}
        \begin{equation}
            \label{eqn: bound on the finite-volume correction to the contour-polymer factor}
            \begin{multlined}
                \bars[\Bigg]{\sum_{\substack{\Gamma\in\cC(\Lambda,q)\\\supp{\Gamma}=X_1}}\brackets[\Bigg]{\prod_{\substack{\gamma\in\Gamma}}\hat{\zeta}_q^{(\Lambda)}(\gamma;\bz_\Delta)-\prod_{\substack{\gamma\in\Gamma}}\hat{\zeta}_q(\gamma;\bz)}}
                \\
                \le \indicator{X_1\subset\Lambda\setminus\bdin[6\bL]{\Lambda}}
                \prefactor{eqn: bound on the finite-volume correction to the contour-polymer factor}\abs{X_1\cap\bL}
                e^{-\Peierls{eqn: bound on the finite-volume correction to the contour-polymer factor}\abs{X_1\cap\bL}}
                \eps^{\ceil{\rmd_\bL(X_1,\bdin[2\bL]{\Lambda})/2}-1}.
            \end{multlined}
        \end{equation}
        \item \label{itm: bound on the first-order derivative of the finite-volume contour-polymer factor} If~\eqref{eqn: recovery condition} and~\eqref{eqn: inductive condition for the bound on the first-order derivative of the contour weight} hold, then the first-order derivative of the true finite-volume contour-polymer factor satisfies
        \declareprefactor{eqn: bound on the first-order derivative of the finite-volume contour-polymer factor}
		\declarePeierls{eqn: bound on the first-order derivative of the finite-volume contour-polymer factor}
		\declarespatial{eqn: bound on the first-order derivative of the finite-volume contour-polymer factor}
        \begin{equation}
            \label{eqn: bound on the first-order derivative of the finite-volume contour-polymer factor}
            \bars[\Bigg]{\cD_1\sum_{\substack{\Gamma\in\cC(\Lambda,q)\\\supp{\Gamma}=X_1}}\prod_{\substack{\gamma\in\Gamma}}\zeta_q^{(\Lambda)}(\gamma;\bz_\Delta)}
            \le\indicator{R(\Delta)}(\fr_1)\prefactor{eqn: bound on the first-order derivative of the finite-volume contour-polymer factor}\abs{X_1\cap\bL}e^{-\Peierls{eqn: bound on the first-order derivative of the finite-volume contour-polymer factor}\abs{X_1\cap\bL}}e^{-\spatial{eqn: bound on the first-order derivative of the finite-volume contour-polymer factor}[\fn_1(X_1)-1]}.
        \end{equation}
        \item \label{itm: bound on the second-order derivative of the finite-volume contour-polymer factor} If~\eqref{eqn: recovery condition},~\eqref{eqn: inductive condition for the bound on the first-order derivative of the contour weight}, and~\eqref{eqn: inductive condition for the bound on the second-order derivative of the contour weight} hold, then the second-order derivative of the true finite-volume contour-polymer factor satisfies
        \declareprefactor{eqn: bound on the second-order derivative of the finite-volume contour-polymer factor}
		\declarespatial{eqn: bound on the second-order derivative of the finite-volume contour-polymer factor}
        \begin{equation}
            \label{eqn: bound on the second-order derivative of the finite-volume contour-polymer factor}
            \begin{multlined}
                \bars[\Bigg]{\cD_1\cD_2\sum_{\substack{\Gamma\in\cC(\Lambda,q)\\\supp{\Gamma}=X_1}}\prod_{\substack{\gamma\in\Gamma}}\zeta_q^{(\Lambda)}(\gamma;\bz_\Delta)}
                \le\indicator{R(\Delta)}(\fr_1,\fr_2)
                \\
                \times \prefactor{eqn: bound on the second-order derivative of the finite-volume contour-polymer factor}\abs{X_1\cap\bL}^2 e^{-\Peierls{eqn: bound on the first-order derivative of the finite-volume contour-polymer factor}\abs{X_1\cap\bL}}\brackets[\Big]{e^{-\spatial{eqn: bound on the second-order derivative of the finite-volume contour-polymer factor}[\fn_{12}(X_1)-2]}+e^{-\spatial{eqn: bound on the second-order derivative of the finite-volume contour-polymer factor}[\fn_1(X_1)-1]}e^{-\spatial{eqn: bound on the second-order derivative of the finite-volume contour-polymer factor}[\fn_2(X_1)-1]}}.
            \end{multlined}
        \end{equation}
    \end{enumerate}
\end{proposition}

\begin{proof}[Proof of~\ref{itm: bound on the finite- and infinite-volume contour-polymer factors}]
    Denote the connected components of $X_1$ by $Z_1,\dots,Z_n$.
    For $1\le i\le n$, denote by $\fR_{i}$ the set of rectangles in $\fR$ anchored in $Z_i$.
    We begin by interchanging the summation and product:
    \begin{equation}
        \label{eqn: finite-volume contour-polymer factor after interchanging summation and product}
        \sum_{\substack{\Gamma\in\cC(\Lambda,q)\\\supp{\Gamma}=X_1\\R_\Gamma\supset\fR}}\prod_{\substack{\gamma\in\Gamma}}\bars*{\hat{\zeta}_q^{(\Lambda)}(\gamma;\bz_\Delta)}
        =\prod_{i=1}^n\sum_{\substack{\gamma\in C(\Lambda,q)\\\supp{\gamma}=Z_i\\R_\gamma\supset\fR_i}}\bars*{\hat{\zeta}_q^{(\Lambda)}(\gamma;\bz_\Delta)}.
    \end{equation}
    Recalling the definition of $\hat{\zeta}_q^{(\Lambda)}(\gamma;\bz_\Delta)$ from~\eqref{eqn: finite-volume contour weight}, combining Prop.~\ref{itm: bounds on the finite- and infinite-volume contour self-potentials} and Lem.~\ref{lem: Peierls bound for contour weight} (noting that, for $\Gamma\in\cC(\Lambda,q)$, $\supp{\Gamma}\subset\Lambda\setminus\bdin[6\bL]{\Lambda}$ on which $\bz_\Delta\equiv\bz$) yields
    \begin{equation}
        \sum_{\substack{\gamma\in C(\Lambda,q)\\\supp{\gamma}=Z_i}}\bars*{\hat{\zeta}_q^{(\Lambda)}(\gamma;\bz_\Delta)}
        \le z^{\abs{\fR_i}} 3^{\abs{Z_i\cap\bL}} 
        e^{-\Peierls{eqn: Peierls bound for contour weight}\abs{Z_i\cap\bL}}
        e^{3e^{1+\nu}(z\ell^2)\eps\abs{Z_i\cap\bL}}
        e^{\tau\abs{Z_i\cap\bL}},
    \end{equation}
    where we used that the $\bL$-external boundaries of the holes of $\gamma$ are disjoint subsets of $Z_i$, and $3^{\abs{Z_i\cap\bL}}$ bounds the number of possible $\psi_\gamma$.
    Taking the product over $i$ completes the proof for the finite-volume truncated contour-polymer factor.
    The same argument applies to the infinite-volume truncated contour-polymer factor.
\end{proof}

\begin{proof}[Proof of~\ref{itm: bound on the finite-volume correction to the contour-polymer factor}]
    We must have that $X_1\subset\Lambda\setminus\bdin[6\bL]{\Lambda}$ by the definition of $q$-boundary conditions for profiles, which requires that they be constant and equal to $q$ for all $v\in\bL$ whose enclosing smoothing square $S_u\supset T_v$ satisfies $\rmd_{6\bL}(u,\Lambda^c)\le 2$.
    Denote the connected components of $X_1$ by $Z_1,\dots,Z_n$.
    Interchanging the summation and product as in~\eqref{eqn: finite-volume contour-polymer factor after interchanging summation and product} and constructing a telescoping sum yields
    \begin{equation}
        \label{eqn: telescoping sum in the finite-volume correction to the contour-polymer factor}
        \begin{multlined}
            \sum_{\substack{\Gamma\in\cC(\Lambda,q)\\\supp{\Gamma}=X_1}}\brackets[\Bigg]{\prod_{\substack{\gamma\in\Gamma}}\hat{\zeta}_q^{(\Lambda)}(\gamma;\bz_\Delta)-\prod_{\substack{\gamma\in\Gamma}}\hat{\zeta}_q(\gamma;\bz)}
            =\prod_{i=1}^n\sum_{\substack{\gamma\in C(\Lambda,q)\\\supp{\gamma}=Z_i}}\hat{\zeta}_q^{(\Lambda)}(\gamma;\bz_\Delta)-\prod_{i=1}^n\sum_{\substack{\gamma\in C(\Lambda,q)\\\supp{\gamma}=Z_i}}\hat{\zeta}_q(\gamma;\bz)
            \\
            =\sum_{i=1}^n
            \prod_{j=1}^{i-1}\brackets[\Bigg]{\sum_{\substack{\gamma\in C(\Lambda,q)\\\supp{\gamma}=Z_j}}\hat{\zeta}_q(\gamma;\bz)}
            \brackets[\Bigg]{\sum_{\substack{\gamma\in C(\Lambda,q)\\\supp{\gamma}=Z_i}}\left(\hat{\zeta}_q^{(\Lambda)}(\gamma;\bz_\Delta)-\hat{\zeta}_q(\gamma;\bz)\right)}
            \prod_{j=i+1}^{n}\brackets[\Bigg]{\sum_{\substack{\gamma\in C(\Lambda,q)\\\supp{\gamma}=Z_j}}\hat{\zeta}_q^{(\Lambda)}(\gamma;\bz_\Delta)}.
        \end{multlined}
    \end{equation}
    Using that $\hat{\zeta}^0_q(\gamma;\bz)$ depends only on the restriction of $\bz$ to $\Lambda\setminus\bdin[6\bL]{\Lambda}$ for $\gamma\in C(\Lambda,q)$, we factor
    \begin{equation}
        \begin{multlined}
            \hat{\zeta}_q^{(\Lambda)}(\gamma;\bz_\Delta)-\hat{\zeta}_q(\gamma;\bz)
            =\hat{\zeta}_q(\gamma;\bz)\Bigg[\exp\Bigg\{\sum_{M\in\cM^q(\Z^2)\setminus\cM^q(\Lambda)}\varphi(M)\chi_\gamma(M)\bz(M)
            \\
            +\sum_{M\in\cM^q(\Lambda)\setminus\cM^q(\Lambda\setminus\bdin[2\bL]{\Lambda})}\varphi(M)\chi_\gamma(M)\brackets*{\bz(M)-\bz_\Delta(M)}\Bigg\}-1\Bigg].
        \end{multlined}
    \end{equation}
    The exponent is bounded in absolute value by the RHS of~\eqref{eqn: bound on the finite-volume correction to the contour self-potential}.
    Thus, we may apply $\abs{e^x-1}\le\abs{x}e^{\abs{x}}$ and Prop.~\ref{itm: bound on the finite-volume correction to the contour self-potential} to get that
    \begin{align}
        {}&\abs{\hat{\zeta}_q^{(\Lambda)}(\gamma;\bz_\Delta)-\hat{\zeta}_q(\gamma;\bz)}
        \nonumber
        \\
        \le{}& 3e^{1+\nu}z\abs{\supp{\gamma}}\eps^{\ceil{\rmd_\bL(\supp{\gamma},\bdin[2\bL]{\Lambda})/2}-1}
        e^{3e^{1+\nu}z\abs{\supp{\gamma}}\eps^{\ceil{\rmd_\bL(\supp{\gamma},\bdin[2\bL]{\Lambda})/2}-1}}\bars*{\hat{\zeta}_q(\gamma;\bz)}
        \nonumber
        \\
        \le{}& 3e^{1+\nu}(z\ell^2)\abs{Z_i\cap\bL}\eps^{\ceil{\rmd_\bL(X_1,\bdin[2\bL]{\Lambda})/2}-1}
        e^{3e^{1+\nu}(z\ell^2)\eps^2\abs{X_1\cap\bL}}\bars*{\hat{\zeta}_q(\gamma;\bz)},
        \label{eqn: telescoping term in the finite-volume correction to the contour-polymer factor}
    \end{align}
    having used $\rmd_\bL(\supp{\gamma},\bdin[2\bL]{\Lambda})\ge 5$ in the last step.
    Inserting~\eqref{eqn: telescoping term in the finite-volume correction to the contour-polymer factor} into~\eqref{eqn: telescoping sum in the finite-volume correction to the contour-polymer factor} and using Lem.~\ref{lem: Peierls bound for contour weight} yields~\eqref{eqn: bound on the finite-volume correction to the contour-polymer factor}.
\end{proof}

\begin{proof}[Proof of~\ref{itm: bound on the first-order derivative of the finite-volume contour-polymer factor}]
    We write
    \begin{equation}
        \cD_1\prod_{\substack{\gamma\in\Gamma}}\zeta_q^{(\Lambda)}(\gamma;\bz_\Delta)
        =\prod_{\substack{\gamma\in\Gamma}}\zeta_q^{(\Lambda)}(\gamma;\bz_\Delta)
        \Bigg[\sum_{\substack{\gamma\in\Gamma}}\cD_1\log\zeta_q^{(\Lambda)}(\gamma;\bz_\Delta)\Bigg].
    \end{equation}
    By assumption, we can use Lem.~\ref{lem: bound on the first-order derivative of the contour weight} and Prop.~\ref{itm: bound on the finite- and infinite-volume contour-polymer factors} to bound
    \begin{equation}
        \begin{multlined}
            \bars[\Bigg]{\cD_1\sum_{\substack{\Gamma\in\cC(\Lambda,q)\\\supp{\Gamma}=X_1}}\prod_{\substack{\gamma\in\Gamma}}\zeta_q^{(\Lambda)}(\gamma;\bz_\Delta)}
            \\
            \le \indicator{R(X_1)}(\fr_1)z e^{-\Peierls{eqn: bound on the truncated finite-volume contour-polymer factor}\abs{X_1\cap\bL}}
            +\brackets[\Big]{\indicator{R(c(X_1))}(\fr_1)\prefactor{eqn: bound on the first-order derivative of the contour weight - localization}
            +\indicator{R^q(\Delta)}(\fr_1)\prefactor{eqn: bound on the first-order derivative of the contour weight - decay}\eps^{\fn_1(X_1)-1}}\abs{X_1\cap\bL} e^{-\Peierls{eqn: bound on the truncated finite-volume contour-polymer factor}\abs{X_1\cap\bL}}.
        \end{multlined}
    \end{equation}
    Observe that $\abs{X_1\cap\bL}\ge 8\rmd_\bL(\fT_1,X_1)$ whenever $\fr_1\in R(c(X_1))$.
    Hence, we can trade part of the Peierls penalty for a spatial decay factor by bounding
    \begin{equation}
        \label{eqn: bound on the first-order derivative of the finite-volume contour-polymer factor - trading for spatial decay}
        \indicator{R(c(X_1))}(\fr_1) e^{-\Peierls{eqn: bound on the truncated finite-volume contour-polymer factor}\abs{X_1\cap\bL}/2}
        \le\indicator{R(c(X_1))}(\fr_1) e^{-4\Peierls{eqn: bound on the truncated finite-volume contour-polymer factor}\rmd_\bL(\fT_1,X_1)}
        \le \indicator{R(\Delta)}(\fr_1) e^{-8\Peierls{eqn: bound on the truncated finite-volume contour-polymer factor}[\fn_1(X_1)-1]}.
    \end{equation}
    Combining the above estimates yields~\eqref{eqn: bound on the first-order derivative of the finite-volume contour-polymer factor}.
\end{proof}

\begin{proof}[Proof of~\ref{itm: bound on the second-order derivative of the finite-volume contour-polymer factor}]
    We write
    \begin{equation}
        \begin{multlined}
            \cD_1\cD_2\prod_{\substack{\gamma\in\Gamma}}\zeta_q^{(\Lambda)}(\gamma;\bz_\Delta)
            =\prod_{\substack{\gamma\in\Gamma}}\zeta_q^{(\Lambda)}(\gamma;\bz_\Delta)
            \Bigg[\sum_{\substack{\gamma\in\Gamma}}\cD_1\cD_2\log\zeta_q^{(\Lambda)}(\gamma;\bz_\Delta)
            \\
            +\Bigg(\sum_{\substack{\gamma\in\Gamma}}\cD_1\log\zeta_q^{(\Lambda)}(\gamma;\bz_\Delta)\Bigg)\Bigg(\sum_{\substack{\gamma\in\Gamma}}\cD_2\log\zeta_q^{(\Lambda)}(\gamma;\bz_\Delta)\Bigg)\Bigg].
        \end{multlined}
    \end{equation} 
    By assumption, we can use Lem.~\ref{lem: bound on the second-order derivative of the contour weight} and Prop.~\ref{itm: bound on the finite- and infinite-volume contour-polymer factors} to bound
    \begin{align}
        {}&\bars[\Bigg]{\sum_{\substack{\Gamma\in\cC(\Lambda,q)\\\supp{\Gamma}=X_1}}
        \brackets[\Bigg]{\prod_{\substack{\gamma\in\Gamma}}\zeta_q^{(\Lambda)}(\gamma;\bz_\Delta)}
        \brackets[\Bigg]{\sum_{\substack{\gamma\in\Gamma}}\cD_1\cD_2\log\zeta_q^{(\Lambda)}(\gamma;\bz_\Delta)}
        }
        \nonumber
        \\
        \le{}& \prefactor{eqn: bound on the second-order derivative of the contour weight} \abs{X_1\cap\bL} e^{-\Peierls{eqn: bound on the truncated finite-volume contour-polymer factor}\abs{X_1\cap\bL}}
        \brackets*{\indicator{R(c(X_1))}(\fr_1,\fr_2)
        +\indicator{R^q(\Delta)}(\fr_1,\fr_2)\eps^{\fn_{12}(X_1)-2}}
        \nonumber
        \\
        \le{}& \indicator{R(\Delta)}(\fr_1,\fr_2) \prefactor{eqn: bound on the second-order derivative of the contour weight} \abs{X_1\cap\bL} e^{-\Peierls{eqn: bound on the truncated finite-volume contour-polymer factor}\abs{X_1\cap\bL}/2}
        \brackets*{e^{-4\Peierls{eqn: bound on the truncated finite-volume contour-polymer factor}[\fn_1(X_1)+\fn_2(X_1)-2]}+\eps^{\fn_{12}(X_1)-2}},
        \label{eqn: bound on the second-order derivative of the finite-volume contour-polymer factor - second-order contribution}
    \end{align}
    having traded part of the Peierls penalty for a spatial decay factor as in~\eqref{eqn: bound on the first-order derivative of the finite-volume contour-polymer factor - trading for spatial decay} in the last step.
    Similarly, we can use Lem.~\ref{lem: bound on the first-order derivative of the contour weight}, Prop.~\ref{itm: bound on the finite- and infinite-volume contour-polymer factors}, and~\eqref{eqn: bound on the first-order derivative of the finite-volume contour-polymer factor - trading for spatial decay} to bound
    \begin{align*}
        {}&\bars[\Bigg]{\sum_{\substack{\Gamma\in\cC(\Lambda,q)\\\supp{\Gamma}=X_1}}
        \brackets[\Bigg]{\prod_{\substack{\gamma\in\Gamma}}\zeta_q^{(\Lambda)}(\gamma;\bz_\Delta)}
        \Bigg(\sum_{\substack{\gamma\in\Gamma}}\cD_1\log\zeta_q^{(\Lambda)}(\gamma;\bz_\Delta)\Bigg)\Bigg(\sum_{\substack{\gamma\in\Gamma}}\cD_2\log\zeta_q^{(\Lambda)}(\gamma;\bz_\Delta)\Bigg)}
        \\
        \le{}& \indicator{R(X_1)}(\fr_1,\fr_2)z^2 e^{-\Peierls{eqn: bound on the truncated finite-volume contour-polymer factor}\abs{X_1\cap\bL}}
        \\
        +{}& \sum_{i=1,2}\indicator{R(X_1)}(\fr_{3-i}) \brackets[\Big]{\indicator{R(c(X_1))}(\fr_i)\prefactor{eqn: bound on the first-order derivative of the contour weight - localization}
        +\indicator{R^q(\Delta)}(\fr_i)\prefactor{eqn: bound on the first-order derivative of the contour weight - decay}\eps^{\fn_i(X_1)-1}}
        z \abs{X_1\cap\bL} e^{-\Peierls{eqn: bound on the truncated finite-volume contour-polymer factor}\abs{X_1\cap\bL}}
        \\
        +{}&\prod_{i=1,2}\brackets[\Big]{\indicator{R(c(X_1))}(\fr_i)\prefactor{eqn: bound on the first-order derivative of the contour weight - localization}
        +\indicator{R^q(\Delta)}(\fr_i)\prefactor{eqn: bound on the first-order derivative of the contour weight - decay}\eps^{\fn_i(X_1)-1}}\abs{X_1\cap\bL}^2 e^{-\Peierls{eqn: bound on the truncated finite-volume contour-polymer factor}\abs{X_1\cap\bL}}
        \\
        \le{}& \indicator{R(\Delta)}(\fr_1,\fr_2) e^{-\Peierls{eqn: bound on the truncated finite-volume contour-polymer factor}\abs{X_1\cap\bL}/2} \cbraces[\Big]{z^2 e^{-4\Peierls{eqn: bound on the truncated finite-volume contour-polymer factor}[\fn_1(X_1)+\fn_2(X_1)-2]}
        \\
        +{}& z \abs{X_1\cap\bL}\sum_{i=1,2}e^{-4\Peierls{eqn: bound on the truncated finite-volume contour-polymer factor}[\fn_{3-i}(X_1)-1]} \brackets[\Big]{\prefactor{eqn: bound on the first-order derivative of the contour weight - localization}e^{-4\Peierls{eqn: bound on the truncated finite-volume contour-polymer factor}[\fn_i(X_1)-1]}
        +\prefactor{eqn: bound on the first-order derivative of the contour weight - decay}\eps^{\fn_i(X_1)-1}}
        \\
        +{}&\abs{X_1\cap\bL}^2 \prod_{i=1,2}\brackets[\Big]{\prefactor{eqn: bound on the first-order derivative of the contour weight - localization}e^{-4\Peierls{eqn: bound on the truncated finite-volume contour-polymer factor}[\fn_i(X_1)-1]}
        +\prefactor{eqn: bound on the first-order derivative of the contour weight - decay}\eps^{\fn_i(X_1)-1}}}
        .
        \tag{\stepcounter{equation}\theequation}
    \end{align*}
    Combining the above estimates yields~\eqref{eqn: bound on the second-order derivative of the finite-volume contour-polymer factor}.
\end{proof}

\subsubsection{Estimates for the polymer weight}
\label{sec: bounds on the polymer weights}

\begin{proposition}
    \label{prop: bounds on the polymer weights}
    The following estimates hold for the polymer weight.
    Except in~\ref{itm: bound on the finite- and infinite-volume polymer weights}, we assume that $X\in\cT_\ast(\Lambda)$.
    \begin{enumerate}[label=(\alph*),ref=\theproposition(\alph*)]
        \item \label{itm: bound on the finite- and infinite-volume polymer weights} 
        \declarePeierls{eqn: bound on the finite-volume polymer weight}
        For $X\in\cT_\ast(\Lambda)$, the finite-volume truncated polymer weight satisfies
        \begin{equation}
            \label{eqn: bound on the finite-volume polymer weight}
            \abs{\hat{K}_q^{(\Lambda)}(X;\bz_\Delta)}\le e^{-\Peierls{eqn: bound on the finite-volume polymer weight}\abs{X\cap\bL}}.
        \end{equation}
        For $X\in\cT_\ast(\Z^2)$, the infinite-volume truncated polymer weight satisfies
        \begin{equation}
            \label{eqn: bound on the infinite-volume polymer weight}
            \abs{\hat{K}_q(X;\bz)}\le e^{-\Peierls{eqn: bound on the finite-volume polymer weight}\abs{X\cap\bL}}.
        \end{equation}
        \item \label{itm: bound on the finite-volume correction to the polymer weight}
        \declareprefactor{eqn: bound on the finite-volume correction to the polymer weight}
		\declarePeierls{eqn: bound on the finite-volume correction to the polymer weight}
		\declarespatial{eqn: bound on the finite-volume correction to the polymer weight}
        The finite-volume correction to the truncated polymer weight satisfies
        \begin{equation}
            \label{eqn: bound on the finite-volume correction to the polymer weight}
            \abs{\hat{K}_q^{(\Lambda)}(X;\bz_\Delta)-\hat{K}_q(X;\bz)}
            \le \prefactor{eqn: bound on the finite-volume correction to the polymer weight}\abs{X\cap\bL}e^{-\Peierls{eqn: bound on the finite-volume correction to the polymer weight}(\abs{X\cap\bL}-1)}
            e^{-\spatial{eqn: bound on the finite-volume correction to the polymer weight}\rmd_\bL(X,\bdin[2\bL]{\Lambda})}.
        \end{equation}
        \item \label{itm: bound on the first-order derivative of the finite-volume polymer weight}
        \declareprefactor{eqn: bound on the first-order derivative of the finite-volume polymer weight}
		\declarePeierls{eqn: bound on the first-order derivative of the finite-volume polymer weight}
        If~\eqref{eqn: inductive condition for the bound on the first-order derivative of the contour weight} holds, then the first-order derivative of the true finite-volume polymer weight satisfies
        \begin{equation}
            \label{eqn: bound on the first-order derivative of the finite-volume polymer weight}
            \bars*{\cD_1K_q^{(\Lambda)}(X;\bz_\Delta)}
            \le\indicator{R(\Delta)}(\fr_1)\prefactor{eqn: bound on the first-order derivative of the finite-volume polymer weight}\abs{X\cap\bL}e^{-\Peierls{eqn: bound on the first-order derivative of the finite-volume polymer weight}(\abs{X\cap\bL}-1)}e^{-\spatial{eqn: bound on the first-order derivative of the finite-volume contour-polymer factor}[\fn_1(X)-1]}.
        \end{equation}
        \item \label{itm: bound on the second-order derivative of the finite-volume polymer weight}
        \declareprefactor{eqn: bound on the second-order derivative of the finite-volume polymer weight}
		\declareloss{eqn: bound on the second-order derivative of the finite-volume polymer weight}
		\declarePeierls{eqn: bound on the second-order derivative of the finite-volume polymer weight}
		\declarespatial{eqn: bound on the second-order derivative of the finite-volume polymer weight}
        If~\eqref{eqn: inductive condition for the bound on the first-order derivative of the contour weight} and~\eqref{eqn: inductive condition for the bound on the second-order derivative of the contour weight} hold, then the second-order derivative of the true finite-volume polymer weight satisfies
        \begin{equation}
            \label{eqn: bound on the second-order derivative of the finite-volume polymer weight}
            \begin{multlined}
                \bars*{\cD_1\cD_2K_q^{(\Lambda)}(X;\bz_\Delta)}
                \le\indicator{R(\Delta)}(\fr_1,\fr_2)
                \\
                \times\prefactor{eqn: bound on the second-order derivative of the finite-volume polymer weight}\abs{X\cap\bL}^2 e^{-\Peierls{eqn: bound on the second-order derivative of the finite-volume polymer weight}(\abs{X\cap\bL}-2)}\brackets[\Big]{e^{-\spatial{eqn: bound on the second-order derivative of the finite-volume polymer weight}[\fn_{12}(X)-2]}+e^{-\spatial{eqn: bound on the second-order derivative of the finite-volume polymer weight}[\fn_1(X)-1]}e^{-\spatial{eqn: bound on the second-order derivative of the finite-volume polymer weight}[\fn_2(X)-1]}}.
            \end{multlined}
        \end{equation}
    \end{enumerate}
\end{proposition}

\begin{proof}[Proof of~\ref{itm: bound on the finite- and infinite-volume polymer weights}]
    Combining Props.~\ref{itm: bound on the finite- and infinite-volume interaction-polymer factors} and~\ref{itm: bound on the finite- and infinite-volume contour-polymer factors} and Lem.~\ref{itm: covering sum zeroth moment}, we get that
    \begin{equation}
        \begin{multlined}
            \abs{\hat{K}_q^{(\Lambda)}(X;\bz_\Delta)}
            \le\sum_{\substack{X_1,X_2\in\cT(\Lambda)\\X_1\cup X_2=X}}
            \sum_{\substack{\Gamma\in\cC(\Lambda,q)\\\supp{\Gamma}=X_1}}
            \brackets[\Bigg]{\prod_{\substack{\gamma\in\Gamma}}\hat{\zeta}_q^{(\Lambda)}(\gamma;\bz_\Delta)}
            \brackets[\Bigg]{\sum_{\substack{\cY\in\cB_\ast(\Lambda)\\\cup_{Y\in\cY}Y=X_2}}\prod_{Y\in\cY}\abs{e^{V_\Gamma^{(\Lambda)}(Y;\bz_\Delta)}-1}}
            \\
            \le\sum_{\substack{X_1,X_2\in\cT(\Z^2)\\X_1\cup X_2=X}} e^{-(\Peierls{eqn: bound on the truncated finite-volume contour-polymer factor}-\loss{eqn: bound on the finite-volume interaction-polymer factor})\abs{X_1\cap\bL}}e^{-\Peierls{eqn: bound on the finite-volume interaction-polymer factor}\abs{X_2\cap\bL}}
            =e^{-\Peierls{eqn: bound on the finite-volume polymer weight}\abs{X\cap\bL}}.
        \end{multlined}
    \end{equation}
    An analogous computation applies to the infinite-volume truncated polymer weight.
\end{proof}

\begin{proof}[Proof of~\ref{itm: bound on the finite-volume correction to the polymer weight}]
    Observe that we can let the sums over $X_1,X_2\in\cT(\Z^2)$ and $\cY\in\cB_\ast(\Z^2)$ in the definition of the infinite-volume polymer weight~\eqref{eqn: infinite-volume polymer weight} (and its truncated version) run over $X_1,X_2\in\cT(\Lambda)$ and $\cY\in\cB_\ast(\Lambda)$ instead by the assumption that $X\in\cT_\ast(\Lambda)$.
    Hence, we may write
    \begin{align*}
        {}&\hat{K}_q^{(\Lambda)}(X;\bz_\Delta)-\hat{K}_q(X;\bz)
        \\
        ={}&\sum_{\substack{X_1,X_2\in\cT(\Lambda)\\X_1\cup X_2=X}}
        \sum_{\substack{\Gamma\in\cC(\Lambda,q)\\\supp{\Gamma}=X_1}}
        \brackets[\Bigg]{\prod_{\substack{\gamma\in\Gamma}}\hat{\zeta}_q^{(\Lambda)}(\gamma;\bz_\Delta)}
        \sum_{\substack{\cY\in\cB_\ast(\Lambda)\\\cup_{Y\in\cY}Y=X_2}}\brackets[\Bigg]{\prod_{Y\in\cY}\left(e^{V_\Gamma^{(\Lambda)}(Y;\bz_\Delta)}-1\right)-\prod_{Y\in\cY}\left(e^{V_\Gamma(Y;\bz)}-1\right)}
        \\
        +{}&\sum_{\substack{X_1,X_2\in\cT(\Lambda)\\X_1\cup X_2=X}}
        \sum_{\substack{\Gamma\in\cC(\Lambda,q)\\\supp{\Gamma}=X_1}}
        \brackets[\Bigg]{\prod_{\substack{\gamma\in\Gamma}}\hat{\zeta}_q^{(\Lambda)}(\gamma;\bz_\Delta)-\prod_{\substack{\gamma\in\Gamma}}\hat{\zeta}_q(\gamma;\bz)}
        \brackets[\Bigg]{\sum_{\substack{\cY\in\cB_\ast(\Lambda)\\\cup_{Y\in\cY}Y=X_2}}\prod_{Y\in\cY}\left(e^{V_\Gamma(Y;\bz)}-1\right)}
        \\
        -{}&\sum_{\substack{X_1,X_2\in\cT(\Lambda)\\X_1\cup X_2=X}}
        \sum_{\substack{\Gamma\in\cC(\Z^2,q)\setminus\cC(\Lambda,q)\\\supp{\Gamma}=X_1}}
        \brackets[\Bigg]{\prod_{\substack{\gamma\in\Gamma}}\hat{\zeta}_q(\gamma;\bz)}
        \brackets[\Bigg]{\sum_{\substack{\cY\in\cB_\ast(\Lambda)\\\cup_{Y\in\cY}Y=X_2}}\prod_{Y\in\cY}\left(e^{V_\Gamma(Y;\bz)}-1\right)}.
        \tag{\stepcounter{equation}\theequation}
    \end{align*}
    Combining Props.~\ref{itm: bound on the finite- and infinite-volume interaction-polymer factors},~\ref{itm: bound on the finite-volume correction to the interaction-polymer factors},~\ref{itm: bound on the finite- and infinite-volume contour-polymer factors}, and~\ref{itm: bound on the finite-volume correction to the contour-polymer factor}, we bound the first two terms in the above in absolute value by
    \begin{equation}
        \begin{multlined}
            \prefactor{eqn: bound on the finite-volume correction to the interaction-polymer factors}\sum_{\substack{X_1,X_2\in\cT(\Z^2)\\X_1\cup X_2=X}}
            \abs{X_1\cap\bL}
            e^{-(\Peierls{eqn: bound on the truncated finite-volume contour-polymer factor}-\loss{eqn: bound on the finite-volume correction to the interaction-polymer factors})\abs{X_1\cap\bL}}e^{-\Peierls{eqn: bound on the finite-volume correction to the interaction-polymer factors}\abs{X_2\cap\bL}}
            e^{-\spatial{eqn: bound on the finite-volume correction to the interaction-polymer factors}\rmd_\bL(X_1,\bdin[2\bL]{\Lambda})}
            \\
            +\prefactor{eqn: bound on the finite-volume correction to the contour-polymer factor}\sum_{\substack{X_1,X_2\in\cT(\Z^2)\\X_1\cup X_2=X}}
            \indicator{X_1\subset\Lambda\setminus\bdin[6\bL]{\Lambda}}
            \abs{X_1\cap\bL}
            e^{-(\Peierls{eqn: bound on the finite-volume correction to the contour-polymer factor}-\loss{eqn: bound on the finite-volume interaction-polymer factor})\abs{X_1\cap\bL}}e^{-\Peierls{eqn: bound on the finite-volume interaction-polymer factor}\abs{X_2\cap\bL}}
            \eps^{\ceil{\rmd_\bL(X_1,\bdin[2\bL]{\Lambda})/2}-1}.
        \end{multlined}
    \end{equation}
    Using that $\ceil{\rmd_\bL(X_1,\bdin[2\bL]{\Lambda})/2}-1\ge \rmd_\bL(X_1,\bdin[2\bL]{\Lambda})/3$ for $X_1\subset\Lambda\setminus\bdin[6\bL]{\Lambda}$, we  extract a uniform decay factor of $e^{-\spatial{eqn: bound on the finite-volume correction to the polymer weight}\rmd_\bL(X,\bdin[2\bL]{\Lambda})}$.
    The rest is estimated using Lem.~\ref{itm: covering sum first moment}.
    As for the third term, we relax the sum over $\Gamma\in\cC(\Z^2,q)\setminus\cC(\Lambda,q)$ to all $\Gamma\in\cC(\Z^2,q)$, thereby gaining an indicator function $\indicator{X_1\cap\bdin[2\bL]{\Lambda}\ne\emptyset}$, to which we apply
    \begin{equation}
        \indicator{X_1\cap\bdin[2\bL]{\Lambda}\ne\emptyset}
        \le \indicator{X\cap\bdin[2\bL]{\Lambda}\ne\emptyset}
        \le e^{-\spatial{eqn: bound on the finite-volume correction to the polymer weight}\rmd_\bL(X,\bdin[2\bL]{\Lambda})},
    \end{equation}
    and then the rest is bounded using Lem.~\ref{itm: covering sum zeroth moment}.
    Combining the estimates yields~\eqref{eqn: bound on the finite-volume correction to the polymer weight}.
\end{proof}

\begin{proof}[Proof of~\ref{itm: bound on the first-order derivative of the finite-volume polymer weight}]
    Using Props.~\ref{itm: bound on the finite- and infinite-volume interaction-polymer factors},~\ref{itm: bound on the first-order derivative of the finite-volume interaction-polymer factor},~\ref{itm: bound on the finite- and infinite-volume contour-polymer factors} and~\ref{itm: bound on the first-order derivative of the finite-volume contour-polymer factor}, we bound 
    \begin{equation}
        \begin{multlined}
            \bars{\cD_1K_q^{(\Lambda)}(X;\bz_\Delta)}
            \le\indicator{R(\Delta)}(\fr_1)
            \sum_{\substack{X_1,X_2\in\cT(\Z^2)\\X_1\cup X_2=X}}\Big[
            \prefactor{eqn: bound on the first-order derivative of the finite-volume contour-polymer factor}\abs{X_1\cap\bL}e^{-(\Peierls{eqn: bound on the first-order derivative of the finite-volume contour-polymer factor}-\loss{eqn: bound on the finite-volume interaction-polymer factor})\abs{X_1\cap\bL}}e^{-\spatial{eqn: bound on the first-order derivative of the finite-volume contour-polymer factor}[\fn_1(X_1)-1]}\eps^{-\Peierls{eqn: bound on the finite-volume interaction-polymer factor}\abs{X_2\cap\bL}}
            \\
            +\indicator{R(X_2)}(\fr_1)\prefactor{eqn: bound on the first-order derivative of the finite-volume interaction-polymer factor}\abs{X_1\cap\bL}e^{-(\Peierls{eqn: bound on the truncated finite-volume contour-polymer factor}-\loss{eqn: bound on the finite-volume interaction-polymer factor})\abs{X_1\cap\bL}}e^{-\Peierls{eqn: bound on the finite-volume interaction-polymer factor}\abs{X_2\cap\bL}}
            \Big].
        \end{multlined}
    \end{equation}
    We extract a uniform spatial decay factor by bounding 
    \begin{equation}
        \max\set*{e^{-\spatial{eqn: bound on the first-order derivative of the finite-volume contour-polymer factor}[\fn_1(X_1)-1]},\indicator{R(X_2)}(\fr_1)}\le e^{-\spatial{eqn: bound on the first-order derivative of the finite-volume contour-polymer factor}[\fn_1(X)-1]}.
    \end{equation}
    The rest is now explicitly estimated using Lem.~\ref{itm: covering sum first moment}, which yields~\eqref{eqn: bound on the first-order derivative of the finite-volume polymer weight}.
\end{proof}

\begin{proof}[Proof of~\ref{itm: bound on the second-order derivative of the finite-volume polymer weight}]
    Using Props.~\ref{itm: bound on the finite- and infinite-volume interaction-polymer factors},~\ref{itm: bound on the first-order derivative of the finite-volume interaction-polymer factor},~\ref{itm: bound on the second-order derivative of the finite-volume interaction-polymer factor},~\ref{itm: bound on the finite- and infinite-volume contour-polymer factors},~\ref{itm: bound on the first-order derivative of the finite-volume contour-polymer factor} and~\ref{itm: bound on the second-order derivative of the finite-volume contour-polymer factor}, we bound
    \begin{align*}
        \Big\lvert\cD_1{}&\cD_2K_q^{(\Lambda)}(X;\bz_\Delta)\Big\rvert
        \le \indicator{R(\Delta)}(\fr_1,\fr_2)\sum_{\substack{X_1,X_2\in\cT(\Z^2)\\X_1\cup X_2=X}}\abs{X_1\cap\bL}^2
        \\
        \times\Big\{{}& \prefactor{eqn: bound on the second-order derivative of the finite-volume contour-polymer factor} e^{-(\Peierls{eqn: bound on the first-order derivative of the finite-volume contour-polymer factor}-\loss{eqn: bound on the finite-volume interaction-polymer factor})\abs{X_1\cap\bL}}e^{-\Peierls{eqn: bound on the finite-volume interaction-polymer factor}\abs{X_2\cap\bL}}\brackets[\Big]{e^{-\spatial{eqn: bound on the second-order derivative of the finite-volume contour-polymer factor}[\fn_{12}(X_1)-2]}+e^{-\spatial{eqn: bound on the second-order derivative of the finite-volume contour-polymer factor}[\fn_1(X_1)-1]}e^{-\spatial{eqn: bound on the second-order derivative of the finite-volume contour-polymer factor}[\fn_2(X_1)-1]}}
        \\
        {}&+\prefactor{eqn: bound on the first-order derivative of the finite-volume interaction-polymer factor}\prefactor{eqn: bound on the first-order derivative of the finite-volume contour-polymer factor}e^{-(\Peierls{eqn: bound on the first-order derivative of the finite-volume contour-polymer factor}-\loss{eqn: bound on the finite-volume interaction-polymer factor})\abs{X_1\cap\bL}}e^{-\Peierls{eqn: bound on the finite-volume interaction-polymer factor}\abs{X_2\cap\bL}}\brackets[\Big]{\indicator{R(X_2)}(\fr_2)e^{-\spatial{eqn: bound on the first-order derivative of the finite-volume contour-polymer factor}[\fn_1(X_1)-1]}
        +\indicator{R(X_2)}(\fr_1)e^{-\spatial{eqn: bound on the first-order derivative of the finite-volume contour-polymer factor}[\fn_2(X_1)-1]}}
        \\
        {}&+\indicator{R(X_2)}(\fr_1,\fr_2)\prefactor{eqn: bound on the second-order derivative of the finite-volume interaction-polymer factor} e^{-(\Peierls{eqn: bound on the truncated finite-volume contour-polymer factor}-\loss{eqn: bound on the finite-volume interaction-polymer factor})\abs{X_1\cap\bL}}e^{-\Peierls{eqn: bound on the second-order derivative of the finite-volume interaction-polymer factor}\abs{X_2\cap\bL}} \Big\}.
        \tag{\stepcounter{equation}\theequation}
    \end{align*}
    We convert the indicators to spatial decays by bounding $\indicator{R(X_2)}(\fr_i)\le e^{-\spatial{eqn: bound on the second-order derivative of the finite-volume polymer weight}[\fn_i(X_2)-1]}$.
    By unifying the spatial decay rates and bounding $\fn_i(X_j)\ge \fn_i(X)$, we extract all the spatial decay factors.
    The rest is explicitly estimated using Lem.~\ref{itm: covering sum second moment}, which yields~\eqref{eqn: bound on the second-order derivative of the finite-volume polymer weight}.
\end{proof}

\subsection{Proof of Proposition~\ref{prop: main}}
\label{sec: proof of the main proposition}

\begin{lemma}
    \label{lem: the truncated polymer model can be cluster expanded}
    Let $\alpha>0$ be such that $e^{-(\Peierls{eqn: bound on the finite-volume polymer weight}-2\alpha)}<\frac{1}{64}$ and
    \begin{equation}
        \label{eqn: condition on the Peierls constant for polymer weights}
        9c_0(e^{-(\Peierls{eqn: bound on the finite-volume polymer weight}-2\alpha)})\le \alpha.
    \end{equation}
    Then, the cluster expansion applies to the truncated polymer model with either finite- or infinite-volume weights and the choices of auxiliary functions $a(X)=b(X):=\alpha\abs{X\cap\bL}$ (in the context of Thm.~\ref{thm: Ueltschi criterion}) for $X\in\cT_\ast(\Z^2)$.
\end{lemma}

\begin{proof}
    Using Prop.~\ref{itm: bound on the finite- and infinite-volume polymer weights}, Lem.~\ref{itm: incompatible polymer sum}, and~\eqref{eqn: condition on the Peierls constant for polymer weights}, we estimate
    \begin{equation}
        \sum_{\substack{X\in\cT_\ast(\Lambda)\\X\not\sim X_1}}\abs{\hat{K}_q^{(\Lambda)}(X;\bz_\Delta)}e^{a(X)+b(X)}
        \le 9c_0(e^{-(\Peierls{eqn: bound on the finite-volume polymer weight}-2\alpha)})\abs{X_1\cap\bL}
        \le a(X_1).
    \end{equation}
    The same argument applies to the infinite-volume truncated polymer weights.
\end{proof}

We proceed with the proof of Prop.~\ref{prop: main}.

\begin{proof}[Proof of Prop.~\ref{prop: main}]
    We argue by induction on the order of $\Lambda$.

    \paragraph{Base case}
    If $\Lambda$ is of order $0$, then $\hat{\Xi}_{\bz_\Delta}(\Lambda\mid q)=\Xi_{\bz_\Delta}(\Lambda\mid q)=\Xi_{\bz_\Delta}^q(\Lambda)=\Xi_\bz^q(\Delta)$.
    The bounds~\eqref{eqn: inductive bound on the first-order derivative of the log polymer partition function} and~\eqref{eqn: inductive bound on the second-order derivative of the log polymer partition function} hold trivially since the ratio of partition functions is identically $1$.
    For~\eqref{eqn: inductive bound on the log flipping term}, we start by changing $\bz$ to $z$ by noting that 
    \begin{equation}
        \abs{\log\frac{\Xi_{\bz}^q(\Delta)}{\Xi_{z}^{q}(\Delta)}}\le\nu.
    \end{equation}
    We now treat $\log\Xi_{z}^{q}(\Delta)$ using the cluster expansion.
    For $M\in\cM^q(\Delta)$, we write $n(M)$ for the number of distinct rectangles in $M$, and note the identity $1=\frac{1}{n(M)}\sum_{r\in R^q(\Delta)}\indicator{M}(r)$.
    Thus, we can expand
    \begin{equation}
        \begin{multlined}
            \log \Xi_{z}^q(\Delta)
            =\sum_{r\in R^q(\Delta)}\sum_{\substack{M\in\cM^q(\Delta)\\M\ni r}}\frac{\varphi(M)}{n(M)}z^{\abs{M}}
            \\
            =\sum_{r\in R^q(\Delta)}\brackets[\Bigg]{\sum_{\substack{M\in\cM^q(\Z^2)\\M\ni r}}\frac{\varphi(M)}{n(M)}z^{\abs{M}}-\sum_{\substack{M\in\cM^q(\Z^2)\setminus\cM^q(\Delta)\\M\ni r}}\frac{\varphi(M)}{n(M)}z^{\abs{M}}}.
        \end{multlined}
    \end{equation}
    The first sum in the brackets is independent of $q$ and $r$ and cancels with the corresponding sum in the expansion of $\log\Xi^{-q}_z(\Delta)$.
    The second double sum leads to a boundary term: any contributing cluster must intersect $R^q(\bdin[2\bL]{\Lambda})$ and contain at least one rectangle other than $r$, so we can bound using~\eqref{eqn: rectangle cluster tail estimate}
    \begin{equation}
        \sum_{r\in R^q(\Delta)}\sum_{\substack{M\in\cM^q(\Z^2)\setminus\cM^q(\Delta)\\M\ni r}}\frac{\abs{\varphi(M)}}{n(M)}z^{\abs{M}}
        \le \sum_{r\in R^q(\bdin[2\bL]{\Lambda})}\sum_{\substack{M\in\cM^q(\Z^2)\\M\ni r\\\abs{M}\ge 2}}\abs{\varphi(M)}z^{\abs{M}}
        \le ez\eps\bars{\bdin[2\bL]{\Lambda}}.
    \end{equation}
    Since $\bars{\bdin[2\bL]{\Lambda}}\le 25\bars{\bdex[\bL]{\Lambda}}$, the proof of the base case is complete provided that
    \begin{equation}
        \label{eqn: closure condition for the base case}
        2\nu+50e(z\ell^2)\eps\le \tau/2.
    \end{equation}

    \paragraph{Inductive step}
    Suppose now that the bound holds for all regions of order at most $n$.
    Let $\Lambda$ be a region of order $n+1$.
    The maximum order of contours in $\cC(\Lambda,q)$ is $n+1$.
    Since the truncated weights of contours of order $\le n+1$ depend only on the truncated partition functions in regions of order at most $n$, the recovery condition of Lem.~\ref{lem: recovery of true weights} applies, so we can remove the truncation on the LHS of~\eqref{eqn: inductive bound on the log flipping term}.

    To differentiate clusters of polymers as we will deal with now from the clusters of rectangles introduced earlier in Sec.~\ref{sec: cluster expansion of the q-oriented partition function}, we introduce the notation $\cU(\cdot)$ to denote the collection of clusters of polymers in $\cT_\ast(\cdot)$.
    The notation $n_X(U)$ retains its meaning as the multiplicity of $X\in\cT_\ast(\cdot)$ in $U\in\cU(\cdot)$ as specified in Sec.~\ref{sec: cluster expansion tools}.
    Lastly, we write $\supp{U}:=\cup_{X\in U}X$.

    \subparagraph{Proof of~\eqref{eqn: inductive bound on the first-order derivative of the log polymer partition function}}

    Using~\eqref{eqn: cluster bound with one anchor} and Prop.~\ref{itm: bound on the first-order derivative of the finite-volume polymer weight}, we bound
    \begin{equation}
        \label{eqn: bound on the first-order derivative of the log polymer partition function}
        \begin{multlined}
            \bars*{\cD_1\log\frac{\Xi_{\bz_\Delta}(\Lambda\mid q)}{\Xi_{\bz_\Delta}^q(\Lambda)}}
            \le\sum_{X_1\in\cT_\ast(\Lambda)}\abs{\cD_1K_q^{(\Lambda)}(X_1;\bz_\Delta)}\sum_{\substack{U\in\cU(\Lambda)\\U\ni X_1}}\abs{\varphi(U)}n_{X_1}(U)\prod_{X\in U\setminus\mset{X_1}}\abs{K_q^{(\Lambda)}(X;\bz_\Delta)}
            \\
            \le\indicator{R(\Delta)}(\fr_1)\prefactor{eqn: bound on the first-order derivative of the finite-volume polymer weight}\sum_{X_1\in\cT_\ast(\Lambda)}\abs{X_1\cap\bL}e^{-\Peierls{eqn: bound on the first-order derivative of the finite-volume polymer weight}(\abs{X_1\cap\bL}-1)}e^{a(X_1)}e^{-\spatial{eqn: bound on the first-order derivative of the finite-volume contour-polymer factor}[\fn_1(X_1)-1]}.
        \end{multlined}
    \end{equation}
    Applying Lem.~\ref{itm: polymer sum with distance decay from one point} completes the proof provided that
    \begin{equation}
        \label{eqn: closure condition for the bound on the first-order derivative of the log polymer partition function}
        \prefactor{eqn: bound on the first-order derivative of the finite-volume polymer weight}e^{\Peierls{eqn: bound on the first-order derivative of the finite-volume polymer weight}}c_1(e^{-(\Peierls{eqn: bound on the first-order derivative of the finite-volume polymer weight}-\alpha)})s_1(e^{-\spatial{eqn: bound on the first-order derivative of the finite-volume contour-polymer factor}})\le \prefactor{eqn: inductive bound on the first-order derivative of the log polymer partition function}.
    \end{equation}

    \subparagraph{Proof of~\eqref{eqn: inductive bound on the second-order derivative of the log polymer partition function}}

    We bound $\bars{\cD_1\cD_2\log\frac{\Xi_{\bz_\Delta}(\Lambda\mid q)}{\Xi_{\bz_\Delta}^q(\Lambda)}}$ by
    \begin{align*}
        {}&
        \sum_{X_1\in\cT_\ast(\Lambda)}\abs{\cD_1\cD_2K_q^{(\Lambda)}(X_1;\bz_\Delta)}\sum_{\substack{U\in\cU(\Lambda)\\U\ni X_1}}\abs{\varphi(U)}n_{X_1}(U)\prod_{X\in U\setminus\mset{X_1}}\abs{K_q^{(\Lambda)}(X;\bz_\Delta)}
        \\
        {}&+
        \sum_{X_1\in\cT_\ast(\Lambda)}\prod_{i=1,2}\abs{\cD_i K_q^{(\Lambda)}(X_1;\bz_\Delta)} \sum_{\substack{U\in\cU(\Lambda)\\U\supset\mset{X_1,X_1}}}\abs{\varphi(U)}n_{X_1}(U)[n_{X_1}(U)-1]\prod_{X\in U\setminus\mset{X_1,X_1}}\abs{K_q^{(\Lambda)}(X;\bz_\Delta)}
        \\
        {}&+\sum_{(X_1,X_2)\subset\cT_\ast(\Lambda)}\prod_{i=1,2}\abs{\cD_i K_q^{(\Lambda)}(X_i;\bz_\Delta)}\sum_{\substack{U\in\cU(\Lambda)\\U\supset\mset{X_1,X_2}}}\abs{\varphi(U)}n_{X_1}(U)n_{X_2}(U)\prod_{X\in U\setminus\mset{X_1,X_2}}\abs{K_q^{(\Lambda)}(X;\bz_\Delta)}.
        \tag{\stepcounter{equation}\theequation}
    \end{align*}
    To extract a penalty factor for the separation between $X_1$ and $X_2$, we use the identity
    \begin{equation}
        \begin{multlined}
            \prod_{X\in U\setminus\mset{X_1,X_2}}\abs{K_q^{(\Lambda)}(X;\bz_\Delta)}
            =e^{-\sum_{X\in U\setminus\mset{X_1,X_2}}b(X)} \prod_{X\in U\setminus\mset{X_1,X_2}}e^{b(X)}\abs{K_q^{(\Lambda)}(X;\bz_\Delta)}
        \end{multlined}
    \end{equation}
    and the observation that, provided that $\varphi(U)\ne 0$,
    \begin{equation}
        \sum_{X\in U\setminus\mset{X_1,X_2}}b(X)
        \ge \alpha\{\rmd_\bL(X_1,X_2)-1\}_+.
    \end{equation}
    Hence, by~\eqref{eqn: cluster bound with one anchor}, \eqref{eqn: cluster bound with two identical anchors}, \eqref{eqn: cluster bound with two distinct anchors}, and Props.~\ref{itm: bound on the first-order derivative of the finite-volume polymer weight} and~\ref{itm: bound on the second-order derivative of the finite-volume polymer weight}, we can bound the above by $\indicator{R(\Delta)}(\fr_1,\fr_2)$ times
    \begin{align}
        {}& \prefactor{eqn: bound on the second-order derivative of the finite-volume polymer weight}\sum_{X_1\in\cT_\ast(\Lambda)}\abs{X_1\cap\bL}^2 e^{-\Peierls{eqn: bound on the second-order derivative of the finite-volume polymer weight}(\abs{X_1\cap\bL}-2)}e^{a(X_1)}\brackets[\Big]{e^{-\spatial{eqn: bound on the second-order derivative of the finite-volume polymer weight}[\fn_{12}(X_1)-2]}+e^{-\spatial{eqn: bound on the second-order derivative of the finite-volume polymer weight}[\fn_1(X_1)+\fn_2(X_1)-2]}}
        \nonumber
        \\
        {}&+ 2\prefactor{eqn: bound on the first-order derivative of the finite-volume polymer weight}^2 \sum_{X_1\in\cT_\ast(\Lambda)}
        \abs{X_1\cap\bL}^2 e^{-2\Peierls{eqn: bound on the first-order derivative of the finite-volume polymer weight}(\abs{X_1\cap\bL}-1)}e^{3a(X_1)} e^{-\spatial{eqn: bound on the first-order derivative of the finite-volume contour-polymer factor}[\fn_1(X_1)+\fn_2(X_1)-2]}
        \nonumber
        \\
        {}&+2\prefactor{eqn: bound on the first-order derivative of the finite-volume polymer weight}^2\sum_{(X_1,X_2)\subset\cT_\ast(\Lambda)}e^{-\alpha\{\rmd_\bL(X_1,X_2)-1\}_+}
        \prod_{i=1,2}\abs{X_i\cap\bL}e^{-\Peierls{eqn: bound on the first-order derivative of the finite-volume polymer weight}(\abs{X_i\cap\bL}-1)}
        e^{\frac{3}{2}a(X_i)}
        e^{-\spatial{eqn: bound on the first-order derivative of the finite-volume contour-polymer factor}[\fn_i(X_i)-1]}.
        \label{eqn: bound on the second-order derivative of the log polymer partition function}
    \end{align}
    Using Lems.~\ref{itm: polymer sum with combined distance decay from two points},~\ref{itm: polymer sum with separate distance decay from two points} and~\ref{itm: double polymer sum with distance decay from two points}, we bound the above by
    \begin{align*}
        {}& \prefactor{eqn: bound on the second-order derivative of the finite-volume polymer weight}e^{2\Peierls{eqn: bound on the second-order derivative of the finite-volume polymer weight}}\max\set{e^{-(\Peierls{eqn: bound on the second-order derivative of the finite-volume polymer weight}-\alpha)},e^{-\spatial{eqn: bound on the second-order derivative of the finite-volume polymer weight}}}^{(1-\free{itm: polymer sum with combined distance decay from two points})(\fd_{12}-2)}c_2(e^{-(\Peierls{eqn: bound on the second-order derivative of the finite-volume polymer weight}-\alpha)})s_2(e^{-\free{itm: polymer sum with combined distance decay from two points}\spatial{eqn: bound on the second-order derivative of the finite-volume polymer weight}})
        \\
        {}&+\prefactor{eqn: bound on the second-order derivative of the finite-volume polymer weight}e^{2\Peierls{eqn: bound on the second-order derivative of the finite-volume polymer weight}}
        \max\set{e^{-(\Peierls{eqn: bound on the second-order derivative of the finite-volume polymer weight}-\alpha)},e^{-\spatial{eqn: bound on the second-order derivative of the finite-volume polymer weight}}}^{(1-\free{itm: polymer sum with separate distance decay from two points})(\fd_{12}-2)}c_2(e^{-\free{itm: polymer sum with separate distance decay from two points}(\Peierls{eqn: bound on the second-order derivative of the finite-volume polymer weight}-\alpha)})s_1(e^{-\free{itm: polymer sum with separate distance decay from two points}\spatial{eqn: bound on the second-order derivative of the finite-volume polymer weight}})
        \\
        {}&+2\prefactor{eqn: bound on the first-order derivative of the finite-volume polymer weight}^2 e^{2\Peierls{eqn: bound on the first-order derivative of the finite-volume polymer weight}}
        \max\set{e^{-(2\Peierls{eqn: bound on the first-order derivative of the finite-volume polymer weight}-3\alpha)},e^{-\spatial{eqn: bound on the first-order derivative of the finite-volume contour-polymer factor}}}^{(1-\free{itm: polymer sum with separate distance decay from two points})(\fd_{12}-2)}c_2(e^{-\free{itm: polymer sum with separate distance decay from two points}(2\Peierls{eqn: bound on the first-order derivative of the finite-volume polymer weight}-3\alpha)})s_1(e^{-\free{itm: polymer sum with separate distance decay from two points}\spatial{eqn: bound on the first-order derivative of the finite-volume contour-polymer factor}})
        \\
        {}&+2\prefactor{eqn: bound on the first-order derivative of the finite-volume polymer weight}^2 e^{2\Peierls{eqn: bound on the first-order derivative of the finite-volume polymer weight}+\alpha}
        \max\set{e^{-(\Peierls{eqn: bound on the first-order derivative of the finite-volume polymer weight}-\frac{3}{2}\alpha)},e^{-\spatial{eqn: bound on the first-order derivative of the finite-volume contour-polymer factor}},e^{-\alpha}}^{(1-\free{itm: double polymer sum with distance decay from two points})(\fd_{12}-2)}c_1(e^{-\free{itm: double polymer sum with distance decay from two points}(\Peierls{eqn: bound on the first-order derivative of the finite-volume polymer weight}-\frac{3}{2}\alpha)})^2 s_1(e^{-\free{itm: double polymer sum with distance decay from two points}\spatial{eqn: bound on the first-order derivative of the finite-volume contour-polymer factor}})^2,
        \tag{\stepcounter{equation}\theequation}
    \end{align*}
    which completes the proof provided that
    \begin{equation}
        \label{eqn: closure condition for the prefactor in the bound on the second-order derivative of the log polymer partition function}
        \begin{multlined}
            \prefactor{eqn: bound on the second-order derivative of the finite-volume polymer weight}e^{2\Peierls{eqn: bound on the second-order derivative of the finite-volume polymer weight}}
            \brackets[\Big]{c_2(e^{-(\Peierls{eqn: bound on the second-order derivative of the finite-volume polymer weight}-\alpha)})s_2(e^{-\free{itm: polymer sum with combined distance decay from two points}\spatial{eqn: bound on the second-order derivative of the finite-volume polymer weight}})
            +c_2(e^{-\free{itm: polymer sum with separate distance decay from two points}(\Peierls{eqn: bound on the second-order derivative of the finite-volume polymer weight}-\alpha)})s_1(e^{-\free{itm: polymer sum with separate distance decay from two points}\spatial{eqn: bound on the second-order derivative of the finite-volume polymer weight}})}
            \\
            +2\prefactor{eqn: bound on the first-order derivative of the finite-volume polymer weight}^2 e^{2\Peierls{eqn: bound on the first-order derivative of the finite-volume polymer weight}}
            \brackets[\Big]{c_2(e^{-\free{itm: polymer sum with separate distance decay from two points}(2\Peierls{eqn: bound on the first-order derivative of the finite-volume polymer weight}-3\alpha)})s_1(e^{-\free{itm: polymer sum with separate distance decay from two points}\spatial{eqn: bound on the first-order derivative of the finite-volume contour-polymer factor}})
            +e^{\alpha}c_1(e^{-\free{itm: double polymer sum with distance decay from two points}(\Peierls{eqn: bound on the first-order derivative of the finite-volume polymer weight}-\frac{3}{2}\alpha)})^2 s_1(e^{-\free{itm: double polymer sum with distance decay from two points}\spatial{eqn: bound on the first-order derivative of the finite-volume contour-polymer factor}})^2}
            \le \prefactor{eqn: inductive bound on the second-order derivative of the log polymer partition function}
        \end{multlined}
    \end{equation}
    and
    \begin{equation}
        \label{eqn: closure condition for the spatial decay rate in the bound on the second-order derivative of the log polymer partition function}
        \begin{multlined}
            \max\Big\{
                e^{-(1-\free{itm: polymer sum with combined distance decay from two points})(\Peierls{eqn: bound on the second-order derivative of the finite-volume polymer weight}-\alpha)},
                e^{-(1-\free{itm: polymer sum with combined distance decay from two points})\spatial{eqn: bound on the second-order derivative of the finite-volume polymer weight}},
                e^{-(1-\free{itm: polymer sum with separate distance decay from two points})(\Peierls{eqn: bound on the second-order derivative of the finite-volume polymer weight}-\alpha)},
                e^{-(1-\free{itm: polymer sum with separate distance decay from two points})\spatial{eqn: bound on the second-order derivative of the finite-volume polymer weight}},
                \\
                e^{-(1-\free{itm: polymer sum with separate distance decay from two points})(2\Peierls{eqn: bound on the first-order derivative of the finite-volume polymer weight}-3\alpha)},
                e^{-(1-\free{itm: polymer sum with separate distance decay from two points})\spatial{eqn: bound on the first-order derivative of the finite-volume contour-polymer factor}},
                e^{-(1-\free{itm: double polymer sum with distance decay from two points})(\Peierls{eqn: bound on the first-order derivative of the finite-volume polymer weight}-\frac{3}{2}\alpha)},
                e^{-(1-\free{itm: double polymer sum with distance decay from two points})\spatial{eqn: bound on the first-order derivative of the finite-volume contour-polymer factor}},
                e^{-(1-\free{itm: double polymer sum with distance decay from two points})\alpha}
            \Big\}
            \le e^{-\spatial{eqn: inductive bound on the second-order derivative of the log polymer partition function}}.
        \end{multlined}
    \end{equation}

    \subparagraph{Proof of~\eqref{eqn: inductive bound on the log flipping term}}

    We start by changing $\bz$ to $z$.
    By Taylor's theorem, there exists $\tilde{\bz}$ such that
    \begin{equation}
        \log\frac{\Xi_{\bz_\Delta}(\Lambda\mid -q)}{\Xi_{\bz_\Delta}(\Lambda\mid q)}
        =\log\frac{\Xi_{z_\Delta}(\Lambda\mid -q)}{\Xi_{z_\Delta}(\Lambda\mid q)}
        +\sum_{i=1}^n\indicator{\Delta}(\fr_i)\eval{\partial_i\log\frac{\Xi_{\bz_\Delta}(\Lambda\mid -q)}{\Xi_{\bz_\Delta}(\Lambda\mid q)}}_{\tilde{\bz}}(\bz_i-z).
    \end{equation}
    We bound
    \begin{equation}
        \bars*{\partial_i\log\frac{\Xi_{\bz_\Delta}(\Lambda\mid -q)}{\Xi_{\bz_\Delta}(\Lambda\mid q)}}
        \le \bars*{\partial_i\log\frac{\Xi_{\bz_\Delta}(\Lambda\mid -q)}{\Xi_{\bz_\Delta}^{-q}(\Lambda)}}
        +\bars*{\partial_i\log\frac{\Xi_{\bz_\Delta}(\Lambda\mid q)}{\Xi^{q}_{\bz_\Delta}(\Lambda)}}
        +\bars*{\partial_i\log\frac{\Xi_{\bz_\Delta}^{-q}(\Lambda)}{\Xi_{\bz_\Delta}^q(\Lambda)}}.
    \end{equation}
    The first two terms are bounded by~\eqref{eqn: inductive bound on the first-order derivative of the log polymer partition function}.
    The last term is bounded using Lem.~\ref{itm: expansion of the first-order derivative of the log q-oriented partition function}.
    Combining these estimates, we conclude that
    \begin{equation}
        \label{eqn: bound on the error of fugacity homogenization in the log flipping term}
        \bars*{\log\frac{\Xi_{\bz_\Delta}(\Lambda\mid -q)}{\Xi_{\bz_\Delta}(\Lambda\mid q)}
        -\log\frac{\Xi_{z_\Delta}(\Lambda\mid -q)}{\Xi_{z_\Delta}(\Lambda\mid q)}}
        \le n(e^{\nu/n}-1)\brackets[\Big]{2\prefactor{eqn: inductive bound on the first-order derivative of the log polymer partition function}+(1+e\eps)z}.
    \end{equation}
    Hence, from here on, we deal exclusively with constant $z$.

    By Lem.~\ref{lem: the truncated polymer model can be cluster expanded}, we can apply the cluster expansion to the polymer representation of the model in $\Lambda$ to get
    \begin{equation}
        \log\frac{\hat{\Xi}_{z_\Delta}(\Lambda\mid q)}{\Xi_{z_\Delta}^q(\Lambda)}
        =\sum_{U\in\cU(\Lambda)}\varphi(U)\prod_{X\in U}\hat{K}_q^{(\Lambda)}(X;z_\Delta),
    \end{equation}
    so that
    \begin{equation}
        \log\frac{\hat{\Xi}_{z_\Delta}(\Lambda\mid -q)}{\hat{\Xi}_{z_\Delta}(\Lambda\mid q)}
        =\log\frac{\Xi_{z_\Delta}^{-q}(\Lambda)}{\Xi_{z_\Delta}^q(\Lambda)}
        +\sum_{U\in\cU(\Lambda)}\varphi(U)\left[\prod_{X\in U}\hat{K}_{-q}^{(\Lambda)}(X;z_\Delta)-\prod_{X\in U}\hat{K}_q^{(\Lambda)}(X;z_\Delta)\right].
    \end{equation}

    By the same argument as in the base case, we bound the first term by
    \begin{equation}
        \label{eqn: difference of q-oriented partition functions is a boundary term}
        \bars[\bigg]{\log\frac{\Xi_{z_\Delta}^{-q}(\Lambda)}{\Xi_{z_\Delta}^q(\Lambda)}}
        \le 50e(z\ell^2)\eps\abs{\bdex[\bL]{\Lambda}\cap\bL}.
    \end{equation}
    Next, we write the second term as the sum of
    \begin{equation}
        \label{eqn: difference of bulk terms in the cluster expansion of the polymer partition functions}
        \sum_{U\in\cU(\Lambda)}\varphi(U)\left[\prod_{X\in U}\hat{K}_{-q}(X;z)-\prod_{X\in U}\hat{K}_q(X;z)\right]
    \end{equation}
    and
    \begin{equation}
        \label{eqn: finite-volume corrections to the bulk terms in the cluster expansion of the polymer partition functions}
        \sum_{\sigma\in\set{-1,1}}\sigma\sum_{U\in\cU(\Lambda)}\varphi(U)\left[\prod_{X\in U}\hat{K}_{\sigma q}(X;z)-\prod_{X\in U}\hat{K}_{\sigma q}^{(\Lambda)}(X;z_\Delta)\right].
    \end{equation}
    To bound~\eqref{eqn: difference of bulk terms in the cluster expansion of the polymer partition functions}, we start by writing
    \begin{equation}
        \begin{multlined}
            \sum_{U\in\cU(\Lambda)}\varphi(U)\prod_{X\in U}\hat{K}_q(X;z)
            =\sum_{v\in(\Lambda\cup\bdex[\bL]{\Lambda})\cap\bL}\sum_{\substack{U\in\cU(\Z^2)\\v\in\supp{U}}}\frac{\varphi(U)}{\abs{\supp{U}\cap\bL}}\prod_{X\in U}\hat{K}_q(X;z)
            \\
            -\sum_{v\in(\Lambda\cup\bdex[\bL]{\Lambda})\cap\bL}\sum_{\substack{U\in\cU(\Z^2)\setminus\cU(\Lambda)\\v\in\supp{U}}}\frac{\varphi(U)}{\abs{\supp{U}\cap\bL}}\prod_{X\in U}\hat{K}_q(X;z).
        \end{multlined}
    \end{equation}
    By symmetry, the first inner sum is independent of $v\in(\Lambda\cup\bdex[\bL]{\Lambda})\cap\bL$ and $q$ and thus cancels.
    The second double sum results in a boundary term as follows:
    \begin{align}
        {}&\sum_{v\in(\Lambda\cup\bdex[\bL]{\Lambda})\cap\bL}\sum_{\substack{U\in\cU(\Z^2)\setminus\cU(\Lambda)\\v\in\supp{U}}}\frac{\abs{\varphi(U)}}{\abs{\supp{U}\cap\bL}}\prod_{X\in U}\abs{\hat{K}_q(X;z)}
        \nonumber
        \\
        \le{}& \sum_{v\in\bdex[\bL]{\Lambda}\cap\bL}
        \sum_{\substack{X_1\in\cT_\ast(\Z^2)\\X_1\supset T_v}}\abs{\hat{K}_q(X_1;z)}
        \sum_{\substack{U\in\cU(\Z^2)\\U\ni X_1}}\abs{\varphi(U)}\prod_{X\in U\setminus\mset{X_1}}\abs{\hat{K}_q(X;z)}
        \nonumber
        \\
        \le{}& c_0(e^{-(\Peierls{eqn: bound on the finite-volume polymer weight}-\alpha)})\bars{\bdex[\bL]{\Lambda}\cap\bL},
        \label{eqn: boundary error after bulk cancellation}
    \end{align}
    having used~\eqref{eqn: cluster bound with one anchor} and Lem.~\ref{itm: anchored polymer sum} in the last step.
    It remains to bound~\eqref{eqn: finite-volume corrections to the bulk terms in the cluster expansion of the polymer partition functions}.
    By Props.~\ref{itm: bound on the finite- and infinite-volume polymer weights} and~\ref{itm: bound on the finite-volume correction to the polymer weight}, it suffices to estimate
    \begin{align*}
        {}&\sum_{U\in\cU(\Lambda)}\abs{\varphi(U)}\abs{\prod_{X\in U}\hat{K}_q^{(\Lambda)}(X;z_\Delta)-\prod_{X\in U}\hat{K}_q(X;z)}
        \\
        \le{}&\sum_{U\in\cU(\Lambda)}\abs{\varphi(U)}\sum_{i=1}^n\prod_{j=1}^{i-1}\abs{\hat{K}_q(X_j;z)}\abs{\hat{K}_q^{(\Lambda)}(X_i;z_\Delta)-\hat{K}_q(X_i;z)}\prod_{j=i+1}^n\abs{\hat{K}_q^{(\Lambda)}(X_j;z_\Delta)}
        \\
        \le{}&\prefactor{eqn: bound on the finite-volume correction to the polymer weight}\sum_{U\in\cU(\Lambda)}\abs{\varphi(U)}\sum_{X_1\in U}\abs{X_1\cap\bL}e^{-\Peierls{eqn: bound on the finite-volume correction to the polymer weight}(\abs{X_1\cap\bL}-1)}e^{-\spatial{eqn: bound on the finite-volume correction to the polymer weight}\rmd_\bL(X_1,\bdin[2\bL]{\Lambda})}\prod_{X\in U\setminus\mset{X_1}}e^{-\Peierls{eqn: bound on the finite-volume polymer weight}\abs{X\cap\bL}}.
        \tag{\stepcounter{equation}\theequation}
    \end{align*}
    We rewrite the above bound as
    \begin{equation}
        \label{eqn: boundary error from finite-volume corrections to the bulk terms}
        \begin{multlined}
            \prefactor{eqn: bound on the finite-volume correction to the polymer weight}\sum_{X_1\in\cT_\ast(\Lambda)}\abs{X_1\cap\bL}e^{-\Peierls{eqn: bound on the finite-volume correction to the polymer weight}(\abs{X_1\cap\bL}-1)}e^{-\spatial{eqn: bound on the finite-volume correction to the polymer weight}\rmd_\bL(X_1,\bdin[2\bL]{\Lambda})}\sum_{\substack{U\in\cU(\Lambda)\\U\ni X_1}}\abs{\varphi(U)}n_{X_1}(U)\prod_{X\in U\setminus\mset{X_1}}e^{-\Peierls{eqn: bound on the finite-volume polymer weight}\abs{X\cap\bL}}
            \\
            \le 25\prefactor{eqn: bound on the finite-volume correction to the polymer weight} e^{\Peierls{eqn: bound on the finite-volume correction to the polymer weight}} c_1(e^{-(\Peierls{eqn: bound on the finite-volume correction to the polymer weight}-\alpha)})s(e^{-\spatial{eqn: bound on the finite-volume correction to the polymer weight}})\abs{\bdex[\bL]{\Lambda}\cap\bL},
        \end{multlined}
    \end{equation}
    having used~\eqref{eqn: cluster bound with one anchor}, Lem.~\ref{itm: polymer sum with distance decay from region}, and the bound $\abs{\bdin[2\bL]{\Lambda}\cap\bL}\le 25\abs{\bdex[\bL]{\Lambda}\cap\bL}$ in the inequality.

    Combining~\eqref{eqn: bound on the error of fugacity homogenization in the log flipping term}, \eqref{eqn: difference of q-oriented partition functions is a boundary term}, \eqref{eqn: boundary error after bulk cancellation}, and~\eqref{eqn: boundary error from finite-volume corrections to the bulk terms}, we conclude the proof of~\eqref{eqn: inductive bound on the log flipping term} provided that
    \begin{equation}
        \label{eqn: closure condition for the bound on the log flipping term}
        \begin{multlined}
            n(e^{\nu/n}-1)\brackets[\Big]{2\prefactor{eqn: inductive bound on the first-order derivative of the log polymer partition function}+(1+e\eps)z}
            +50e^{1+\nu}(z\ell^2)\eps
            +2c_0(e^{-(\Peierls{eqn: bound on the finite-volume polymer weight}-\alpha)})
            \\
            +50\prefactor{eqn: bound on the finite-volume correction to the polymer weight} e^{\Peierls{eqn: bound on the finite-volume correction to the polymer weight}} c_1(e^{-(\Peierls{eqn: bound on the finite-volume correction to the polymer weight}-\alpha)})s(e^{-\spatial{eqn: bound on the finite-volume correction to the polymer weight}})
            \le \tau/2.
        \end{multlined}
    \end{equation}

    The conditions on the parameters are summarized in Appx.~\ref{appx: convergence and closure conditions}; we will check in Sec.~\ref{sec: explicit choices of parameters} that they can be simultaneously satisfied.
    Pending this verification, the proposition is proven.
\end{proof}

\section{Nematic order}

\subsection{Explicit choices of parameters}
\label{sec: explicit choices of parameters}

In this subsection, we estimate the aspect ratio $\fk$ required for the formation of a nematic phase in the hard rectangle model by analyzing the constrained optimization problem derived in the previous sections and summarized in Appx.~\ref{appx: constants and parameters}.
We will be intentionally crude in our calculations, since the point is not to obtain the best possible bound but rather to give one that is both explicit and easily checked.
A full, computer-based analysis is likely to yield better bounds; see Rem.~\ref{rem: reuse of free parameters} for another possible point of improvement.

\begin{proposition}
    \label{prop: explicit choices of parameters}
    Let $K:=\frac{z\ell^2}{1000}$, $K_0:=150$, and $\nu=0$.
    Suppose that $K>K_0$ and $e^{-2K}<\eps<e^{-K}$.
    Then, the conditions in Appx.~\ref{appx: convergence and closure conditions} are simultaneously satisfied with
    \begin{equation}
        \begin{multlined}
            n\text{ arbitrary},\;
            \free{itm: lower bound on m(delta)-1}=\free{itm: lower bound on m_12(delta)-2}=\free{itm: lower bound on m_Gamma^(Lambda)(delta)-1}=\frac{3}{50},\;
            \free{itm: polymer sum with combined distance decay from two points}=\free{itm: polymer sum with separate distance decay from two points}=\free{itm: double polymer sum with distance decay from two points}=\frac{1}{2},\;
            \tau=\frac{K}{10},
            \\
            \spatial{eqn: inductive bound on the second-order derivative of the log polymer partition function}=\frac{K}{3000},\;
            \prefactor{eqn: inductive bound on the first-order derivative of the log polymer partition function}=\frac{z}{4},\;
            \prefactor{eqn: inductive bound on the second-order derivative of the log polymer partition function}=z^2,\;
            \text{and }
            \alpha=\frac{K}{1000}.
        \end{multlined}
    \end{equation}
\end{proposition}

\begin{remark}
    The choice of $\nu=0$ apparently runs counter to the requirement that $\nu>0$.
    However, it does not cause any issues since the resulting identities and \textit{strict} inequalities are preserved by taking $\nu$ to be sufficiently small with respect to a fixed $n$.
\end{remark}

\begin{proof}[Proof of Prop.~\ref{prop: explicit choices of parameters}]
    The assumptions and choices imply the following identities and bounds:
    \begin{equation}
        e^{-\Peierls{eqn: bound on the finite-volume interaction-polymer factor}}
        =e^{-\Peierls{eqn: bound on the finite-volume correction to the interaction-polymer factors}}
        =e^{-\Peierls{eqn: bound on the second-order derivative of the finite-volume interaction-polymer factor}}
        =\eps^{3/50},
    \end{equation}
    \begin{equation}
        e^{-\spatial{eqn: bound on the finite-volume correction to the interaction-polymer factors}}=e^{-\spatial{eqn: bound on the finite-volume correction to the polymer weight}}=\eps^{1/150},
    \end{equation}
    \begin{equation}
        \Peierls{eqn: Peierls bound for contour weight}>1.15K,\quad
        \Peierls{eqn: bound on the truncated finite-volume contour-polymer factor},\Peierls{eqn: bound on the finite-volume correction to the contour-polymer factor}>1.04K,\quad
        \Peierls{eqn: bound on the first-order derivative of the finite-volume contour-polymer factor}>0.52K,
    \end{equation}
    \begin{equation}
        e^{-\spatial{eqn: bound on the first-order derivative of the finite-volume contour-polymer factor}}
        =e^{-\spatial{eqn: bound on the second-order derivative of the finite-volume contour-polymer factor}}
        =e^{-\spatial{eqn: bound on the second-order derivative of the finite-volume polymer weight}}
        =\eps,
    \end{equation}
    \begin{equation}
        c\bigg(\eps;\frac{7}{50},\frac{11}{5},2\bigg)<5\cdot 10^{32},\quad
        c\bigg(\eps;\frac{103}{750},\frac{824}{375},2\bigg)<5\cdot 10^{34},
    \end{equation}
    \begin{equation}
        c\bigg(\eps;\frac{7}{50},\frac{21}{5},2\bigg)<2\cdot 10^{64},\quad
        c\bigg(\eps;\frac{7}{25},\frac{22}{5},4\bigg)<2\cdot 10^{37},
    \end{equation}
    \begin{equation}
        \loss{eqn: bound on the finite-volume interaction-polymer factor},
        \loss{eqn: bound on the finite-volume correction to the interaction-polymer factors}
        <10^{39}K\eps,\quad
        \prefactor{eqn: bound on the finite-volume correction to the interaction-polymer factors}<10^{41}K\eps,\quad
        \prefactor{eqn: bound on the first-order derivative of the finite-volume interaction-polymer factor}<10^{35}z\eps,\quad
        \prefactor{eqn: bound on the second-order derivative of the finite-volume interaction-polymer factor}<10^{69}z^2\eps^2,
    \end{equation}
    \begin{equation}
        \prefactor{eqn: bound on the first-order derivative of the contour weight - localization},\prefactor{eqn: bound on the first-order derivative of the contour weight - decay}<3z,\quad
        \prefactor{eqn: bound on the second-order derivative of the contour weight}<64z^2,\quad
        \prefactor{eqn: bound on the finite-volume correction to the contour-polymer factor}<10^4K,\quad
        \prefactor{eqn: bound on the first-order derivative of the finite-volume contour-polymer factor}<8z,\quad
        \prefactor{eqn: bound on the second-order derivative of the finite-volume contour-polymer factor}<128z^2,
    \end{equation}
    \begin{equation}
        \eps^{3/50}<
        e^{-\Peierls{eqn: bound on the finite-volume polymer weight}},\;
        e^{-\Peierls{eqn: bound on the finite-volume correction to the polymer weight}},\;
        e^{-\Peierls{eqn: bound on the first-order derivative of the finite-volume polymer weight}},\;
        e^{-\Peierls{eqn: bound on the second-order derivative of the finite-volume polymer weight}}
        <1.01\eps^{3/50}.
    \end{equation}

    We now verify the conditions in Appx.~\ref{appx: convergence and closure conditions}.
    Conds.~\ref{cond:1}, \ref{cond:2}, \ref{cond:3}, \ref{cond:4}, \ref{cond:5}, \ref{cond:6}, and~\ref{cond:7} are immediate.
    Conds.~\ref{cond:9}, \ref{cond:11}, \ref{cond:13}, \ref{cond:14}, \ref{cond:15}, \ref{cond:16}, \ref{cond:18}, \ref{cond:19}, and the second part of Cond.~\ref{cond:17} follow from the additional bounds
    \begin{equation}
        \Peierls{eqn: bound on the finite-volume polymer weight}-2\alpha,\;
        \Peierls{eqn: bound on the finite-volume correction to the polymer weight}-\alpha,\;
        \Peierls{eqn: bound on the first-order derivative of the finite-volume polymer weight}-\alpha,\;
        \Peierls{eqn: bound on the second-order derivative of the finite-volume polymer weight}-\alpha,\;
        \frac{1}{2}(2\Peierls{eqn: bound on the first-order derivative of the finite-volume polymer weight}-3\alpha)
        >\frac{K}{20},\quad
        \frac{1}{2}(\Peierls{eqn: bound on the second-order derivative of the finite-volume polymer weight}-\alpha),\;
        \frac{1}{2}\bigg(\Peierls{eqn: bound on the first-order derivative of the finite-volume polymer weight}-\frac{3}{2}\alpha\bigg)
        >\frac{K}{35}.
    \end{equation}
    For Cond.~\ref{cond:8}, 
    we estimate
    \begin{equation}
        9c_0(e^{-(\Peierls{eqn: bound on the finite-volume polymer weight}-2\alpha)})
        < 600 e^{-29K/500}
        < \frac{K}{1000}=\alpha.
    \end{equation}
    For Cond.~\ref{cond:10}, we have
    \begin{equation}
        50e(z\ell^2)\eps
        \le 50e(1000K)e^{-K}
        <\frac{K}{20}=\frac{\tau}{2}.
    \end{equation}
    For Cond.~\ref{cond:12}, using
    \begin{equation}
        \label{eqn: bound on phi13}
        \prefactor{eqn: bound on the first-order derivative of the finite-volume polymer weight}
        <z(16e^{-0.519K}+2\cdot 10^{35}e^{-2.03K})
        <ze^{-0.49K}
    \end{equation}
    and the auxiliary bounds
    \begin{equation}
        e^{\Peierls{eqn: bound on the first-order derivative of the finite-volume polymer weight}}\le\eps^{-3/50}<e^{3K/25},\quad
        c_1(e^{-(\Peierls{eqn: bound on the first-order derivative of the finite-volume polymer weight}-\alpha)})<1,\quad
        s_1(e^{-\spatial{eqn: bound on the first-order derivative of the finite-volume contour-polymer factor}})<30,
    \end{equation}
    we have
    \begin{equation}
        \prefactor{eqn: bound on the first-order derivative of the finite-volume polymer weight}e^{\Peierls{eqn: bound on the first-order derivative of the finite-volume polymer weight}}c_1(e^{-(\Peierls{eqn: bound on the first-order derivative of the finite-volume polymer weight}-\alpha)})s_1(e^{-\spatial{eqn: bound on the first-order derivative of the finite-volume contour-polymer factor}})
        <30z e^{-0.37K}
        <\frac{z}{4}
        =\prefactor{eqn: inductive bound on the first-order derivative of the log polymer partition function}.
    \end{equation}
    For the first part of Cond.~\ref{cond:17}, using~\eqref{eqn: bound on phi13},
    \begin{equation}
        \prefactor{eqn: bound on the second-order derivative of the finite-volume polymer weight}
        <z^2e^{-0.48K},
    \end{equation}
    and the auxiliary bounds
    \begin{equation}
        c_2(e^{-(\Peierls{eqn: bound on the second-order derivative of the finite-volume polymer weight}-\alpha)})<1,\quad
        c_2(e^{-\frac{1}{2}(\Peierls{eqn: bound on the second-order derivative of the finite-volume polymer weight}-\alpha)})<120,\quad
        c_2(e^{-\frac{1}{2}(2\Peierls{eqn: bound on the first-order derivative of the finite-volume polymer weight}-3\alpha)})<1,\quad
        e^{\alpha}c_1(e^{-\frac{1}{2}(\Peierls{eqn: bound on the first-order derivative of the finite-volume polymer weight}-\frac{3}{2}\alpha)})^2<500,
    \end{equation}
    \begin{equation}
        s_1(e^{-\frac{1}{2}\spatial{eqn: bound on the first-order derivative of the finite-volume contour-polymer factor}})<30,\quad
        s_2(e^{-\frac{1}{2}\spatial{eqn: bound on the second-order derivative of the finite-volume polymer weight}})<90,
    \end{equation}
    we bound the LHS of~\eqref{eqn: closure condition for the prefactor in the bound on the second-order derivative of the log polymer partition function} by
    \begin{equation}
        4\cdot 10^4 z^2 e^{-0.24K}+10^6 z^2 e^{-0.74K}
        <z^2
        =\prefactor{eqn: inductive bound on the second-order derivative of the log polymer partition function}.
    \end{equation}
    Finally, for Cond.~\ref{cond:20}, we have
    \begin{equation}
        \prefactor{eqn: bound on the finite-volume correction to the polymer weight} e^{\Peierls{eqn: bound on the finite-volume correction to the polymer weight}} <1.01,\quad
        2c_0(e^{-(\Peierls{eqn: bound on the finite-volume polymer weight}-\alpha)})<e^{-K/40},\quad
        c_1(e^{-(\Peierls{eqn: bound on the finite-volume correction to the polymer weight}-\alpha)})<0.01,\quad
        s(e^{-\spatial{eqn: bound on the finite-volume correction to the polymer weight}})<9,
    \end{equation}
    so the LHS of~\eqref{eqn: closure condition for the bound on the log flipping term} is bounded by
    \begin{equation}
        2\cdot 10^5Ke^{-K}+e^{-K/40}+5 < \frac{K}{20} = \frac{\tau}{2}.
    \end{equation}
    Thus, all the conditions in Appx.~\ref{appx: convergence and closure conditions} are satisfied.
\end{proof}

\subsection{Proof of Theorem~\ref{thm: main}}

In this subsection, we prove Theorem~\ref{thm: main}.
We work in the setting common to Props.~\ref{prop: main} and~\ref{prop: explicit choices of parameters}.

\paragraph{Estimate for the aspect ratio}

It is straightforward to show that the assumption $e^{-2K}<\eps<e^{-K}$ of Prop.~\ref{prop: explicit choices of parameters} is equivalent to the following condition on the range of $z$:
\begin{equation}
    \frac{500}{\ell^2}W\bigg(\frac{\ell^2}{500e(2\fw-1)(2\fl-1)}\bigg)
    <z
    <\frac{1000}{\ell^2}W\bigg(\frac{\ell^2}{1000e(2\fw-1)(2\fl-1)}\bigg),
\end{equation}
where $W$ denotes the principal branch of the Lambert $W$-function.
Again, as is easily checked, this condition is consistent with the assumption $K>K_0$ of Prop.~\ref{prop: explicit choices of parameters} if and only if
\begin{equation}
    W\bigg(\frac{\ell^2}{1000e(2\fw-1)(2\fl-1)}\bigg)>150,
\end{equation}
or, equivalently,
\begin{equation}
    \frac{\ell^2}{1000e(2\fw-1)(2\fl-1)}>150e^{150}.
\end{equation}
Consequently, since
\begin{equation}
    \frac{\ell^2}{(2\fw-1)(2\fl-1)}
    \ge \frac{\fl^2}{4(2\fw-1)(2\fl-1)}
    > \frac{\fl^2}{16\fw\fl}
    = \frac{\fk}{16},
\end{equation}
it suffices that
\begin{equation}
    \fk>(16)(1000e)(150e^{150})
    \approx 9.09\cdot 10^{71}
\end{equation}
to ensure a nonempty range of $z$ satisfying the assumptions of Prop.~\ref{prop: explicit choices of parameters}.

\paragraph{One-point correlation function}

We write the one-point correlation function of the hard rectangle model in $\Lambda$ with $q$-boundary conditions as
\begin{equation}
    \rho^q_{z,\Lambda;1}(\fr_1)
    =\eval{\cD_1\log \Xi_\bz^q(\Lambda)}_z
    +\eval{\cD_1\log\frac{\Xi_\bz(\Lambda\mid q)}{\Xi_\bz^q(\Lambda)}}_z.
\end{equation}
Combining Lem.~\ref{itm: expansion of the first-order derivative of the log q-oriented partition function},~\eqref{eqn: inductive bound on the first-order derivative of the log polymer partition function} of Prop.~\ref{prop: main}, and Prop.~\ref{prop: explicit choices of parameters}, we find that, uniformly in $q$ and $\Lambda$,
\begin{equation}
    \abs{\rho^q_{z,\Lambda;1}(\fr_1)-\indicator{R^q(\Lambda)}(\fr_1)z}
    \le \indicator{R(\Lambda)}(\fr_1)\bigg(\frac{1}{4}+e\eps\bigg)z.
\end{equation}

\paragraph{Two-point correlation function}

We write the two-point correlation function of the hard rectangle model in $\Lambda$ with $q$-boundary conditions as
\begin{equation}
    \rho^q_{z,\Lambda;2}(\fr_1,\fr_2)-\rho^q_{z,\Lambda;1}(\fr_1)\rho^q_{z,\Lambda;1}(\fr_2)
    =\eval{\cD_1\cD_2\log \Xi_\bz^q(\Lambda)}_z
    +\eval{\cD_1\cD_2\log\frac{\Xi_\bz(\Lambda\mid q)}{\Xi_\bz^q(\Lambda)}}_z.
\end{equation}
Combining Lem.~\ref{itm: expansion of the second-order derivative of the log q-oriented partition function},~\eqref{eqn: inductive bound on the second-order derivative of the log polymer partition function} of Prop.~\ref{prop: main}, and Prop.~\ref{prop: explicit choices of parameters}, we find that, uniformly in $q$ and $\Lambda$,
\begin{equation}
    \abs{\rho^q_{z,\Lambda;2}(\fr_1,\fr_2)-\rho^q_{z,\Lambda;1}(\fr_1)\rho^q_{z,\Lambda;1}(\fr_2)}
    \le \indicator{R(\Lambda)}(\fr_1,\fr_2)\brackets[\Big]{e^3\eps^{\fd_{12}-2}+z^2 e^{-(z\ell^2/3000000)(\fd_{12}-2)}}.
\end{equation}

\section*{Acknowledgments}

We thank Deepak Dhar and Ian Jauslin for introducing us to the relevant literature, and Minhao Bai and Joel Lebowitz for helpful comments.
The convenient set of parameter choices in Prop.~\ref{prop: explicit choices of parameters} was found with the help of GPT 5.5 Pro.
The author was supported by an SAS Fellowship and in part by an SGS Conference and Research Travel Award at Rutgers University.

\bibliographystyle{plain}
\bibliography{bibliography.bib}

\appendix

\section{Constants, parameters, and conditions}
\label{appx: constants and parameters}

In this appendix, we collect the constants, parameters, and convergence and closure conditions appearing in the proofs.
Together, they form a sufficient set of conditions for the existence of a nematic phase in the hard rectangle model.

\subsection{Fundamental parameters}
\label{appx: fundamental parameters}

The fundamental parameters are as follows.
\begin{longtable}{l l l}
    \toprule
    \textbf{Parameter(s)} & \textbf{Domain} & \textbf{Location} \\
    \midrule
    \endfirsthead
    
    \toprule
    \textbf{Parameter(s)} & \textbf{Domain} & \textbf{Location} \\
    \midrule
    \endhead
    
    \midrule
    \endfoot
    
    \bottomrule
    \endlastfoot
    
    $\fw$ & $\Z_{\ge 1}$ & \\
    $\fl$ & $\Z_{\ge \max\set{3,2\fw-1}}$ & \\
    $z,\nu$ & $(0,\infty)$ &  \\
    $n$ & $\Z_{\ge 1}$ &  \\
    $\tau,\prefactor{eqn: inductive bound on the first-order derivative of the log polymer partition function},\prefactor{eqn: inductive bound on the second-order derivative of the log polymer partition function},\spatial{eqn: inductive bound on the second-order derivative of the log polymer partition function}$ & $(0,\infty)$ & Prop.~\ref{prop: main} \\
    $\free{itm: lower bound on m(delta)-1}$ & $(0,\frac{1}{16})$ & Lem.~\ref{itm: lower bound on m(delta)-1} \\
    $\free{itm: lower bound on m_12(delta)-2}$ & $(0,\frac{2}{31})$ & Lem.~\ref{itm: lower bound on m_12(delta)-2} \\
    $\free{itm: lower bound on m_Gamma^(Lambda)(delta)-1}$ & $(0,\frac{1}{16})$ & Lem.~\ref{itm: lower bound on m_Gamma^(Lambda)(delta)-1} \\
    $\free{itm: polymer sum with combined distance decay from two points}$ & $(0,1)$ & Lem.~\ref{itm: polymer sum with combined distance decay from two points} \\
    $\free{itm: polymer sum with separate distance decay from two points}$ & $(0,1)$ & Lem.~\ref{itm: polymer sum with separate distance decay from two points} \\
    $\free{itm: double polymer sum with distance decay from two points}$ & $(0,1)$ & Lem.~\ref{itm: double polymer sum with distance decay from two points} \\
    $\alpha$ & $(0,\infty)$ & Lem.~\ref{lem: the truncated polymer model can be cluster expanded} \\
\end{longtable}

\subsection{Derived quantities}
\label{appx: derived quantities}

Recall the shorthand notation $\ell=\ceil{\fl/2}$ and $\eps=e^{1+\nu}(2\fw-1)(2\fl-1)z$.
The constants and derived parameters are as follows.

\begin{longtable}{l l l}
    \toprule
     & \textbf{Expression} & \textbf{Location} \\
    \midrule
    \endfirsthead
    
    \toprule
     & \textbf{Expression} & \textbf{Location} \\
    \midrule
    \endhead
    
    \midrule
    \endfoot
    
    \bottomrule
    \endlastfoot
        
    \multicolumn{3}{l}{\textbf{Combinatorial constants}} \\
    \midrule
    $c_n(\vareps)$ & $\sum_{s=1}^\infty 8^{2s}s^n\vareps^s$ & Lem.~\ref{itm: anchored polymer sum} \\
    $c(\vareps;a,b,C)$ & $9\cdot \sum_{s=1}^\infty 8^{2s}C^{s}\vareps^{\{as-b\}_+}$ & Lem.~\ref{itm: incompatible polymer sum with piecewise linear decay} \\
    $s_k(\delta),k\ge 1$ & $\sum_{v\in\bL}\delta^{\{\ceil{\rmd_\bL(T_v,T_0)/2}-k\}_+}$ & Lem.~\ref{itm: polymer sum with distance decay from one point} \\
    $s(\delta)$ & $\sum_{v\in\bL}\delta^{\rmd_\bL(T_v,T_0)}$ & Lem.~\ref{itm: polymer sum with distance decay from region} \\
    \midrule

    \multicolumn{3}{l}{\textbf{Minimum cluster cardinalities}} \\
    \midrule
    $\fd_{12}$ & $\ceil{\rmd_\bL(\fT_1,\fT_2)/2}+1$ & Lem.~\ref{itm: expansion of the second-order derivative of the log q-oriented partition function} \\
    $\fn_1(\cdot)$ & $\max\set{\ceil{\rmd_\bL(\fT_1,\cdot)/2},1}$ & Prop.~\ref{itm: bound on the first-order derivative of the finite-volume contour self-potential} \\
    $\fn_{12}(\cdot)$ & $\max\set{\fd_{12},\fn_1(\cdot),\fn_2(\cdot)}$ & Prop.~\ref{itm: bound on the second-order derivative of the finite-volume contour self-potential} \\
    $m(\cdot)$ & $\max\set{3,\ceil{(\abs{\cdot\cap\bL}-1)/5}}$ & Lem.~\ref{lem: contour-interaction potential auxiliary lemma} \\
    $m_\Gamma^{(\Lambda)}(\cdot)$ & $\max\set{\ceil{\rmd_\bL(\supp{\Gamma},\bdin[2\bL]{\Lambda})/2},m(\cdot)}$ & Prop.~\ref{itm: bound on the finite-volume correction to the contour-interaction potential} \\
    $\fm_{12}(\cdot)$ & $\max\set{\fd_{12},m(\cdot)}$ & Prop.~\ref{itm: bound on the second-order derivative of the finite-volume contour-interaction potential} \\
    \midrule

    \multicolumn{3}{l}{\textbf{Peierls decay rates}} \\
    \midrule

    $e^{-\Peierls{eqn: bound on the finite-volume interaction-polymer factor}}$ & $\eps^{\free{itm: lower bound on m(delta)-1}}$ & Prop.~\ref{itm: bound on the finite- and infinite-volume interaction-polymer factors} \\

    $e^{-\Peierls{eqn: bound on the finite-volume correction to the interaction-polymer factors}}$ & $\max\set{\eps^{\free{itm: lower bound on m(delta)-1}},\eps^{\free{itm: lower bound on m_Gamma^(Lambda)(delta)-1}}}$ & Prop.~\ref{itm: bound on the finite-volume correction to the interaction-polymer factors} \\

    $e^{-\Peierls{eqn: bound on the second-order derivative of the finite-volume interaction-polymer factor}}$ & $\max\set{\eps^{\free{itm: lower bound on m(delta)-1}},\eps^{\free{itm: lower bound on m_12(delta)-2}}}$ & Prop.~\ref{itm: bound on the second-order derivative of the finite-volume interaction-polymer factor} \\

    $\Peierls{eqn: Peierls bound for contour weight}$ & $\frac{1}{324}\brackets[\big]{\frac{3}{8}(1-16e\eps)(z\ell^2)-\frac{3}{4}\log 3}-2e(z\ell^2)\eps-2\nu$
    & Lem.~\ref{lem: Peierls bound for contour weight}  \\

    $\Peierls{eqn: bound on the truncated finite-volume contour-polymer factor}$ & $\Peierls{eqn: Peierls bound for contour weight}-3e^{1+\nu}(z\ell^2)\eps-\tau-\log3$ & Prop.~\ref{itm: bound on the finite- and infinite-volume contour-polymer factors}  \\

    $\Peierls{eqn: bound on the finite-volume correction to the contour-polymer factor}$ & $\Peierls{eqn: bound on the truncated finite-volume contour-polymer factor}-3e^{1+\nu}(z\ell^2)\eps^2$ & Prop.~\ref{itm: bound on the finite-volume correction to the contour-polymer factor} \\

    $\Peierls{eqn: bound on the first-order derivative of the finite-volume contour-polymer factor}$ & $\Peierls{eqn: bound on the truncated finite-volume contour-polymer factor}/2$ & Prop.~\ref{itm: bound on the first-order derivative of the finite-volume contour-polymer factor} \\
    
    $e^{-\Peierls{eqn: bound on the finite-volume polymer weight}}$ & $e^{-(\Peierls{eqn: bound on the truncated finite-volume contour-polymer factor}-\loss{eqn: bound on the finite-volume interaction-polymer factor})}+e^{-\Peierls{eqn: bound on the finite-volume interaction-polymer factor}}+e^{-(\Peierls{eqn: bound on the truncated finite-volume contour-polymer factor}-\loss{eqn: bound on the finite-volume interaction-polymer factor})}e^{-\Peierls{eqn: bound on the finite-volume interaction-polymer factor}}$ & Prop.~\ref{itm: bound on the finite- and infinite-volume polymer weights}  \\

    $e^{-\Peierls{eqn: bound on the finite-volume correction to the polymer weight}}$ & $\begin{aligned}[t]
        \max\{&e^{-(\Peierls{eqn: bound on the truncated finite-volume contour-polymer factor}-\loss{eqn: bound on the finite-volume correction to the interaction-polymer factors})}+e^{-\Peierls{eqn: bound on the finite-volume correction to the interaction-polymer factors}}+e^{-(\Peierls{eqn: bound on the truncated finite-volume contour-polymer factor}-\loss{eqn: bound on the finite-volume correction to the interaction-polymer factors})}e^{-\Peierls{eqn: bound on the finite-volume correction to the interaction-polymer factors}},\\
        &e^{-(\Peierls{eqn: bound on the finite-volume correction to the contour-polymer factor}-\loss{eqn: bound on the finite-volume interaction-polymer factor})}+e^{-\Peierls{eqn: bound on the finite-volume interaction-polymer factor}}+e^{-(\Peierls{eqn: bound on the finite-volume correction to the contour-polymer factor}-\loss{eqn: bound on the finite-volume interaction-polymer factor})}e^{-\Peierls{eqn: bound on the finite-volume interaction-polymer factor}},
        e^{-\Peierls{eqn: bound on the finite-volume polymer weight}}\} \end{aligned}$ 
    & Prop.~\ref{itm: bound on the finite-volume correction to the polymer weight} \\

    $e^{-\Peierls{eqn: bound on the first-order derivative of the finite-volume polymer weight}}$ & $
    \begin{aligned}[t]
        \max\{&e^{-(\Peierls{eqn: bound on the first-order derivative of the finite-volume contour-polymer factor}-\loss{eqn: bound on the finite-volume interaction-polymer factor})}+e^{-\Peierls{eqn: bound on the finite-volume interaction-polymer factor}}+e^{-(\Peierls{eqn: bound on the first-order derivative of the finite-volume contour-polymer factor}-\loss{eqn: bound on the finite-volume interaction-polymer factor})}e^{-\Peierls{eqn: bound on the finite-volume interaction-polymer factor}},
        \\
        &e^{-(\Peierls{eqn: bound on the truncated finite-volume contour-polymer factor}-\loss{eqn: bound on the finite-volume interaction-polymer factor})}+e^{-\Peierls{eqn: bound on the finite-volume interaction-polymer factor}}+e^{-(\Peierls{eqn: bound on the truncated finite-volume contour-polymer factor}-\loss{eqn: bound on the finite-volume interaction-polymer factor})}e^{-\Peierls{eqn: bound on the finite-volume interaction-polymer factor}}\}\end{aligned}
    $ & Prop.~\ref{itm: bound on the first-order derivative of the finite-volume polymer weight} \\

    $e^{-\Peierls{eqn: bound on the second-order derivative of the finite-volume polymer weight}}$ & $
    \begin{aligned}[t]
        \max\{&e^{-(\Peierls{eqn: bound on the first-order derivative of the finite-volume contour-polymer factor}-\loss{eqn: bound on the finite-volume interaction-polymer factor})}+e^{-\Peierls{eqn: bound on the finite-volume interaction-polymer factor}}+e^{-(\Peierls{eqn: bound on the first-order derivative of the finite-volume contour-polymer factor}-\loss{eqn: bound on the finite-volume interaction-polymer factor})}e^{-\Peierls{eqn: bound on the finite-volume interaction-polymer factor}},
        \\
        &e^{-(\Peierls{eqn: bound on the truncated finite-volume contour-polymer factor}-\loss{eqn: bound on the finite-volume interaction-polymer factor})}+e^{-\Peierls{eqn: bound on the second-order derivative of the finite-volume interaction-polymer factor}}+e^{-(\Peierls{eqn: bound on the truncated finite-volume contour-polymer factor}-\loss{eqn: bound on the finite-volume interaction-polymer factor})}e^{-\Peierls{eqn: bound on the second-order derivative of the finite-volume interaction-polymer factor}}\}
    \end{aligned}
    $ & Prop.~\ref{itm: bound on the second-order derivative of the finite-volume polymer weight} \\
    \midrule

    \multicolumn{3}{l}{\textbf{Spatial convergence rates}} \\
    \midrule
    $e^{-\spatial{eqn: bound on the finite-volume correction to the interaction-polymer factors}}$ & $\eps^{\frac{1}{6}(1-16\free{itm: lower bound on m_Gamma^(Lambda)(delta)-1})}$ & Prop.~\ref{itm: bound on the finite-volume correction to the interaction-polymer factors} \\

    $e^{-\spatial{eqn: bound on the first-order derivative of the finite-volume contour-polymer factor}}$ & $\max\set{e^{-8\Peierls{eqn: bound on the truncated finite-volume contour-polymer factor}},\eps}$ & Prop.~\ref{itm: bound on the first-order derivative of the finite-volume contour-polymer factor} \\

    $e^{-\spatial{eqn: bound on the second-order derivative of the finite-volume contour-polymer factor}}$ & $\max\set{e^{-4\Peierls{eqn: bound on the truncated finite-volume contour-polymer factor}},\eps}$ & Prop.~\ref{itm: bound on the second-order derivative of the finite-volume contour-polymer factor} \\

    $e^{-\spatial{eqn: bound on the finite-volume correction to the polymer weight}}$ & $\max\set{e^{-\spatial{eqn: bound on the finite-volume correction to the interaction-polymer factors}},\eps^{1/3}}$ & Prop.~\ref{itm: bound on the finite-volume correction to the polymer weight} \\

    $e^{-\spatial{eqn: bound on the second-order derivative of the finite-volume polymer weight}}$ & $\max\set{
        e^{-\spatial{eqn: bound on the first-order derivative of the finite-volume contour-polymer factor}},e^{-\spatial{eqn: bound on the second-order derivative of the finite-volume contour-polymer factor}}}$ & Prop.~\ref{itm: bound on the second-order derivative of the finite-volume polymer weight} \\
    \midrule
    
    \multicolumn{3}{l}{\textbf{Rate losses}} \\
    \midrule
    $\loss{eqn: bound on the finite-volume interaction-polymer factor}$ & $162c(\eps;\frac{1}{5}-\free{itm: lower bound on m(delta)-1},\frac{11}{5},2)e^{1+\nu}(z\ell^2)\eps$ & Prop.~\ref{itm: bound on the finite- and infinite-volume interaction-polymer factors}  \\

    $\loss{eqn: bound on the finite-volume correction to the interaction-polymer factors}$ & $243c(\eps;\frac{1}{5}-\free{itm: lower bound on m(delta)-1},\frac{11}{5},2)e^{1+\nu}(z\ell^2)\eps$ & Prop.~\ref{itm: bound on the finite-volume correction to the interaction-polymer factors}  \\

    \midrule

    \multicolumn{3}{l}{\textbf{Prefactors}} \\
    \midrule
    $\prefactor{eqn: bound on the finite-volume correction to the interaction-polymer factors}$ & $81c(\eps;\frac{1}{15}(2+\free{itm: lower bound on m_Gamma^(Lambda)(delta)-1}),\frac{16}{15}(2+\free{itm: lower bound on m_Gamma^(Lambda)(delta)-1}),2)e^{1+\nu}(z\ell^2)\eps$ & Prop.~\ref{itm: bound on the finite-volume correction to the interaction-polymer factors} \\

    $\prefactor{eqn: bound on the first-order derivative of the finite-volume interaction-polymer factor}$ & $9c(\eps;\frac{1}{5}-\free{itm: lower bound on m(delta)-1},\frac{11}{5},2)e^{1+\nu}z\eps$ & Prop.~\ref{itm: bound on the first-order derivative of the finite-volume interaction-polymer factor}  \\

    $\prefactor{eqn: bound on the second-order derivative of the finite-volume interaction-polymer factor}$ & $\begin{aligned}[t]
        &\big[9e^3c(\eps;\tfrac{1}{5}-\free{itm: lower bound on m_12(delta)-2},\tfrac{21}{5},2)+9e^2c(\eps;\tfrac{2}{5}-2\free{itm: lower bound on m(delta)-1},\tfrac{22}{5},4)
        \\
        &+81e^2c(\eps;\tfrac{1}{5}-\free{itm: lower bound on m(delta)-1},\tfrac{11}{5},2)^2\big]e^{2\nu}z^2\eps^2
    \end{aligned}$ & Prop.~\ref{itm: bound on the second-order derivative of the finite-volume interaction-polymer factor}  \\

    $\prefactor{eqn: bound on the first-order derivative of the contour weight - localization}$ & $2\prefactor{eqn: inductive bound on the first-order derivative of the log polymer partition function}+2(1+e\eps)e^{\nu/n}z$ & Lem.~\ref{lem: bound on the first-order derivative of the contour weight} \\

    $\prefactor{eqn: bound on the first-order derivative of the contour weight - decay}$ & $e^{1+\nu/n}z$ & Lem.~\ref{lem: bound on the first-order derivative of the contour weight} \\

    $\prefactor{eqn: bound on the second-order derivative of the contour weight}$ & $2\prefactor{eqn: inductive bound on the second-order derivative of the log polymer partition function}+3e^{3+2\nu/n}z^2$ & Lem.~\ref{lem: bound on the second-order derivative of the contour weight} \\

    $\prefactor{eqn: bound on the finite-volume correction to the contour-polymer factor}$ & $3e^{1+\nu}(z\ell^2)$ & Prop.~\ref{itm: bound on the finite-volume correction to the contour-polymer factor} \\

    $\prefactor{eqn: bound on the first-order derivative of the finite-volume contour-polymer factor}$ & $z+\prefactor{eqn: bound on the first-order derivative of the contour weight - localization}+\prefactor{eqn: bound on the first-order derivative of the contour weight - decay}$ & Prop.~\ref{itm: bound on the first-order derivative of the finite-volume contour-polymer factor} \\

    $\prefactor{eqn: bound on the second-order derivative of the finite-volume contour-polymer factor}$ & $\prefactor{eqn: bound on the second-order derivative of the contour weight}+(z+\prefactor{eqn: bound on the first-order derivative of the contour weight - localization}+\prefactor{eqn: bound on the first-order derivative of the contour weight - decay})^2$ & Prop.~\ref{itm: bound on the second-order derivative of the finite-volume contour-polymer factor} \\

    $\prefactor{eqn: bound on the finite-volume correction to the polymer weight}$ & $\prefactor{eqn: bound on the finite-volume correction to the interaction-polymer factors}e^{-(\Peierls{eqn: bound on the truncated finite-volume contour-polymer factor}-\loss{eqn: bound on the finite-volume correction to the interaction-polymer factors})}(1+e^{-\Peierls{eqn: bound on the finite-volume correction to the interaction-polymer factors}})
    +\prefactor{eqn: bound on the finite-volume correction to the contour-polymer factor}e^{-(\Peierls{eqn: bound on the finite-volume correction to the contour-polymer factor}-\loss{eqn: bound on the finite-volume interaction-polymer factor})}(1+e^{-\Peierls{eqn: bound on the finite-volume interaction-polymer factor}})+e^{-\Peierls{eqn: bound on the finite-volume polymer weight}}$ & Prop.~\ref{itm: bound on the finite-volume correction to the polymer weight} \\

    $\prefactor{eqn: bound on the first-order derivative of the finite-volume polymer weight}$ & $\prefactor{eqn: bound on the first-order derivative of the finite-volume contour-polymer factor}e^{-(\Peierls{eqn: bound on the first-order derivative of the finite-volume contour-polymer factor}-\loss{eqn: bound on the finite-volume interaction-polymer factor})}(1+e^{-\Peierls{eqn: bound on the finite-volume interaction-polymer factor}})+\prefactor{eqn: bound on the first-order derivative of the finite-volume interaction-polymer factor}e^{-(\Peierls{eqn: bound on the truncated finite-volume contour-polymer factor}-\loss{eqn: bound on the finite-volume interaction-polymer factor})}(1+e^{-\Peierls{eqn: bound on the finite-volume interaction-polymer factor}})$ & Prop.~\ref{itm: bound on the first-order derivative of the finite-volume polymer weight} \\

    $\prefactor{eqn: bound on the second-order derivative of the finite-volume polymer weight}$ & $
    \begin{aligned}[t]
        &\prefactor{eqn: bound on the second-order derivative of the finite-volume contour-polymer factor}e^{-(\Peierls{eqn: bound on the first-order derivative of the finite-volume contour-polymer factor}-\loss{eqn: bound on the finite-volume interaction-polymer factor})}(1+e^{-\Peierls{eqn: bound on the finite-volume interaction-polymer factor}})
        +2\prefactor{eqn: bound on the first-order derivative of the finite-volume interaction-polymer factor}\prefactor{eqn: bound on the first-order derivative of the finite-volume contour-polymer factor}e^{-(\Peierls{eqn: bound on the first-order derivative of the finite-volume contour-polymer factor}-\loss{eqn: bound on the finite-volume interaction-polymer factor})}(1+e^{-\Peierls{eqn: bound on the finite-volume interaction-polymer factor}})
        \\
        &+\prefactor{eqn: bound on the second-order derivative of the finite-volume interaction-polymer factor}e^{-(\Peierls{eqn: bound on the truncated finite-volume contour-polymer factor}-\loss{eqn: bound on the finite-volume interaction-polymer factor})}(1+e^{-\Peierls{eqn: bound on the second-order derivative of the finite-volume interaction-polymer factor}})
    \end{aligned}
    $ & Prop.~\ref{itm: bound on the second-order derivative of the finite-volume polymer weight} \\
\end{longtable}

\subsection{Convergence and closure conditions}
\label{appx: convergence and closure conditions}

The convergence and closure conditions are as follows.

\setcounter{condrow}{0}
\begin{longtable}{>{\raggedright\arraybackslash}p{0.04\columnwidth} >{\raggedright\arraybackslash}p{0.27\columnwidth} >{\raggedright\arraybackslash}p{0.12\columnwidth} >{\raggedright\arraybackslash}p{0.46\columnwidth}}
    \toprule
     & \textbf{Condition(s)} & \textbf{Location} & \textbf{Description} \\
    \midrule
    \endfirsthead
    
    \toprule
     & \textbf{Condition(s)} & \textbf{Location} & \textbf{Description} \\
    \midrule
    \endhead
    
    \midrule
    \endfoot
    
    \bottomrule
    \endlastfoot
    
    \condrowlabel{1} & \eqref{eqn: convergence condition for expansion of q-oriented partition function} & Sec.~\ref{sec: cluster expansion of the q-oriented partition function} & Closure; Ueltschi's criterion~\cite{ueltschi2004cluster} for the convergence of the cluster expansion applied to the $q$-oriented hard rectangle model \\

    \condrowlabel{2} & $2\eps^{\frac{1}{5}-\free{itm: lower bound on m(delta)-1}}<\frac{1}{64}$ & Prop.~\ref{prop: bounds on the interaction-polymer factors} & Convergence; to apply Lem.~\ref{itm: incompatible polymer sum with piecewise linear decay} \\

    \condrowlabel{3} & $2\eps^{\frac{1}{15}(2+\free{itm: lower bound on m_Gamma^(Lambda)(delta)-1})}<\frac{1}{64}$ & Prop.~\ref{itm: bound on the finite-volume correction to the interaction-polymer factors} & Convergence; to apply Lem.~\ref{itm: incompatible polymer sum with piecewise linear decay} \\

    \condrowlabel{4} & $2\eps^{\frac{1}{5}-\free{itm: lower bound on m_12(delta)-2}}<\frac{1}{64}$ & Prop.~\ref{itm: bound on the second-order derivative of the finite-volume interaction-polymer factor} & Convergence; to apply Lem.~\ref{itm: incompatible polymer sum with piecewise linear decay} \\

    \condrowlabel{5} & $2^2\eps^{\frac{2}{5}-2\free{itm: lower bound on m(delta)-1}}<\frac{1}{64}$ & Prop.~\ref{itm: bound on the second-order derivative of the finite-volume interaction-polymer factor} & Convergence; to apply Lem.~\ref{itm: incompatible polymer sum with piecewise linear decay} \\

    \condrowlabel{6} & $\Peierls{eqn: Peierls bound for contour weight}>0$ & Lem.~\ref{lem: Peierls bound for contour weight} & \\

    \condrowlabel{7} & $e^{-\Peierls{eqn: bound on the second-order derivative of the finite-volume polymer weight}}\le 1$ & Prop.~\ref{itm: bound on the second-order derivative of the finite-volume polymer weight} & Convergence; to apply Lem.~\ref{itm: covering sum second moment} \\

    \condrowlabel{8} & \eqref{eqn: condition on the Peierls constant for polymer weights} & Lem.~\ref{lem: the truncated polymer model can be cluster expanded} & Closure; Ueltschi's criterion applied to the truncated polymer model \\

    \condrowlabel{9} & $e^{-(\Peierls{eqn: bound on the finite-volume polymer weight}-2\alpha)}<\frac{1}{64}$ & Lem.~\ref{lem: the truncated polymer model can be cluster expanded} & Convergence; to apply Lem.~\ref{itm: incompatible polymer sum} \\

    \condrowlabel{10} & \eqref{eqn: closure condition for the base case} & Prop.~\ref{prop: main} & Closure; base case for~\eqref{eqn: inductive bound on the log flipping term} \\

    \condrowlabel{11} & $e^{-(\Peierls{eqn: bound on the first-order derivative of the finite-volume polymer weight}-\alpha)}<\frac{1}{64}$, $e^{-\spatial{eqn: bound on the first-order derivative of the finite-volume contour-polymer factor}}<1$ & Prop.~\ref{prop: main} & Convergence; to apply Lem.~\ref{itm: polymer sum with distance decay from one point} in~\eqref{eqn: bound on the first-order derivative of the log polymer partition function} \\

    \condrowlabel{12} & \eqref{eqn: closure condition for the bound on the first-order derivative of the log polymer partition function} & Prop.~\ref{prop: main} & Closure; inductive step for~\eqref{eqn: inductive bound on the first-order derivative of the log polymer partition function} \\

    \condrowlabel{13} & $e^{-(\Peierls{eqn: bound on the second-order derivative of the finite-volume polymer weight}-\alpha)}<\frac{1}{64}$, $e^{-\spatial{eqn: bound on the second-order derivative of the finite-volume polymer weight}}<1$ & Prop.~\ref{prop: main} & Convergence; to apply Lem.~\ref{itm: polymer sum with combined distance decay from two points} in~\eqref{eqn: bound on the second-order derivative of the log polymer partition function} \\

    \condrowlabel{14} & $e^{-\free{itm: polymer sum with separate distance decay from two points}(\Peierls{eqn: bound on the second-order derivative of the finite-volume polymer weight}-\alpha)}<\frac{1}{64}$, $e^{-\spatial{eqn: bound on the second-order derivative of the finite-volume polymer weight}}<1$ & Prop.~\ref{prop: main} & Convergence; to apply Lem.~\ref{itm: polymer sum with separate distance decay from two points} in~\eqref{eqn: bound on the second-order derivative of the log polymer partition function} \\

    \condrowlabel{15} & $e^{-\free{itm: polymer sum with separate distance decay from two points}(2\Peierls{eqn: bound on the first-order derivative of the finite-volume polymer weight}-3\alpha)}<\frac{1}{64}$, $e^{-\spatial{eqn: bound on the first-order derivative of the finite-volume contour-polymer factor}}<1$ & Prop.~\ref{prop: main} & Convergence; to apply Lem.~\ref{itm: polymer sum with separate distance decay from two points} in~\eqref{eqn: bound on the second-order derivative of the log polymer partition function} \\

    \condrowlabel{16} & $e^{-\free{itm: double polymer sum with distance decay from two points}(\Peierls{eqn: bound on the first-order derivative of the finite-volume polymer weight}-\frac{3}{2}\alpha)}<\frac{1}{64}$, $e^{-\spatial{eqn: bound on the first-order derivative of the finite-volume contour-polymer factor}}<1$ & Prop.~\ref{prop: main} & Convergence; to apply Lem.~\ref{itm: double polymer sum with distance decay from two points} in~\eqref{eqn: bound on the second-order derivative of the log polymer partition function} \\

    \condrowlabel{17} & \eqref{eqn: closure condition for the prefactor in the bound on the second-order derivative of the log polymer partition function},~\eqref{eqn: closure condition for the spatial decay rate in the bound on the second-order derivative of the log polymer partition function} & Prop.~\ref{prop: main} & Closure; inductive step for~\eqref{eqn: inductive bound on the second-order derivative of the log polymer partition function} \\

    \condrowlabel{18} & $e^{-(\Peierls{eqn: bound on the finite-volume polymer weight}-\alpha)}<\frac{1}{64}$ & Prop.~\ref{prop: main} & Convergence; to apply Lem.~\ref{itm: anchored polymer sum} in~\eqref{eqn: boundary error after bulk cancellation} \\

    \condrowlabel{19} & $e^{-(\Peierls{eqn: bound on the finite-volume correction to the polymer weight}-\alpha)}<\frac{1}{64}$, $e^{-\spatial{eqn: bound on the finite-volume correction to the polymer weight}}<1$ & Prop.~\ref{prop: main} & Convergence; to apply Lem.~\ref{itm: polymer sum with distance decay from region} in~\eqref{eqn: boundary error from finite-volume corrections to the bulk terms} \\

    \condrowlabel{20} & \eqref{eqn: closure condition for the bound on the log flipping term} & Prop.~\ref{prop: main} & Closure; inductive step for~\eqref{eqn: inductive bound on the log flipping term} \\
\end{longtable}

\subsection{Summary of the optimization program}

The problem of deriving a numerical upper bound for the aspect ratio required for the formation of a nematic phase in a hard rectangle model can thus be reduced to the following optimization program.
Given $\fw$, determine the least $\fl\ge\max\set{3,2\fw-1}$ for which there exists an admissible choice of the remaining fundamental parameters, in the domains specified in Appx.~\ref{appx: fundamental parameters} and with all derived quantities defined as in Appx.~\ref{appx: derived quantities}, such that all the convergence and closure conditions in Appx.~\ref{appx: convergence and closure conditions} are satisfied.
In addition, the solution $\fl_0=\fl_0(\fw)$ should enjoy a monotonicity property, namely, every $\fl\ge\fl_0$ also lies in the projection of the feasible region (with the specified value of $\fw$) onto the $\fl$-axis.

\end{document}